\documentclass{article}

\usepackage{vmargin}
\usepackage{url,graphicx}
\graphicspath{{./figs/}}

\usepackage{hyperref}

\usepackage{etex}      
\usepackage{amsmath}		
\usepackage{amssymb}		
\usepackage{cmll}
\usepackage{stmaryrd}
\usepackage{proof}
\usepackage{listings}
\usepackage[all]{xy}

\usepackage{color}		

 \usepackage{amsthm}		
 \theoremstyle{definition}
 \newtheorem{definition}{Definition}[section]

 \theoremstyle{plain}
 \newtheorem{proposition}[definition]{Proposition}
 \newtheorem{lemma}[definition]{Lemma}
 \newtheorem{theorem}[definition]{Theorem}
 \newtheorem{corollary}[definition]{Corollary}
 \newtheorem{conjecture}[definition]{Conjecture}
 \newtheorem*{notation*}{Notation}

\newcommand{\prim}[1]{\mathsf{#1}}
\newcommand{\abs}[3]{\lambda {#1}^{#2}. #3}
\newcommand{\app}[2]{#1 \, #2}	      
\newcommand{\letin}[3]{\prim{let}\ {#1 = #2}\ \prim{in}\ #3}

\newcommand{\val}[1]{\underline{#1}}

\newcommand{\plug}[2]{#1 \langle #2 \rangle}
\newcommand{\emptyCtxt}{\langle \cdot \rangle}
\newcommand{\param}[1]{ \{#1\}  }
\newcommand{\abd}[5]{\prim{A}^{#4}_{#1}({#2},{#3}).{#5}}

\newcommand{\twoheadmapsto}{
\mathrel{\ooalign{$\twoheadrightarrow$\cr%
\kern-.15ex\raisebox{.2ex}{\scalebox{1}[0.8]{$\shortmid$}}\cr}}}
\newcommand{\longtwoheadmapsto}{
\mathrel{\ooalign{$\longtwoheadrightarrow$\cr%
\kern-.15ex\raisebox{.2ex}{\scalebox{1}[0.8]{$\shortmid$}}\cr}}}

\newcommand{\finsubset}{\subset_\mathrm{fin}}
\newcommand{\cl}{:}

\newcommand{\F}{\mathbb{F}}
\newcommand{\A}{\mathbb{A}}
\newcommand{\N}{\mathbb{N}}
\newcommand{\iter}{\mathrm{iter}}
\newcommand{\up}{\mathord{\uparrow}}
\newcommand{\dn}{\mathord{\downarrow}}
\newcommand{\lb}[1]{{#1}}
\newcommand{\rot}[1]{\rotatebox[origin=c]{180}{$#1$}}
\newcommand{\Qry}{\mathrm{Qry}}
\newcommand{\Ans}{\mathrm{Ans}}
\newcommand{\In}{\mathit{in}}
\newcommand{\Out}{\mathit{out}}

\newcommand{\Init}{\mathit{Init}}
\newcommand{\Final}{\mathit{Final}}
\renewcommand{\ell}{e}

\newcommand{\lift}[1]{\mathrel{\mathcal{L}_{#1}}}
\newcommand{\liftRT}[1]{\mathrel{\mathcal{L}^*_{#1}}}

\usepackage{booktabs}   
\usepackage{subcaption} 

\begin{document}

\title{Abductive functional programming, a semantic approach}         





\author{Koko Muroya \\ University of Birmingham
	\and Steven Cheung \\ University of Birmingham 
 \and Dan R. Ghica \\ University of Birmingham}
\date{}
\maketitle

\begin{abstract}
We propose a call-by-value lambda calculus extended with a new construct inspired by abductive inference and motivated by the programming idioms of machine learning. Although syntactically simple the abductive construct has a complex and subtle operational semantics which we express using a style based on the Geometry of Interaction. We show that the calculus is sound, in the sense that well typed programs terminate normally. We also give a visual implementation of the semantics which relies on additional garbage collection rules, which we also prove sound. 
\end{abstract}



\section{Introduction}

In machine learning it is common to look at programs as \textit{models} which are used in two \textit{modes}. The first mode, which we shall call the `direct mode', is the usual operating behaviour of any garden variety program in which new inputs are applied and new outputs are obtained. The second mode, which we shall call `learning mode', is what makes machine learning special. In learning mode, special inputs are applied, for which the desired outputs (or at least some fitness criteria for output) are known. Then parameters of the model are changed, or tuned, so that the actual outputs approach, in a measurable way, the desired outputs (or the fitness function is improved). 

Examples of models vary from the simplest, such as linear regression, with only two parameters ($f(x)=p_1\times x+p_2$) to the most complex, such as recurrent neural nets, with many thousands of various parameters defining signal aggregation and the shape of activation functions. What makes machine learning programming interesting and, in some sense, tractable is that the model and the algorithm for tuning the model can be decoupled. The tuning, or optimisation, algorithms, such as gradient descent or simulated annealing, can be abstracted from the model and programmed separately and generically. It is the interaction between the model and the tuning algorithm that enables machine learning programming. 

In this paper we introduce a programming language in which this bi-modal programming idiom is built-in. Our ultimate aim is an ergonomic and efficient functional language which obeys the general methodological principles of information encapsulation as it pertains to the specific programming of machine learning. We propose that this should be achieved by starting from the basis of an applied lambda calculus, then equipping it with a dedicated operation for `parameter extraction' which, given a term (\textit{qua} model in direct mode) produces a new, parameterised model (\textit{qua} model in learning mode). Unlike the direct-mode model, which is a function of inputs, the learning-mode model becomes a function of its parameters and its inputs and, as such, can be used in a tuning algorithm to evaluate how different values of the parameters impact the fitness of the output for given inputs. 

Concretely, lets consider the simple example of linear regression written as a function: $f\,x=p_1*x+p_2$, where the parameters have provisional values $p_1,p_2$. In learning mode, the model becomes $f\,p_1\,p_2\,x=p_1*x+p_2$, and a new direct-mode model with updated parameters $p_1', p_2'$ can be immediately reconstructed as $f\,p_1'\,p_2'$. Various \textit{ad hoc} mechanisms for switching between the two modes can be, of course, explicitly implemented using existing programming language mechanisms. However, providing a native and seamless syntactic mechanism programming this scenario can be significantly simplified.  

A solution that comes close to this ideal is  employed by \textsc{TensorFlow}, in which a separate syntactic class of \textit{variables} serves precisely the role of parameters as discussed above (in its \textsc{Python} bindings)~\cite{abadi2016tensorflow}. However, \textsc{TensorFlow} is presented as a shallow embedding of a domain specific language (DSL) into \textsc{Python}. Moreover, the DSL offers explict constructs for switching between `direct' and `learning' modes under the notion of \textit{session}. \textsc{TensorFlow} is an incredibly useful and well crafted library and associated DSL which gave us much inspiration. Our aim is to extract the essence of this approach and encapsulate it in a stand-alone programming language, rather than an embedded DSL. This way we can hope to eventually develop a genuine programming language for machine learning, avoiding the standard pitfalls of embedded DSLs, such as difficulty of reasoning about code, poor interaction with the rest of the language, especially via libraries, lack of proper type-checking, difficult debugging and so on~\cite{ghosh2011dsl}. 

\subsection{Abductive decoupling of parameters}

We propose a new framework for extracting parameters from models via what we will call \textit{abductive decoupling}. The name is inspired by ``abductive inference'', and the connection is explained informally in Sec.~\ref{sec:abdf}. We use decoupling as the preferred mechanism for extracting parameters from models. Looking at the type system from an abductive perspective, the essence of decoupling is the following admissible rule:
\begin{center}
$\infer[P=P_1\land\cdots\land P_n]
{\Gamma  ,P_1,\ldots,P_n\vdash (P\rightarrow A) \land P}
{\Gamma ,P_1,\ldots,P_n\vdash A}$
\end{center}
Informally, the rule selects $P=P_1\land\cdots\land P_n$ as the ``best explanation'' of $A$ from all  premises and infers  an ``abductive summary'' consisting of  this explanation along with the fact that it implies~$A$. This rule is interesting only in regard to the process of selecting explanatory premises $P_i$. It can be trivialised by selecting either the whole set of premises or none ($P=true$). We note that this is a sound rule \textit{inspired} by abductive inference, which is generally unsound, much like (sound) mathematical induction is inspired by (unsound) logical induction. The source of this inspiration is discussed in the following Section. 

In a programming language with product and function types this rule can be made to correspond to a special family of constants $\prim{A}_A:A\rightarrow(V\rightarrow A)\times V$, where the type $V$ is a type of \textit{collections of parameters} and $A$ a data type. The constant would abductively decouple a (possibly open) term $t$ into the parametrised term and the current parameter values. In a simpler language with no product types, the rule for abductive decoupling is given implicationally as:
\begin{center}
$\infer
{\Gamma\vdash \prim{A}_{T'}(f,p).t:T'\rightarrow T}
{\Gamma,f:V\rightarrow T',p:V\vdash t:T}.
$
\end{center}
In an application $(\prim{A}(f,p).t)t'$ the term $t'$ is abductively decoupled into parameters, bound to $p$, and a parameterised model, bound to $f$. They are then used in $t$, typically for parameter tuning. This is a common pattern, for which we use syntactic sugar: 
\[\letin{f\,@\,p}{t'}{t}.\]

Abductive decoupling should apply only to selected constants in a model because tuning all constants of a model is not generally desirable. This is achieved in the concrete syntax by marking \textit{provisional constants} in models with braces, e.g. $\{7\}$. In direct mode provisional constants are used simply as constants, whereas in learning mode they are targeted by abduction. For example, the abductive decoupling of the term $\{1\}+2$ results in the parameterised model $\lambda p.p[0]+2$ and the singleton vector parameter $p=[1]$. 

\subsection{Informal semantics of abductive decoupling}
Behind this simple new syntax lurks a complex and subtle semantics. Abductive decoupling is no mere superficial syntactic refactoring of provisional constants into arguments, but a \textit{deep} runtime operation which can target provisional constants from free variables of a term. Consider for example:
\begin{align*}
&\letin {y} {\{2\} + 1} { }\\[-1ex]
&\letin {m\, x} {\{3\} + y + x}{ }\\[-1ex]
&\letin{f\,@\,p}m{ }\\[-1ex]
&f\, p\, 7.
\end{align*}
The model $m$ depends directly on a provisional constant ($\{3\}$) but also, indirectly, on a term ($y$) which itself depends on a provisional constant ($\{2\}$). It should be apparent that a syntactic resolution of abductive decoupling is not possible. When abduction occurs, the term $m$ will have already been reduced to $m\,x=\{3\} + (\{2\} + 1) + x$, so following abduction the parameterised model is similar to $\lambda p\lambda x.p[0]+(p[1]+1)+x$. 

On the other hand, the semantics of reduction needs to be appropriately adapted to the presence of provisional constants so that they are not reduced away during computation. In order to predictably employ tuning of the parameters, the identity of the parameters must be preserved during evaluation. Thus, $\{1\}+\{2\}$ should not be reduced to $3$, either as a provisional or as a definitive constant. Also $(\lambda x.x+x)\{1\}$ should be computed in a way that uses only one tunable parameter, rather than creating two via copying. This simple example also indicates that in the process of reduction terms may evolve into forms that are not necessarily syntactically expressible.

A more formal justification for preserving the number of provisional constants during evaluation is the obvious need for a program (or a representation thereof), as it evolves during evaluation to remain observationally equivalent to its previous forms. However, since provisional constants can be detected by abduction, changing the number of provisional constants would be observable by abductive contexts. 

Our semantic challenge is reconciling this behaviour within a conventional call-by-value reduction framework. A handy tool in specifying the operational semantics of abduction is the Geometry of Interaction (GoI)~\cite{Girard89GoI1,Girard90GoI2}. Intended as an operational interpretation of linear logic proofs, the GoI proved to be a useful syntax-independent operational framework for programming languages as well~\cite{Mackie95}. A GoI interpretation maps a program into a network of simple transducers, which executes by passing a token along its edges and processing it in the nodes. This interpretation is naturally suited for call-by-name evaluation, which it can perform on a fixed net. This constant space execution made it possible to compile CBN-based languages such as \textsc{Algol} directly into circuits~\cite{Ghica07GoS1}.
Using GoI as a model for call-by-value in a way that preserves both the equational theory and the cost model was an open problem, solved only recently by a combination of token-passing and graph-rewriting~\cite{DBLP:conf/csl/MuroyaG17} called ``the dynamic GoI''. This is precisely the semantic framework in which abduction will be interpreted. 

\subsection{Contributions}
We introduce a new functional programming construct which we call \textit{abductive decoupling}, which allows \textit{provisional constants} to be automatically extracted from terms. This new construct is motivated by programming idioms and patterns occurring primarily in machine learning. Although this mechamism is expressed in a language via a simple syntactic construct, the semantics is subtle and complex. We specify it using a recently developed ``dynamic'' Geometry of Interaction style and we show the soundness of execution (i.e. the successful termination of any well-typed program) and of garbage-collection rules (i.e. that they have no effect on observable behaviour). To support a better understanding of the semantics of abductive decoupling we also implement an on-line visualiser for execution\footnote{Link to on-line visualiser: \texttt{\url{http://www.cs.bham.ac.uk/~drg/goa/visualiser/index.html}}}.

%
\newpage 
\section{Abductive functional programming: a new paradigm}\label{sec:abdf}

\subsection{Deduction, induction, abduction}
\begin{quotation}
\textit{The division of all inference into Abduction, Deduction, and Induction may almost be said to be the Key of Logic. }

\hfill 
C.S.Peirce 
\end{quotation}

C.S. Peirce, in his celebrated \textit{Illustrations of the Logic of Science}, introduced three kinds of reasoning: deductive, inductive, and abductive. Deduction and induction are widely used in mathematics and computer science, and they have been thoroughly studied by philosophers of science and knowledge. Abduction, on the other hand, is more mysterious. Even the name ``abduction'' is controversial. Peirce claims that the word is a mis-translation of a corrupted text by Aristotle (``$\mathit{\alpha\pi\alpha\gamma\omega\gamma\acute{\eta}}$''), and sometimes used ``\textit{retroduction}'' or ``\textit{hypothesis}'' to refer to it. But the name ``\textit{abduction}'' seems to be the most common, so we will use it. 

\newcommand{\aline}{------------------\\}

According to Peirce the essence of deduction is the syllogism known as ``\textit{Barbara}'':
\begin{quote}
\textit{Rule}: All men are mortal. \\
\textit{Case}: Socrates is a man. \\
\aline 
\textit{Result}: Socrates is a mortal.
\end{quote}

Peirce calls all deduction \textit{analytic} reasoning, the application of general rules to particular cases. Deduction always results in apodeictic knowledge, incontrovertible knowledge you can believe as strongly as you believe the premises. Peirce's interesting formal experiment was to then permute \textit{the Rule}, \textit{the Case}, and \textit{the Result} from this syllogism, resulting in two new patterns of inference which, he claims, play a key role in the logic of scientific discovery. The first one is \textit{induction}:
\begin{quote}
\textit{Case}: Socrates is a man.\\
\textit{Result}: Socrates is a mortal.\\
\aline  
\textit{Rule}: All men are mortal.
\end{quote}
Here, from the specific we infer the general. Of course, as stated above the generalisation seems hasty, as only one specific case-study is generalised into a rule. But consider
\begin{quote}
\textit{Case}: Socrates and Plato and Aristotle and Thales and Solon are men. \\
\textit{Result}: Socrates and Plato and Aristotle and Thales and Solon mortal.\\
\aline
\textit{Rule}: All men are mortal.
\end{quote} 
The Case and Result could be extended to a list of billions, which would be quite convincing as an inductive argument. However, no matter how extensive the evidence, induction always involves a loss of certainty. According to Peirce, induction is an example of a \textit{synthetic} and \textit{ampliative} rule which generates new but uncertain knowledge. 
If a deduction can be believed, an inductively derived rule can only be \textit{presumed}. 

The other permutation of the statements is the rule of abductive inference or, has Peirce originally called it, ``\textit{hypothesis}'':
\begin{quote}
\textit{Result}: Socrates is a mortal.\\
\textit{Rule}: All men are mortal.\\
\aline
\textit{Case}: Socrates is a man.
\end{quote} 
This seems prima facie unsound and, indeed, Peirce acknowledges abduction as the weakest form of (synthetic) inference, and he gives a more convincing instance of abduction in a different example:
\begin{quote}
\textit{Result}: Fossils of fish are found inland.\\
\textit{Rule}: Fish live in the sea.\\
\aline
\textit{Case}: The inland used to be covered by the sea.
\end{quote} 
We can see that in the case of abduction the inference is clearly ampliative and the resulting knowledge has a serious question mark next to it. It is unwise to believe it, but we can \textit{surmise} it. This is the word Peirce uses to describe the correct epistemological attitude regarding abductive inference. Unlike analytic inference, where conclusions can be believed a priori, synthetic inference gives us conclusions that can only be believed a posteriori, and even then always tentatively. This is why experiments play such a key role in science. They are the analytic test of a synthetic statement. 

But the philosophical importance of abduction is greater still. Consider the following instance of abductive reasoning:
\begin{quote}
\textit{Result}: The thermometer reads 20C.\\
\textit{Rule}: If the temperature is 20C then the thermometer reads 20C.\\
\aline
\textit{Case}: The temperature is 20C.
\end{quote} 
Peirce's philosophy was directly opposed to Descartes's extreme scepticism, and abductive reasoning is really the only way out of the quagmire of Cartesian doubt. We can never be totally sure whether the thermometer is working properly. Any instance of trusting our senses or instruments is an instance of abductive reasoning, and this is why we can only generally surmise the reality behind our perceptions. Whereas Descartes was paralysed by the fact that believing our senses can be questioned, Peirce just took it for what it was and moved on. 

\subsection{A computational interpretation of abduction: machine learning}
Formally, the three rules of inference could be written as:
\[
\infer[\text{Deduction}] B {A & A\rightarrow B}\qquad
\infer[\text{Induction}] {A\rightarrow B}{A & B}\qquad 
\infer[\text{Abduction}] A {B & A\rightarrow B}
\]
Using the Curry-Howard correspondence as a language design guide, we will arrive at some programming language constructs corresponding to these rules. Deduction corresponds to producing $B$-data from $A$-data using a function $A\rightarrow B$. Induction would correspond to creating a $A\rightarrow B$ function when we have some $A$-data and some $B$-data. And indeed, computationally we can (subject to some assumptions we will not dwell on in this informal discussion) create a \textit{look-up table} from $A$s to the $B$s, which maybe will produce some default or approximate or interpolated/extrapolated value(s) when some new $A$-data is input. The process is clearly both ampliative, as new knowledge is created in the form of new input-output mappings, and tentative as those mappings may or may not be correct.

Abduction by contrast assumes the existence of some facts $B$ and a mechanism of producing these facts $A\rightarrow B$. As far as we are aware there is no established consensus as to what the $A$s represent, so we make a proposal: the $A$s are the \textit{parameters} of the mechanism $A\rightarrow B$ of producing $B$s, and \textit{abduction} is a general process of choosing the ``best'' $A$s to justify some given $B$s. This is a machine-learning situation. Abduction has been often considered as ``inference to the best explanation'', and our interpretation is consistent with this view if we consider the values of the parameters as the ``explanation'' of the facts. 

Let us consider a simple example written in a generic functional syntax where the model is a linear map with parameters $a$ and $b$. Compared to the rule above, the parameters are $A=float\times float$ and the ``facts'' are a model $B=float\rightarrow float$: 
\begin{align*}
& f : (float \times float) \rightarrow (float \rightarrow float)\\
& f\, (a,b)\, x = a * x + b
\end{align*}
A set of reference facts can be given as a look-up table $data : B$. The machine-learning situation involves the production of an ``optimal'' set of parameters, relative to a pre-determined error (e.g. least-squares) and using a generic optimisation algorithm (e.g. gradient descent):
\begin{align*}
&(a', b') = \mathbf{abduct}\,f\,data\\
&f' = f\,(a',b')
\end{align*}
Note that a concrete, optimal model $f'$ can be now synthesised deductively from the parametrised model $f$ and experimentally tested for accuracy. Since the optimisation algorithm is generic any model can be tuned via abduction, from the simplest (linear regression, as above) to models with thousands of parameters as found in a neural network. 

\subsection{From abductive programming to programmable abduction}\label{sec:abd}

Since abduction can be carried out using generic search or optimisation algorithms, having a fixed built-in such procedure can be rather inconvenient and restrictive. \textsc{Prolog} is an example of a language in which an abduction-like algorithm is fixed. The idea is to make abduction itself programmable. 

In a simple program like the one above abduction coincides with optimisation, which can be programmed (e.g. gradient descent) relative to a specified loss function (e.g. least squares):
\begin{align*}
& f : (\mathit{float} \times \mathit{float}) \rightarrow (\mathit{float} \rightarrow \mathit{float})\\
& f\, (a,b)\, x = a * x + b\\
&(a', b') = \mathit{gradient\_descent}\, f\, \mathit{least\_squares}\, \textit{data}\\
&f' = f\,(a',b')
\end{align*}
Algorithmically this works, but from the point of view of programming language design this is not entirely satisfactory because the type of the \textit{gradient\_descent} function must have a return type $\mathit{float}\times \mathit{float}$ which is not the type that a generic gradient descent function would return. That should be a \textit{vector}. One can program models where the arguments are always vectors
\begin{align*}
& f : vec \rightarrow (\mathit{float} \rightarrow \mathit{float})\\
& f\, p\, x = p[0] * x + p[1]\\
& p' = \mathit{gradient\_descent}\, f\, \mathit{least\_squares}\, \textit{data}\\
& f' = f\, p'
\end{align*}
 But this style of programming becomes increasingly awkward as models become more complicated, as we shall see later. Consider for example a model which is a surface bounded by two parametrised curves:
\begin{align*}
& model\, low\, high\, x = (low\,x, high\,x)
\end{align*}
and given some \textit{data} as a collection of points defined by their $(x,y)$ coordinates is trying to find the best function boundaries such that a measure few points fall outside yet the bounds are tight:
\begin{center}
\includegraphics[]{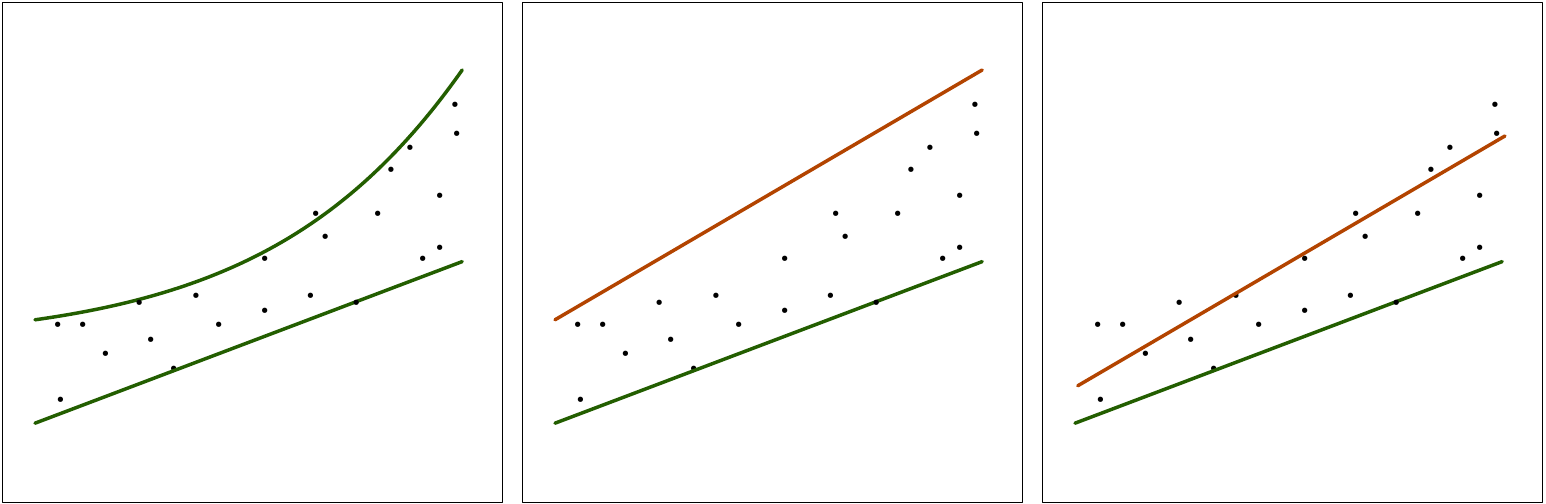}
\end{center}
The first figure shows a good fit by a quadratic and by a linear boundary. The second and third figures show an attempt to fit an upper linear boundary which is either too loose or too tight. 

The candidate parametrised boundaries are:
\begin{align*}
& linr\, (a, b)\, x = a * x + b \\
& quad\, (a, b, c)\, x = a * x * x + b * x + c 
\end{align*} 
and the loss (error) function could be defined as:
\begin{align*}
loss\,(min,max)\,x=
(\mathrm{if}\, x > min\, \mathrm{then}\, 0\, \mathrm{else}\, (x-min))+
(\mathrm{if}\, x < max\, \mathrm{then}\, 0\, \mathrm{else}\, (max-x))
\end{align*}
We would like to use gradient descent, or some other generic optimisation or search algorithms to try out various models such as 
\begin{align*}
&p = \mathit{gradient\_descent}\,(model\,linr\,linr)\,loss\,data\\
&p = \mathit{gradient\_descent}\,(model\,quad\,linr)\,loss\,data
\end{align*}
and collect improved parameters $p$, but it is now not so obvious how to collect the parameters of \textit{linr} and \textit{quad} in a uniform way, optimise them, then plug them back into the model, i.e. into the boundary functions as they are used. 

Our thesis is that this programming pattern, where models are complex and parameters which require optimisation are scattered throughout is a common one in machine learning and optimisation applications. What we propose is a general programming language mechanism, inspired by abduction, which will collect all the parameters of a complex model and actually parametrise the model by them, at run-time. We call this feature \textit{decoupling}, and the example above would be programmed as follows:
\begin{align*}
&m_0=model\, (quad (0,0,0))\, (linr (0,0))\\
&(m, p) = \textbf{decouple}\, m_0\\
&p'= \mathit{gradient\_descent}\,m\,loss\,data\,p\\
&model' = m_0\,p'
\end{align*}
First a model $m_0$ is created by instantiating the parameterised $model$ with some arbitrary, provisional, parameter values. Then $m_0$ is \textit{decoupled} into its parameters ($p:vec$) and a new model ($m:vec\rightarrow \mathit{float}\rightarrow(\mathit{float},\mathit{float})$) where all the parameters are brought together in a single vector argument. Finally a new set of parameters $p':vec$ is computed using generic optimisation using $p$ as the initial point in the search space. 

The decoupling operation needs to distinguish between provisional constants (parameters) and constants which do not require optimisation. We indicate their provisional status using braces in the syntax, so that $\{0\}$ stands for a constant with value 0 but which can be decoupled and made into a parameter. The final form of our example is, for example:
\begin{align*}
& linr\,  x = \{1\} * x + \{0\} \\
& quad\,  x = \{1\} * x * x + \{0\} * x + \{0\}\\
&m_0=model\, quad\, linr\\
&(m, p) = \textbf{decouple}\, m_0\\
&p'= \mathit{gradient\_descent}\,m\,loss\,data\,p\\
&m_1 = m\,p'
\end{align*}
where the provisional constants are given some arbitrary values. 

For comparison, in a system without programmable abduction and without decoupling, some possible implementations would require a explicit re-parametrisation of the model:
\begin{align*}
&linr\, (a,b)\, x = a * x + b\\
&quad\, (a,b,c)\, x = a * x * x + b * x + c\\
&m_0\, (a,b,c,d,e) = model\, (quad\, (a,b,c))\, (linr\, (d,e))  \\ 
&p' = \textbf{abduct}\, m_0\, data\\
&m_1 = m_0\, p'
\end{align*}
Or, if the parameters are collected into vectors, the example would be written as:
\begin{align*}
&linr\, p\, x = p[0] * x + p[1]\\
&quad\, p\, x = p[0] * x * x + p[1] * x + p[2]\\
&m_0\, p = model\, (quad\, (\textit{fst}\, p))\, (linr\, (\textit{snd}\, p)) \\
&p' = \textbf{abduct}\, m_0\, data\\
&m_1 = m_0\, p'
\end{align*}
where \textit{fst} and \textit{snd} are functions that select the appropriate parameters from the model to the functions parametrising the model. 

A final observation regarding the type of vectors resulting from decoupling the parameters of a model. The decoupling of parameters is a complex run-time operation, and the order in which they are stored in the vector is difficult to specify in a way that is exploitable by the programmer. Therefore we should restrict vector operations to those that are \textit{symmetric} under permutations of bases. In practice this means that we do not provide constants for the bases, which means that using vector addition, scalar multiplication and dot product it is not possible to have access to individual coordinates. This restriction allows the formulation of common general-purpose optimisations algorithms such as numerical gradient descent, which are symmetric under permutations of bases. This is a significant restriction only if the search takes advantage of coordinate-specific heuristics, such as the use of regularisation terms~\cite{DBLP:conf/aaai/XuS05}.

\section{Abductive calculus over a field}

Let $\F$ be a (fixed) set and $\A$ be a set of names (or \textit{atoms}).
Let $(\F,{+},{-},{\times},{/})$ be a field and $(V_a,{+_a}, {\times}_a,\bullet_a)$ an $\A$-indexed family of vector spaces over $\mathbb F$.
The types $T$ of the languages are defined by the grammar
$
 T ::= \F \mid V_a \mid T \to T. 
$
We refer to the field type $\F$ and vector types $V_a$ as ground
types.
Besides the standard operations contributed by the field and the vector spaces, denoted by $\Sigma$:
\begin{align*}
 0, 1, k 
 &: \F \tag{field constants}\\
 {+}, {-},{\times},{/}
 &: \F \to \F \to \F
 \tag{operations of the field $\F$} \\
 +_a &: V_a \to V_a \to V_a
 \tag{vector addition} \\
 \times_a &: \F \to V_a \to V_a
 \tag{scalar multiplication} \\
 \bullet_a &: V_a \to V_a \to \F
 \tag{dot product}, 
\end{align*}
we introduce iterated vector operations, denoted by $\Sigma^\iter$:
\begin{align*}
 +^L_a
 &: (V_a \to V_a) \to V_a \to V_a
 \tag{left-iterative vector addition} \\
 \times^L_a
 &: (V_a \to \F) \to V_a \to V_a
 \tag{left-iterative scalar multiplication}
\end{align*}
All the vector operations are indexed by a name $a \in \A$, and
symbols $+$ and $\times$ are overloaded. The role of the name $a$ will be discussed later, for now it may be disregarded.

Iterative vector operations apply vector operations uniformly over the
entire standard basis.
The iterative vector operations are informally defined as folds over the list of ordered vector bases $E_a=[\vec e_0,\ldots,\vec e_{n-1}]$. These are informal definition because lists (and folds) are not part of our syntax:
\begin{align*}
 f +^L_a v_0 &:=\prim{foldr} (\lambda e\lambda v.f(e)+v)\,E_a\,v_0 \\
 f \times^L_a v_0 &:=\prim{foldr} (\lambda e\lambda v.f(e)\times v)\,E_a\,v_0.
\end{align*}
Terms $t$ are defined by the grammar
$
 t ::= x \mid \abs{x}{T}{t} \mid \app{t}{t}
 \mid \val{p} \mid t \mathrel{\$} t
 \mid \param{\val{p}} \mid \abd{a}{f}{x}{T}{t}
$,
where $T$ is a type, $f$ and $x$ are variables,
$\$ \in \Sigma \cup \Sigma^\iter$ is a primitive operation, and
$p \in \F$ is an element of the field.
The novel syntactic elements of the language are provisional constants $\param{\val{p}}$ and a family of type and name-indexed decoupling operations $\abd{a}{f}{x}{T}{t}$, as discussed in Sec.~\ref{sec:abd}. 

Let $A \finsubset \A$ be a finite set of names, $\Gamma$  a
sequence of typed variables $x_i {:} T_i$, and $\vec{p}$  a sequence
of elements of the field $\F$ (i.e.\ a vector over $\F$).
We write $A \vdash \Gamma$ if $A$ is the support of 
$\Gamma$.
The type judgements are of  shape:
$
 A \mid \Gamma \mid \vec{p} \vdash t:T
$,
and type derivation rules are as below.
\begin{center}\footnotesize
 $\infer{A \mid \Gamma, {x:T} \mid - \vdash x : T}
 {A \vdash \Gamma,T}$ \qquad
 $\infer{A \mid \Gamma \mid \vec{p} \vdash \abs{x}{T'}{t} : T' \to T}
 {A \mid \Gamma, {x:T'} \mid \vec{p} \vdash t : T}$ \qquad
 $\infer{A \mid \Gamma \mid \vec{p},\vec{q} \vdash \app{t}{u} : T}
 {A \mid \Gamma \mid \vec{p} \vdash t : T' \to T
 & A \mid \Gamma \mid \vec{q} \vdash u : T}$ \qquad \\[1.5ex]
 $\infer{A \mid \Gamma \mid - \vdash \val{p} : \F}
 {p \in \F}$ \qquad
 $\infer{A \mid \Gamma \mid \vec{p}, \vec{q}
 \vdash t_1 \mathrel{\$} t_2 : T}
 {A \mid \Gamma \mid \vec{p} \vdash t_1 : T_1
 & A \mid \Gamma \mid \vec{q} \vdash t_2 : T_2
 & {\$ : T_1 \to T_2 \to T} \in \Sigma}$
 \qquad \\[1.5ex]
 $\infer{A \mid \Gamma \mid p
 \vdash \param{\val{p}} : \F}{}$ \qquad
 $\infer
 {A \mid \Gamma \mid \vec{p}
 \vdash \abd{a}{f}{x}{T'}{t} : T' \to T}
 {A,a \mid \Gamma, {f:V_a\to T'}, {x:V_a} \mid \vec{p}
 \vdash t : T
 & A \vdash \Gamma,T',T}$
\end{center}
Note that the rules are linear with respect to the parameters
$\vec{p}$.
In a derivable judgement $A \mid \Gamma \mid \vec{p} \vdash t : T$,
the vector $\vec{p}$ gives the collection of all the provisional
constants in the term $t$.

Abductive decoupling $\abd{a}{f}{x}{T'}{t}$ serves as a binder of the name
$a$ and, therefore, it requires in its typing a unique vector type $V_a$ collecting all the provisional constants. Because of name $a$ this vector type cannot be used outside of the scope of the operation. An immediate consequence is that variables $f$ and $x$ used in the decoupling of a term share the type $V_a$ but this type cannot be mixed with parameters produced by other decouplings. The simple reason for preventing this is that the sizes of the vectors may be different. A more subtle reason is that we prefer not to assume a particular order of placing parameters in the vector, yet we aim to preserve determinism of computation. Because the order of parameters is unknown, we must  only allow operations which are invariant over permutations of bases. Therefore only certain iterative vector operations are allowed. The most significant restriction is that point-wise access to the bases or the components is banned. 


\section{GoI-style semantics}
We give an operational semantics of the language as an abstract
machine.
The abstract machine rewrites a graph that is an inductively defined translation of a
program, by passing a token on the graph.
The token triggers graph rewriting in a deterministic way by
carrying data which defines redexes, as well as carrying data representing results of computations. This abstract machine is closely based on the dynamic GoI machine~\cite{DBLP:conf/csl/MuroyaG17}. As it should be soon evident, the graph-rewriting semantics is particularly suitable for  tracking the evolving data dependencies in abductive programs. 

\subsection{Graphs and graph states}

A graph is given by a set of nodes and a set of edges. The nodes are
partitioned into \textit{proper nodes} and \textit{link nodes}.
A distinguished list of link nodes forms the \textit{input interface}
and another list of link nodes forms the \textit{output interface}.
Edges are directed, with one link node and one proper node as
endpoints.
An input link (i.e.\ a link in the input interface) is the source of
exactly one edge and the target of no edge.
Similarly an output link (i.e.\ a link in the output interface) is the
source of no edge and the target of exactly one edge.
Every other link must be the source of one edge and the target of
another one edge.
We may write $G(n,m)$ to indicate that a graph $G$ has $n$ links in
 the input interface and $m$ links in the output interface.
From now on we will refer to proper nodes as just ``nodes,'' and link
nodes as ``links.''

Links are labelled by \emph{enriched} types $\tilde{T}$, defined
by $\tilde{T} ::= T \mid \oc T \mid \rot{\oc} \F$ where $T$ is any
type of terms.
If a graph has only one input, we call it ``root,'' and say
the graph has enriched type $\tilde{T}$ if the root has the enriched
type $\tilde{T}$.
We sometimes refer to enriched types just as ``types,'' while calling
the enriched type $\rot{\oc} \F$ ``provisional type'' and an enriched
type $\oc T$ ``argument type.''


Nodes are labelled, and we call a node labelled with $X$ an
``$X$-node.''
 Labels of nodes fall into several categories. Some of them correspond to the basic syntactic constructs of the lambda calculs: $\lambda$ (abstraction), $@$ (application), $\val{p} \in \F$ (scalar constants),
$\val{\vec{p}} \in \F^n$ (vector constants), $\$ \in \Sigma \cup \Sigma^\iter$ (operations). 
Nodes labelled $\lb{C}_n$ and $\rot{\lb{C}}_n$ handle contraction for definitive terms and for provisional constants, respectively. Node $\lb{P}_n$ handles the decomposition of a vector in its elements (coordinates). Node $\lb{A}$ indicates an abductive decoupling. Nodes $\oc$, $\wn$, $\rot\oc $, $\rot\wn $, $\lb{D}$, $\rot{\lb{D}}$ play the same role as exponential nodes in proof nets, and are needed by the bureaucracy of how sharing is managed.

When drawing graphs certain diagrammatic conventions are
employed. Link nodes are not represented explicitly, and their labels
are only given when they cannot be easily inferred from the rest of
the graph. By graphical convention, the link nodes at the bottom of
the diagram represent the input interface and they are ordered left to
right; the link nodes at the top of the diagram are the output, ordered left to right. A double-stroke edge represents a bunch of edges running in parallel and a double stroke node represents a bunch of nodes. If it is not clear from context we annotate a double-stroke edge with the number of edges in the bunch:
\begin{center}
	\includegraphics[]{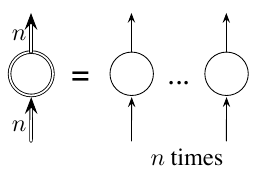}
\end{center}

The connection of edges via nodes must satisfy the rules in
Fig.~\ref{fig:GraphConnection}, where $T_1$ and $T_2$ are types,
$\oc \vec{T}$ denotes a sequence $\oc T_1,\ldots,\oc T_m$ of enriched types, $a \in \A$, $\$^0 : T_1 \to T_2 \to T \in \Sigma$ is a ground-type primitive, and $n$  a natural number.
\begin{figure}[t]
	\centering
    \includegraphics{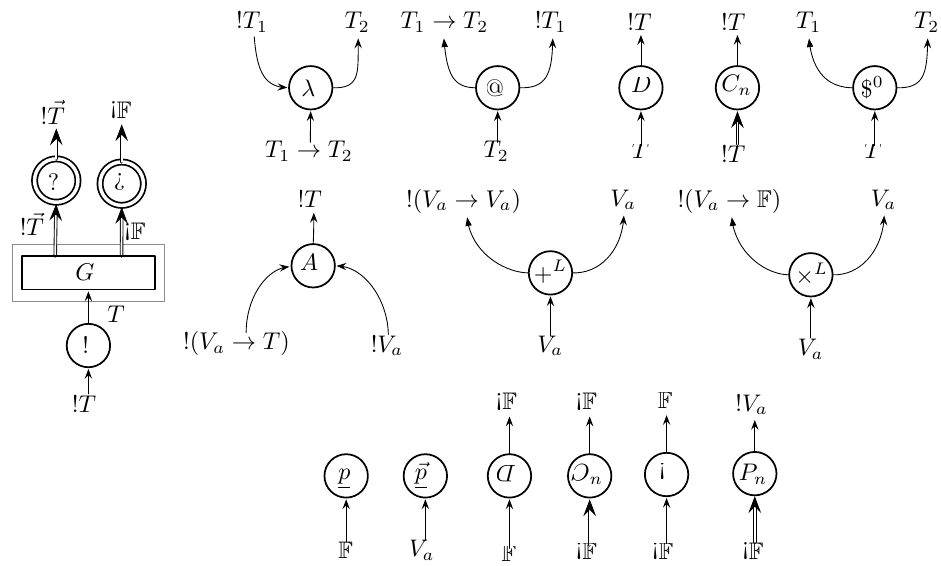}
 \caption{Connection of Edges}
  \label{fig:GraphConnection}
\end{figure}
The  outline box in Figure~\ref{fig:GraphConnection} indicates a subgraph $G(1,n_1 + n_2)$, called an \emph{$\oc$-box}.
Its input is connected to one $\oc$-node (``principal door''), while the outputs are connected to 
$n_1$  $\wn$-nodes (``definitive auxiliary doors''), and 
$n_2$  $\rot\wn $-nodes (``provisional auxiliary doors'').

A \textit{graph context} is a graph, that has exactly one extra new
node with label ``$\square$'' and interfaces of arbitrary numbers and
types of input and output.
We write a graph context as $\mathcal G[\square]$ and call the unique
extra $\square$-node ``hole.''
When a graph $G$ has the same interfaces as the $\square$-node in a
graph context $\mathcal{G}[\square]$, we write
$\mathcal G[G]=\mathcal G[\square/G]$ for the substitution of the hole
by the graph $G$.
The resulting graph $\mathcal{G}[G]$ indeed satisfies the rules in
Fig.~\ref{fig:GraphConnection}, thanks to the matching of interfaces.

We say that a graph is \textit{definitive} if it contains no
$\rot\oc$-nodes and all its output links have the provisional type
$\rot\oc \F$.
When a graph $G(1,0)$ can be decomposed into:
 \begin{center}
 	\includegraphics[]{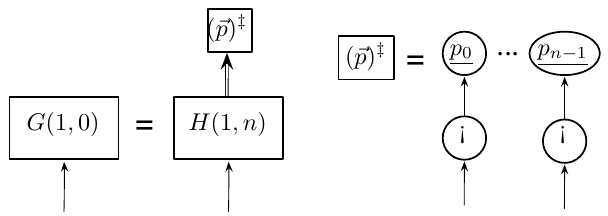}
 \end{center}
where $H(1,n)$ is a definitive graph and $\vec{p} \in \F^n$, we write
$G = H \circ (\vec{p})^\ddag$ and call the graph ``composite.''
We order $\rot\oc$-nodes in the composite graph $G$, effectively in
the component $(\vec{p})^\ddag$, according to the vector $\vec{p}$.
\begin{definition}[Graph states]
 \label{def:GraphStates}
 A \emph{graph state} $((G,\ell), \delta)$ consists of a
 composite
 graph $G = H \circ (\vec{p})^\ddag$ with a distinguished link node
 $\ell$, and \emph{token data} $\delta = (d,f,S,B)$ that consists of:
 a \emph{direction} $d$, a
 \emph{rewriting flag} $f$, a \emph{computation stack} $S$ and a
 \emph{box stack} $B$, defined by:
 \begin{align*}
  d &::= \up \mid \dn \\
  f &::= \square \mid \lambda \mid \$^0 \mid \$^1(n)
  \mid \wn \mid \oc \\
  S &::= \square \mid @ \cl S \mid \star \cl S
  \mid \lambda \cl S \mid \val{p} \cl S \mid \val{\vec{p}} \cl S \\
  B &::= \square \mid \ell' \cl B
 \end{align*}
 where $\$^0 \in \Sigma$ and $\$^1 \in \Sigma^\iter$ are primitives,
 $p \in \F$, $\vec{p}$ is a vector over $\F$, $n$ is a natural
 number, and $\ell'$ is a link of the graph $G$.
\end{definition}
In the definition above we call the link node $\ell$ of $(G,\ell)$ the
\textit{position} of the token.

\subsection{Transitions}

We define a relation on graph states called \textit{transitions}
$((G,\ell), \delta) \to ((G',\ell'), \delta')$. Transitions are either \emph{pass} or  \emph{rewrite}. 

Pass transitions
$((\mathcal{G}[G],\ell), (d,\square,S,B))
\to ((\mathcal{G}[G],\ell'), (d',f',S',B'))$
occur if and only if the rewriting flag in the state is $\square$.
In these transitions, the sub-graph $G$ consists of just one node with
its
interfaces, and the old and new positions $\ell$ and $\ell'$ are both
in $G$ (i.e.\ in its interfaces).
These transitions do not change the overall graph $\mathcal{G}[G]$ but
only token data and token position, as shown in
Fig.~\ref{fig:PassTransitions}.
In particular the stacks are updated by changing only a constant
number of top elements.
In the figure, only the singleton sub-graph $G$ is presented, and
token position is indicated by a black triangle pointing towards the
direction of travel. The symbol $-$ denotes any single element of a
computation stack, and
$X\in\{\rot{!},\rot{?},\rot{\lb D},\lb D\}$ and
$Y\in\{\rot{!},\rot{?},\rot{\lb D}\}$.
When the token goes upwards through a $Z$-node, where
$Z\in \{\rot{\lb C}_n,\lb C_n \mid n>0\}$, 
the previous position $\ell$ is pushed to the box stack.
The pass transition over a $W$-node, where
$W\in \{\rot{\lb C}_n \mid n>0\}$,
assumes the top element $\ell'$ of the box stack to be one of input of
the $W$-node.
The element $\ell'$ is popped and set to be a new position of the
token.
\begin{figure}[t]
	\centering
	\includegraphics{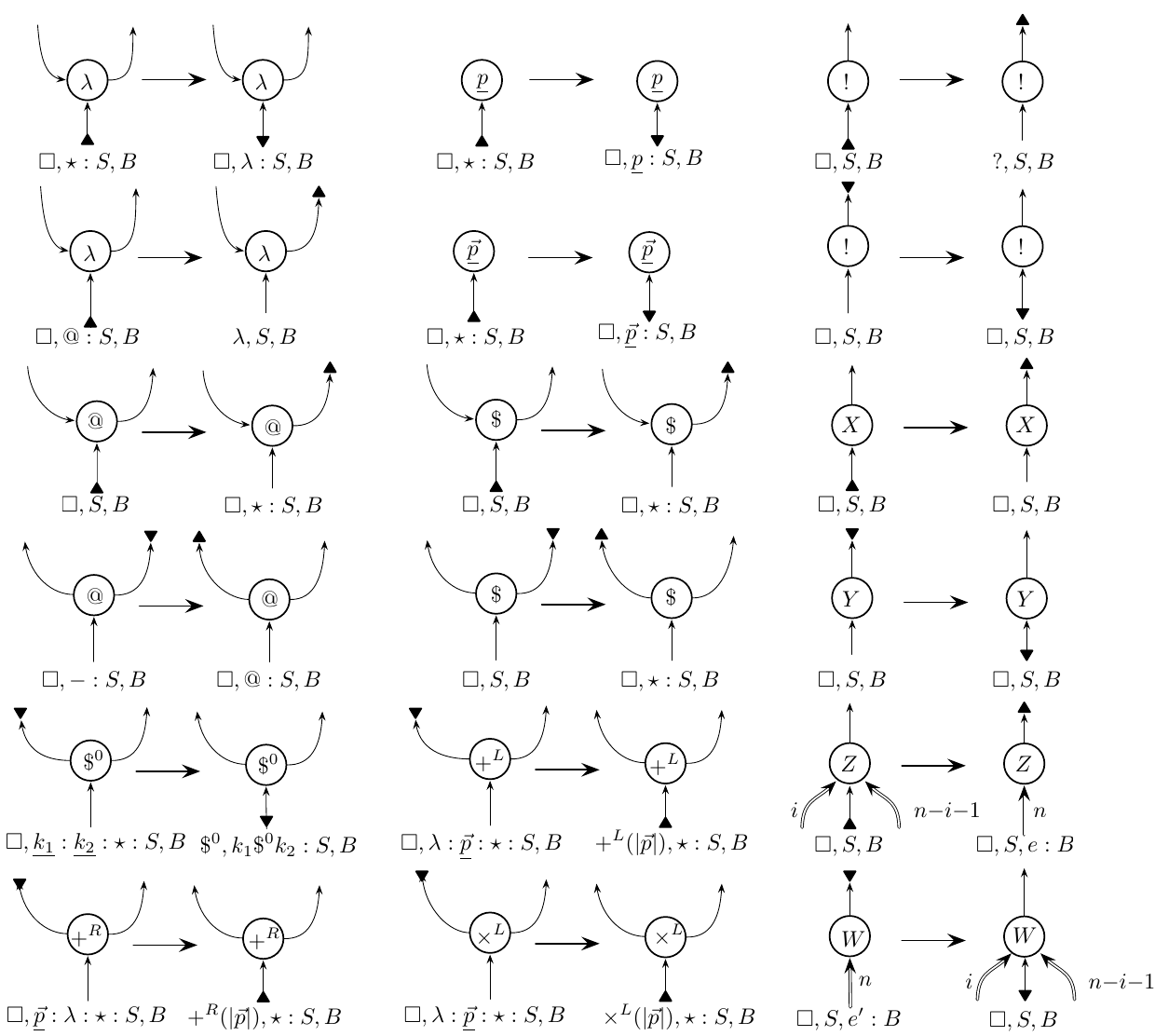}
	\caption{Pass Transitions}
	\label{fig:PassTransitions}
\end{figure}

Inspecting the transition rules reveals basic intuitions about the intended semantics of the language. On the evaluation of an application ($@$) or operation ($\$$), indicated by the token moving into the node, the token is first propagated to the \textit{right} edge and, as it arrives back from the right, it is then propagated to the \textit{left} edge: function application and operators are evaluated right-to-left. A left-to-right application is possible but more convoluted, noting that mainstream CBV language compilers such as \textsc{OCaml} also sometimes use right-to-left evaluation for the sake of efficiency. After evaluating both branches of an operation ($\$$) the token propagates downwards carrying the resulting value. From this point of view a constant can be seen as a trivial operation of arity 0. 

The behaviour of abstraction ($\lambda$) and application ($@$) nodes is more subtle. The token never exits an application node ($ @ $) because in a closed term the application will always eventually trigger a graph rewrite which eliminates a $\lambda$-$@$ pair of nodes. An abstraction node either simply returns the token placing a $\lambda$ at the top of the computation stack to indicate that the function is a ``value'', or it process the token if it sees an $@$ at the top of the computation stack, in the expectation that an applicative cancellation of nodes will follow, as seen next. The other nodes (!, ?, etc.) can only be properly understood in the context of the rewrite transitions.

Rewrite transitions
$((\mathcal{G}[G],\ell), (d,f,S,B))
\to ((\mathcal{G}[G'],\ell'), (d,f',S,B'))$
apply to states where the rewriting flag is not $\square$, i.e.\ to
which pass transitions never apply.
They replace the (sub-)graph $G$ with $G'$, keeping the interfaces,
move the position, and modify the box stack, without changing the
direction and the computation stack.
We call the sub-graph $G$ ``redex,'' and a rewrite transition
``$f$-rewrite transition'' if a rewriting flag is $f$ before the
transition.

The redex may or may not contain the token position $\ell$.
We call a rewrite transition ``local'' if its redex contains the token
position, and ``deep'' if not.
Fig.~\ref{fig:RewriteTransitionsComput},
Fig.~\ref{fig:RewriteTransitionsOpenBox} and
Fig.~\ref{fig:RewriteTransitionsClosedBox} define local rewrites,
showing only redexes.
Fig.~\ref{fig:RewriteTransitionsDeepGeneral}, complemented by
Fig.~\ref{fig:RewriteTransitionsDeep}, defines deep rewrites, whose
redexes we will specify later.
We explain some rewrite transitions in detail.
\begin{figure}[t]
 \centering
	 \includegraphics{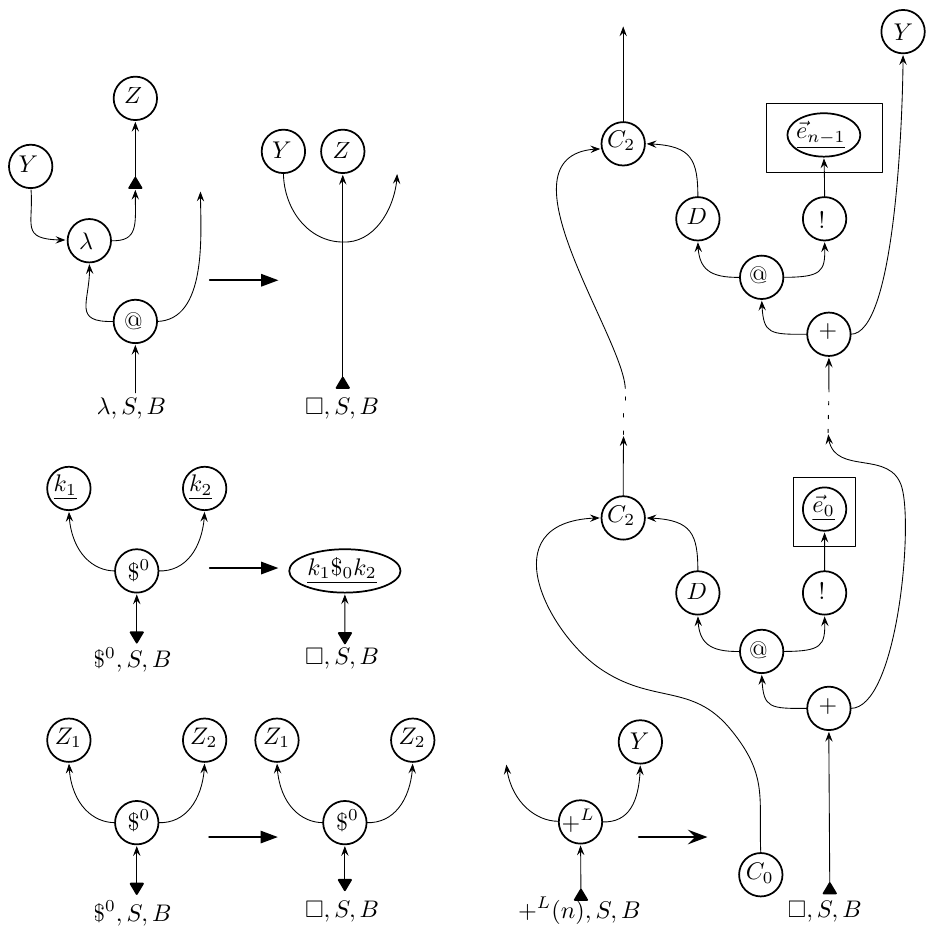}
	\caption{Rewrite Transitions: computation}
	\label{fig:RewriteTransitionsComput}
\end{figure}

The rewrites in Fig.~\ref{fig:RewriteTransitionsComput} are \textit{computational} in the sense that they are the common rewrites for CBV lambda calculus extended with constants (scalars and vectors) and operations. The first rewrite is the elimination of a $\lambda$-$@$ pair, a key step in beta reduction. Following the rewrite, the incoming output edge of $\lambda$ will connect directly to the argument, and the token will enter the body of the function. Simpler operations also reduce their arguments, if they are constants, replacing them with a single constant. If the arguments are not constant-nodes then they are not rewritten out, expressly to prevent the deletion of provisional, abductable, constants. Finally, iterated (fold-like) constants are recursively (on the size of the index $n$ in the token data) unfolded until the computation is expressed in terms of simple operations ($\times^L$ has an unfolding rule similar to that of $+^L$). 
The unfolding introduces nodes
$\val{\vec{e_0}},\ldots,\val{\vec{e_{n-1}}}$ that are the (ordered)
standard basis of the vector space $V_a$. Note that these bases are only computed at run-time and are not accessible from syntax.

The rewrites in Fig.~\ref{fig:RewriteTransitionsDeepGeneral},
Fig.~\ref{fig:RewriteTransitionsDeep},
Fig.~\ref{fig:RewriteTransitionsOpenBox} and
Fig.~\ref{fig:RewriteTransitionsClosedBox} govern the behaviour of
!-boxes and are essential in implementing abductive behaviour. They
are triggered by rewriting flags $\wn$ or $\oc$, whenever the
token reaches the principal door of a $\oc$-box.

The first class of the $\oc$-box rewrites is \emph{deep} rewrites,
whose general form is shown in
Fig.~\ref{fig:RewriteTransitionsDeepGeneral}, and
actual rewriting rules are shown in
Fig.~\ref{fig:RewriteTransitionsDeep}.
Let us write $G[X/Y]$ for a graph $G$ in which all $Y$-nodes are
replaced with $X$-nodes, and similarly, write $G[\epsilon/Y]$ for a
graph $G$ in which all $Y$-nodes are eliminated.
We can see the ``deepness'' of the rules in
Fig.~\ref{fig:RewriteTransitionsDeep}, as they occur in the graph
$E(m+m',0)$ (in
Fig.~\ref{fig:RewriteTransitionsDeepGeneral}) which may have not been
visited by the token yet.
The deep rules can be applied only if the $\lb{A}$-node,
$\lb{P}_n$-node or $\lb{C}_n$-node
(in Fig.~\ref{fig:RewriteTransitionsDeep}) satisfies the following:
\begin{itemize}
 \item the node is ``box-reachable'' (see
       Def.~\ref{def:BoxReachability} below) from one of
       definitive auxiliary doors of the $\oc$-box $G$
 \item the node is in the same ``level'' as one of definitive
       auxiliary doors of the $\oc$-box $G$, i.e.\ the node is in a
       $\oc$-box if and only if the door is in the same $\oc$-box.
\end{itemize}
\begin{definition}[Box-reachability]
 \label{def:BoxReachability}
 In a graph, a node/link $v$ is \emph{box-reachable} from a node/link
 $v'$
 if there exists a finite sequence of directed paths $p_0,\ldots,p_i$
 such that: (i) for any $j = 0,\ldots,i-1$, the path $p_j$ ends with
 the root of a $\oc$-box and the path $p_{j+1}$ begins with an output
 link of the $\oc$-box, and (ii) the path $p_0$ begins with $v$ and
 the path $p_i$ ends with $v'$.
\end{definition}
\noindent
We call the sequence of paths in the above definition ``box-path.''
Box-reachability is a weaker notion of the normal graphical
reachability which is witnessed by a single directed path, as it can
skip $\oc$-boxes.
Redex searching for deep rules can be done by searching a graph from
definitive auxiliary doors while skipping $\oc$-boxes on the way.
Note that the deep rules do not apply to $\lb{C}_0$-nodes and
$\lb{P}_0$-nodes, as they do not satisfy the box-reachablity
condition.

\begin{figure}[t]
 \begin{minipage}[t]{\hsize}
  \centering
	  \includegraphics{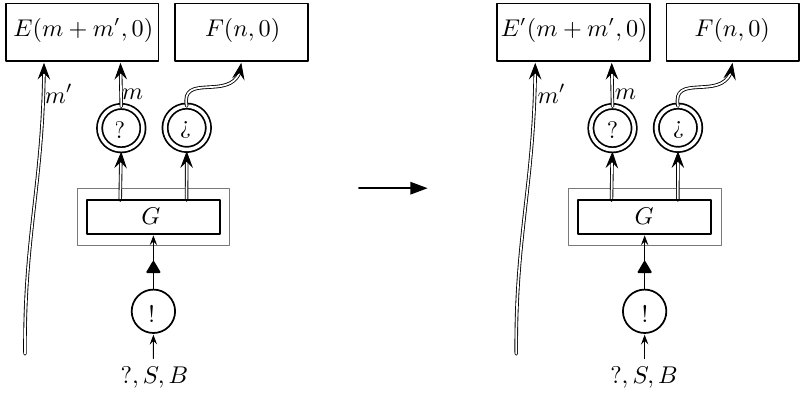}
  \caption{Rewrite transitions: general form of deep rewrites}
  \label{fig:RewriteTransitionsDeepGeneral}
 \end{minipage}
 \begin{minipage}[b]{\hsize}
  \centering
  \includegraphics[]{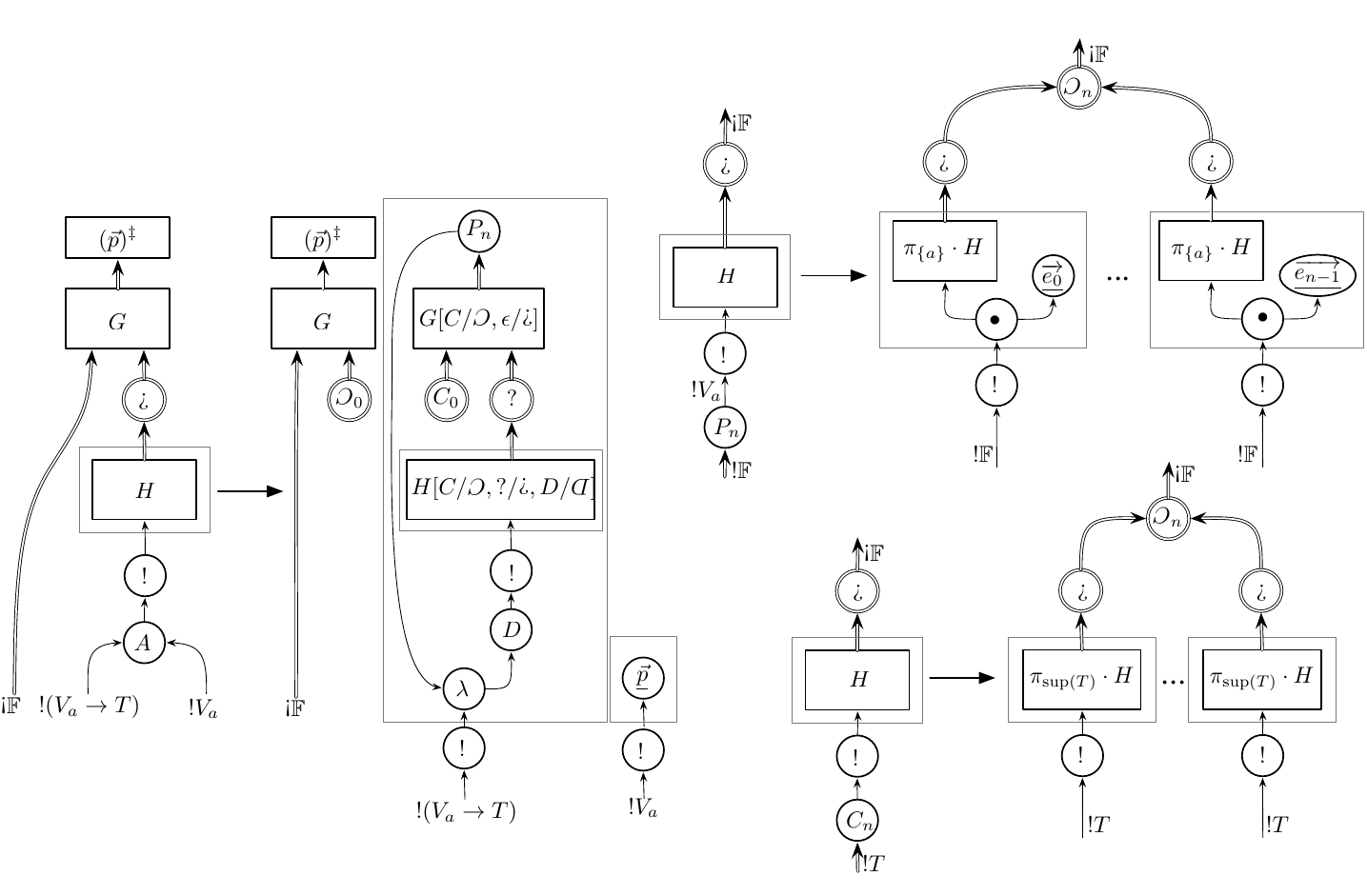}
  \caption{Deep rules: decoupling, projection, contraction}
  \label{fig:RewriteTransitionsDeep}
 \end{minipage}
\end{figure}

Upon applying the first deep rule, the two input edges of the node $A$ will connect to the  decoupled function and arguments. The function is created by replacing the provisional constants $(\vec p)^\ddagger$ with a projection and a $\lambda$-node (plus a dereliction ($\rot{\lb D}$) node). A copy of the provisional constants used by other parts of the graph is left in place, and another copy is transformed into a single vector node and linked to the second input of the decoupling, which now has access to the current parameter values. 
Note that the sub-graph $G \circ (\vec{p})^\ddag$ is not modified by
the decoupling rule.
We define the redex of the deep decoupling rule to be the $\oc$-box $H$
with its doors and the connected $\lb{A}$-node, excluding the
unchanged sub-graph $G \circ (\vec{p})^\ddag$.

The second deep rule handles vector projections. Any graph $H$
handling a vector value in a vector with $n$ dimensions is replicated
$n$ times to handle each coordinate separately. The projected value is
computed by applying the dot product using the corresponding standard base. Finally, the names in $H$ are refreshed using the name permutation action $\pi_N$, where $N\subseteq \mathbb A$, defined as follows: all names in $N$ are preserved, all other names are replaced with fresh (globally to the whole graph) names. Finally, contraction is also eliminated by replicating the graph $H$ which handles it, while refreshing all names in $H$ which do not appear in $T$. 
In the deep projection/contraction rules, redexes are as shown in
Fig.~\ref{fig:RewriteTransitionsDeep}.
The $\lb{C}_n$-node and the $\lb{P}_n$-node in the redexes necessarily
satisfy $n > 0$, due to the applicable condtion of deep rules.

Names indexing the vector types must be refreshed because as a result of copying, any decoupling may be executed several times, and each time the resulting models and parameters must be kept distinct from previously decoupled models and parameters. 
This is discussed in more depth in Appendix~\ref{app:Validity}.

\begin{figure}[t]
	\centering
	\includegraphics[]{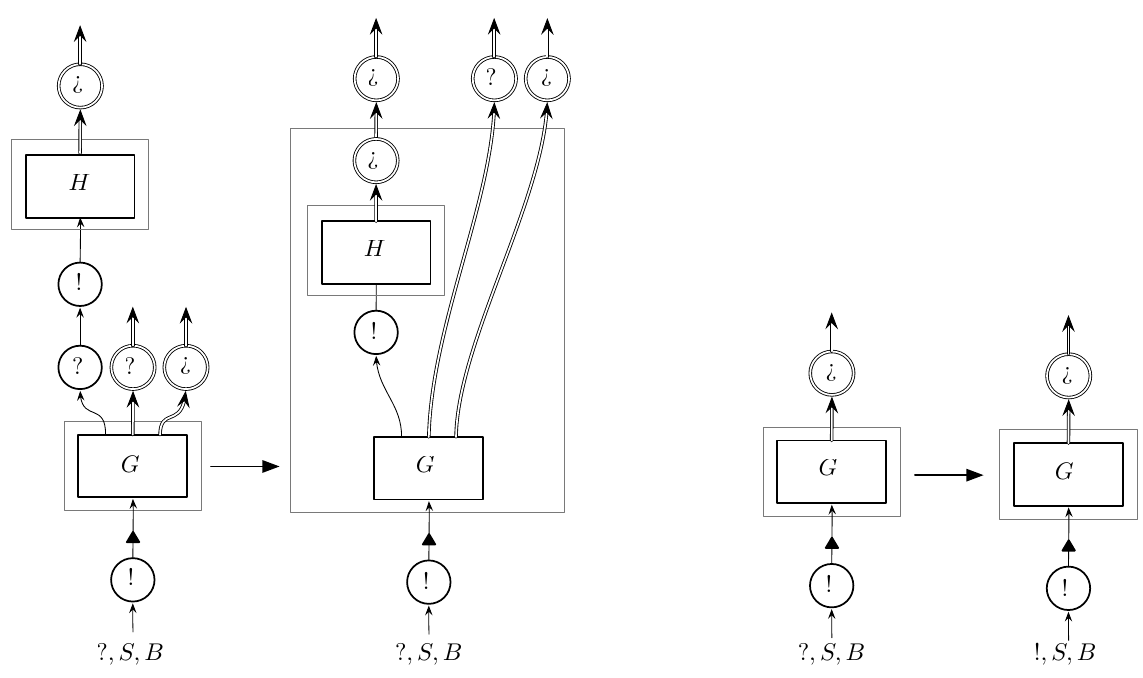}
	\caption{Rewrite Transitions: closing $\oc$-boxes}
	\label{fig:RewriteTransitionsOpenBox}
\end{figure}
Fig.~\ref{fig:RewriteTransitionsOpenBox} shows the second class of
$\oc$-box rewrites.
The left rewrite happens to the $\oc$-box $G$ above the token, namely
it absorbs all other boxes $H$, one by one, to which it is directly
connected. Because the $?$-nodes of $!$-boxes arise from the use of
global or free variables, this box-absorption process mirrors that of
closure-creation in conventional operational semantics.
After all the definitive auxiliary doors of the $\oc$-box are
eliminated, the flag changes from `?' to `!', meaning that the token
can further trigger the last class of rewrites, shown in
Fig.~\ref{fig:RewriteTransitionsClosedBox}, which handles copying.
\begin{figure}[h]
 \centering
 \includegraphics{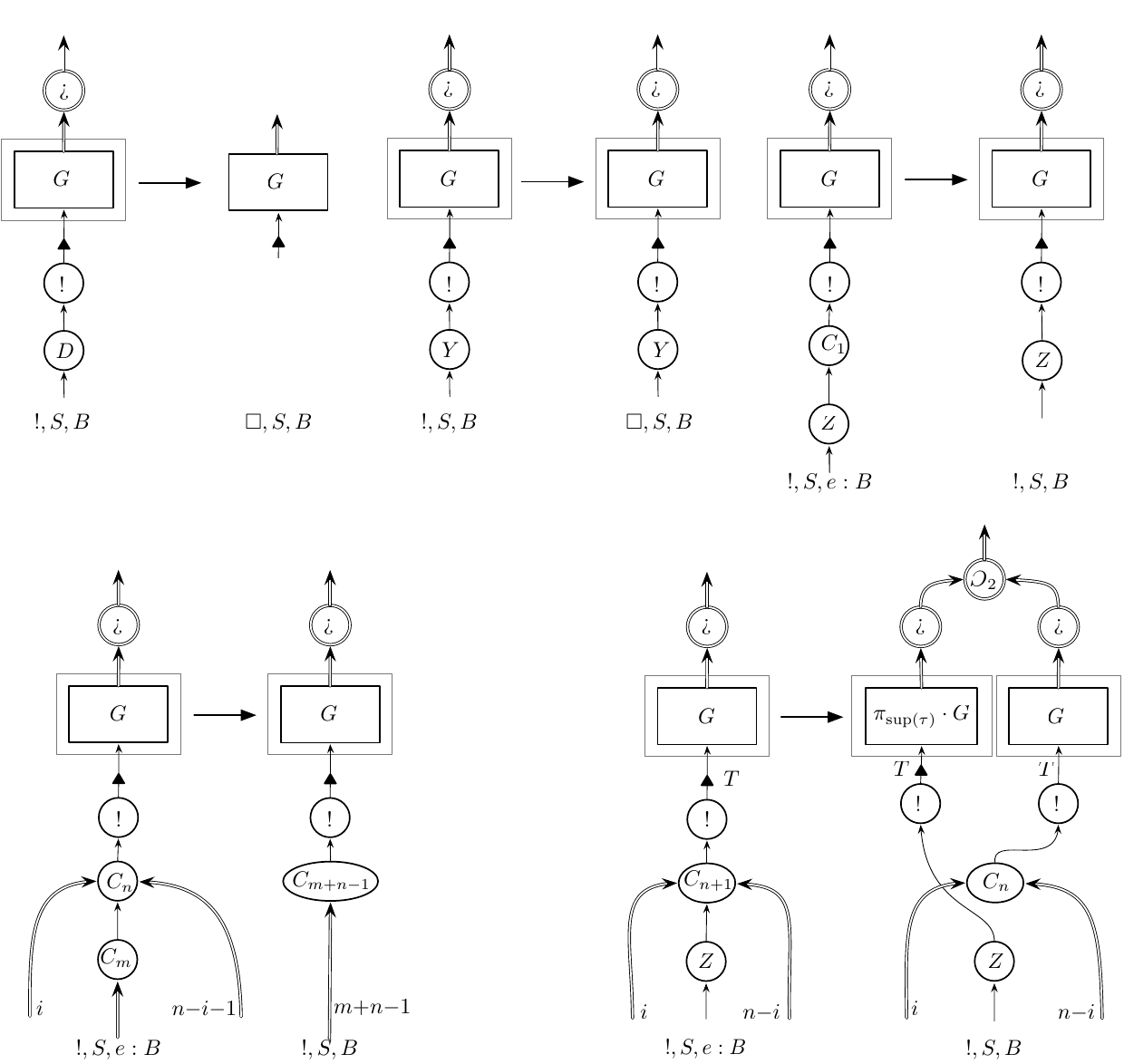}
 \caption{Rewrite Transitions: copying closed $\oc$-boxes}
 \label{fig:RewriteTransitionsClosedBox}
\end{figure}

Rewrites in Fig.~\ref{fig:RewriteTransitionsClosedBox} are several
simple bureaucratic rewrites, involving copying of closed !-boxes.
The two top-left rewrites, where
$Y\not\in\{\lb D, \lb C_k\mid k\in\mathbb N \}$,
change rewrite mode to pass mode, by setting the rewriting flag to
$\square$.
The top-right rewrite eliminates the trivial contraction $\lb{C}_1$,
discarding the top element $\ell$ of the box stack.
The bottom-left rewrite combines contraction nodes.
It consumes the top element $\ell$ of the box stack to detect the
lower contraction node ($\lb{C}_m$), therefore the link $\ell$ is
assumed to be between two contraction nodes ($\lb{C}_n$ and
$\lb{C}_m$).
The bottom-right rewrite, where
$Z\not\in\{\lb C_k \mid k\in\mathbb N \}$, actually copies a
$\oc$-box.
It consumes the top element $\ell$ of the box stack to detect the
$Z$-node, therefore the link $\ell$ is assumed to be between the
$\lb{C}_{n+1}$-node and the $Z$-node.

Finally, we can confirm that all the transitions presented so far are
well-defined.
\begin{proposition}[Form preservation]
 \label{prop:FormPreservation}
 All transitions indeed send a graph state to another graph state, in
 particular a composite graph $G \circ (\vec{p})^\ddag$ to a composite
 graph
 $G' \circ (\vec{p})^\ddag$ of the same type.
\end{proposition}
\begin{proof}
 Any transitions make changes only in a definitive graph and keeps
 the graph $\vec{p}^\ddag$ which contains only constant nodes and
 $\rot\oc$-nodes.
 They do not change the shape and types of interfaces of a redex.
\end{proof}

All the pass transitions are deterministic.
The rewrite rewrites are usually deterministic, except for the deep rewrites and the copying rewrites. However, these rewrites are confluent as no redexes are shared between rewrites,  so the overall beginning-to-end execution is deterministic.
\begin{definition}[Initial/final states and execution]
 Let $G$ be a composite graph with root $\ell$.
 An \emph{initial} state $\Init(G)$ on the graph $G$ is
 given by $((G,\ell),(\up,\square,\star \cl \square,\square))$.
 A \emph{final} state $\Final(G,X)$ on the graph $G$, with a single
 element $X$ of a computation stack, is given by
 $((G,\ell),(\dn,\square,X \cl \square,\square))$.
 An \emph{execution} on the graph $G$ is any sequence of transitions
 from the initial state $\Init(G)$.
\end{definition}
\begin{proposition}[Determinism of execution]
 \label{prop:Determinism}
 For any initial state $\Init(G)$, the final state $\Final(G,X)$ such
 that $\Init(G) \to^* \Final(G,X)$ is unique up to name permutation,
 if it exists.
\end{proposition}
\begin{proof}
 See Appendix~\ref{app:Determinism}.
\end{proof}

\section{Operational semantics of the abductive calculus}

\subsection{Translation of terms to graphs}

A derivable type judgement $A \mid \Gamma \mid \vec{p} \vdash t:T$ is
translated to a composite graph which we write as
$(A \mid \Gamma \mid \vec{p} \vdash t:T)^\ddag
= (A \mid \Gamma \mid \vec{p} \vdash t:T)^\dag \circ \vec{p}$.

Fig.~\ref{fig:InductiveTranslation} shows the inductive definition
of the definitive graph $(A \mid \Gamma \mid \vec{p} \vdash t:T)^\dag$.
Given a sequence $\Gamma = {x_0:T_0},\ldots,{x_{m-1}:T_{m-1}}$ of typed
variables, $\oc \Gamma$ denotes the sequence
$\oc T_0,\ldots,\oc T_{m-1}$ of (enriched) types.
\begin{figure}
 \centering
 \includegraphics{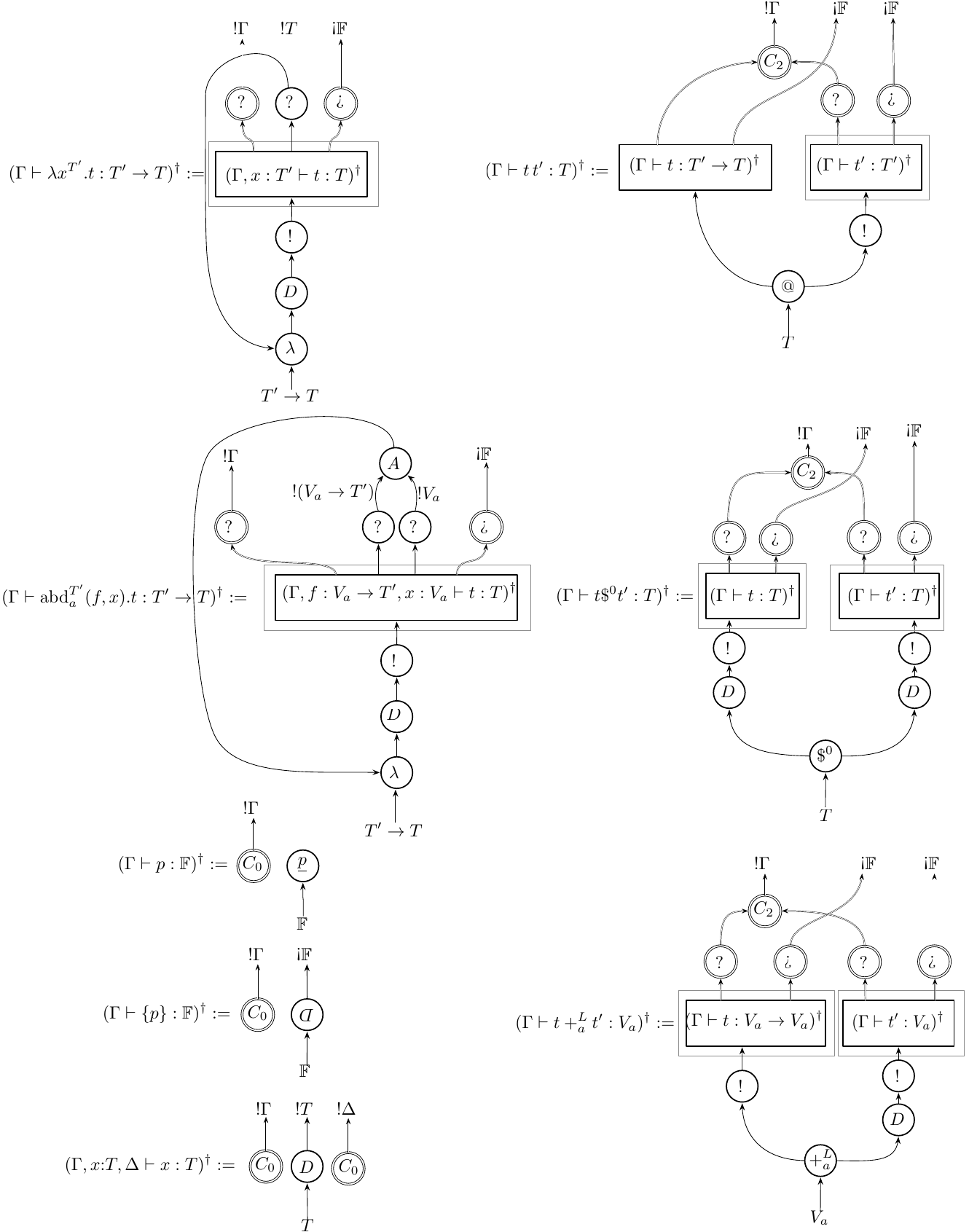}
 \caption{Inductive Translation}
 \label{fig:InductiveTranslation}
\end{figure}
Note that the translation does not contain any $\lb{P}_n$-nodes,
$\val{\vec{q}}$-nodes or $\rot{\lb{C}}_n$-nodes; they are generated by
rewrite transitions.


\subsection{Examples}

With the operational semantics in place we can return to formally re-examine the examples from the Introduction. All the examples below are executed using the on-line visualiser\footnote{\texttt{\url{http://bit.ly/2uaorPx}}}. All examples are pre-loaded into the visualiser menu. Note that the on-line visualiser uses additional \textit{garbage collection} rules as discussed in Sec.~\ref{sec:gc}.

\subsubsection{Simple abduction}

An extremely simple abductive program is:
$
\letin{f\,@\,p}{\{1\}+2}{f\, p}
$, which  decouples a parameter from a ground-type term and re-applies it immediately to the resulting model, evaluating to the same value, thus deprecating a provisional constant in a model to a definitive constant. This can be useful for computationally simplifying a model. Some key steps in the execution are given in Fig.~\ref{fig:figex1}.
\begin{figure}
\centering 
	\includegraphics[scale=0.6]{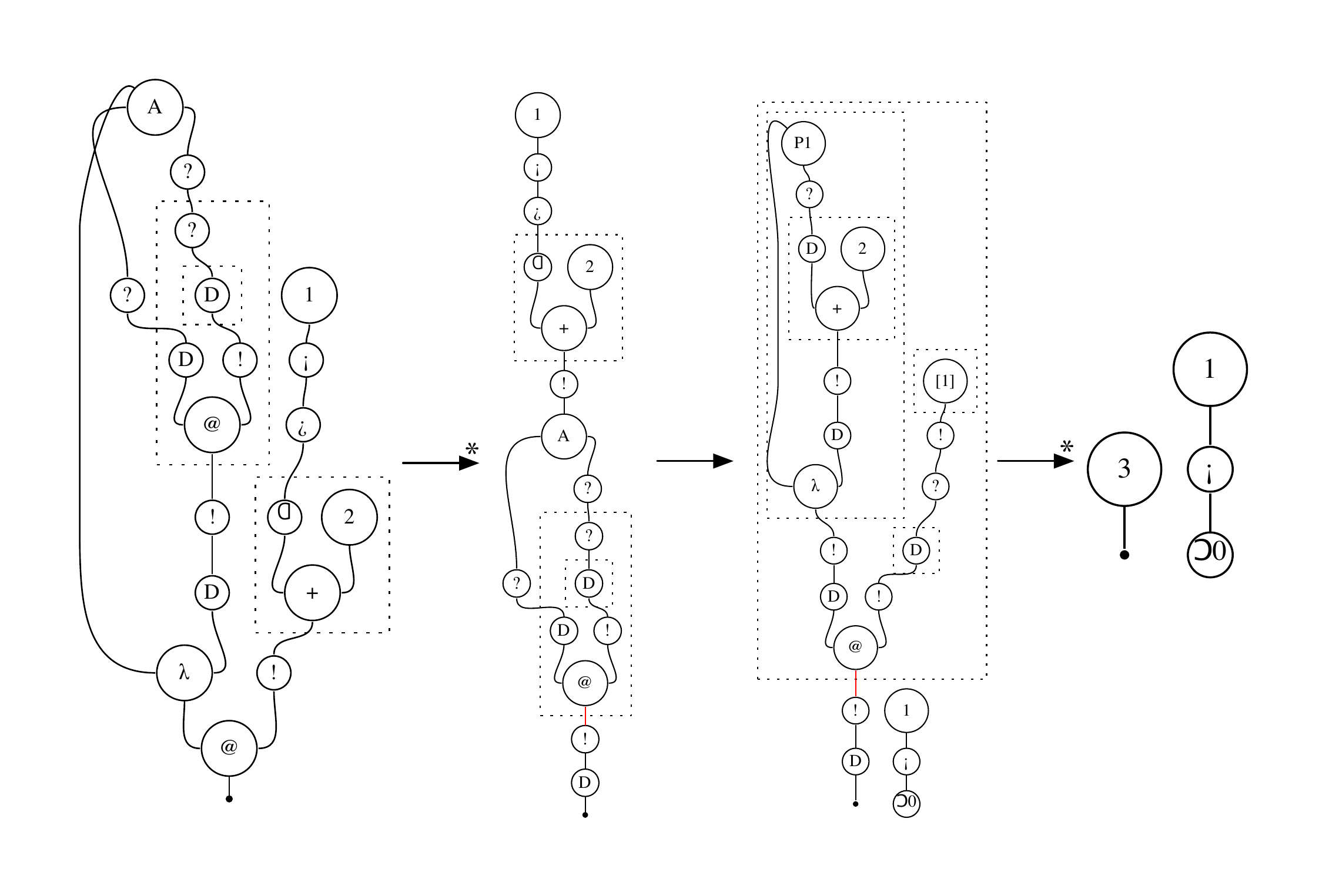}
	\caption{Simple abduction}
	\label{fig:figex1}
\end{figure} 
The first diagram represents the initial graph, the second and third just before and after decoupling, and the fourth is the final value. Note that the diagram still includes the provisional constant of the original term, because of the linearity requirement. We will discuss this in Sec.~\ref{sec:lin}.

\subsubsection{Deep decoupling}
The second example was meant to illustrate the fact that decoupling is indeed a semantic rather than syntactic operation, which is applied to graphs constructed through evaluation:
\begin{align*}
&\letin {y} {\{2\} + 1} { }\\[-1ex]
&\letin {m\, x} {\{3\} + y + x}{ }\\[-1ex]
&\letin{f\,@\,p}m{ }\\[-1ex]
&f\, p\, 7.
\end{align*}
The key stages of execution (initial, just before decoupling, just after decoupling, final) are shown in Fig.~\ref{fig:figex2}.
The provisional constants are highlighted in red and the  $A$-node in blue. We can see how, initially, the two provisional constants belong to distinct sub-graphs of the program, but are brought together during execution and are decoupled together. 
\begin{figure}
	\centering
	\includegraphics[scale=0.6]{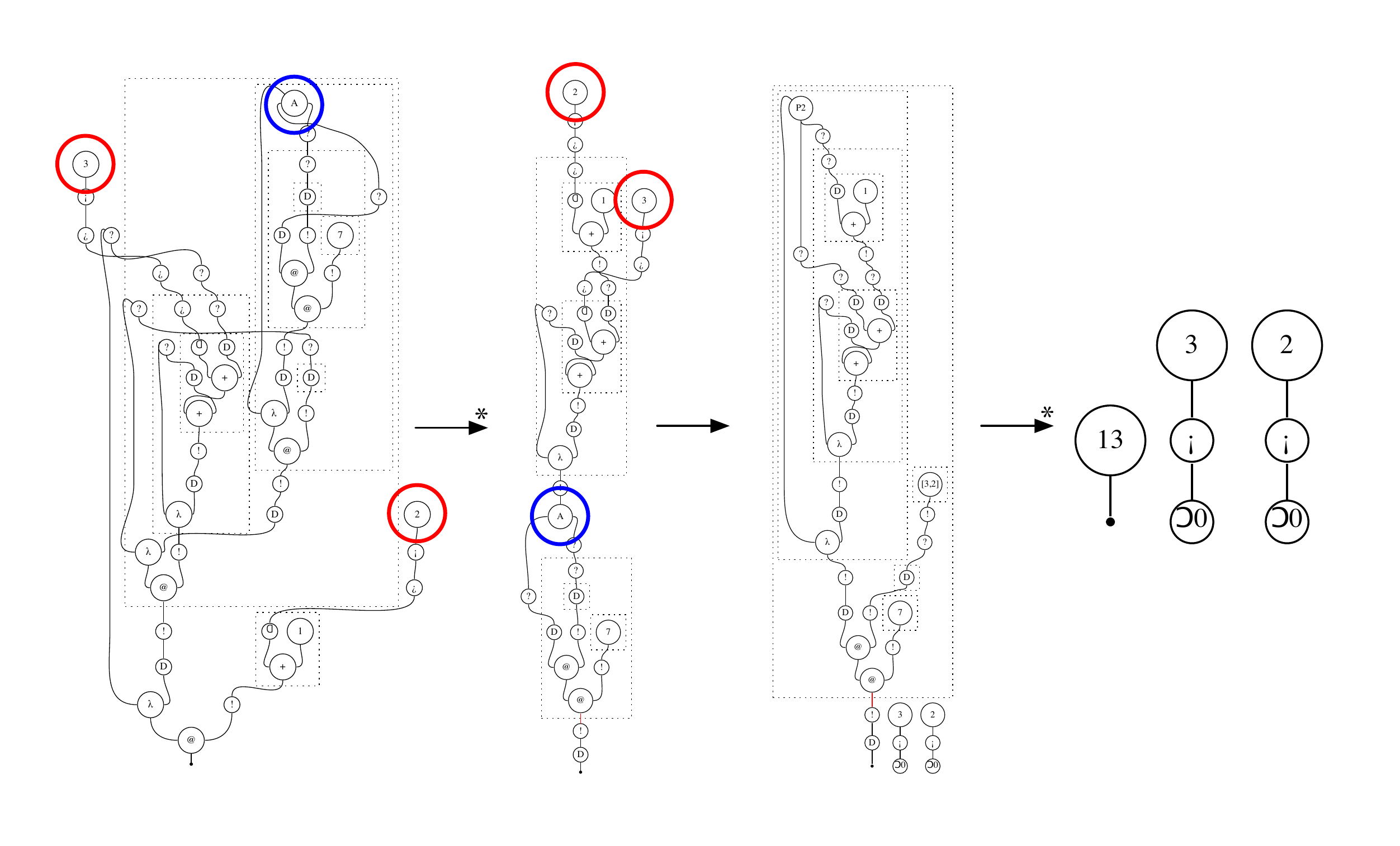}
	\caption{Deep decoupling}
	\label{fig:figex2}
\end{figure}

\subsubsection{Linearity of provisional constants}\label{sec:lin}

In this section we illustrate with two very simple example why linearity of provisional constants is important. 

Consider $t=(\lambda x.x+x)\{1\}$ versus  $t'=\{1\}+\{1\}$. When evaluated in direct mode, these terms produce the same value. However, consider the way they are evaluated in the abductive context $\letin{f\,@\,p}{\square}{p\bullet p}$, as seen in Fig.~\ref{fig:figex3}. The same four key stages in execution are illustrated for both example. In the case of $t$ we can see how the single provisional constant becomes shared (via the $\rot C$-node) while the resulting model has only one parameter. On the other hand, $t'$ has two parameters, resulting in a model with two arguments. In both cases the models are discarded (the $C_0$-node) and only the parameter vector is processed via dot product -- resulting in two different final values. Note that the non-accessible nodes (``garbage'') are greyed out. 
\begin{figure}
	\centering
	\includegraphics[scale=0.5]{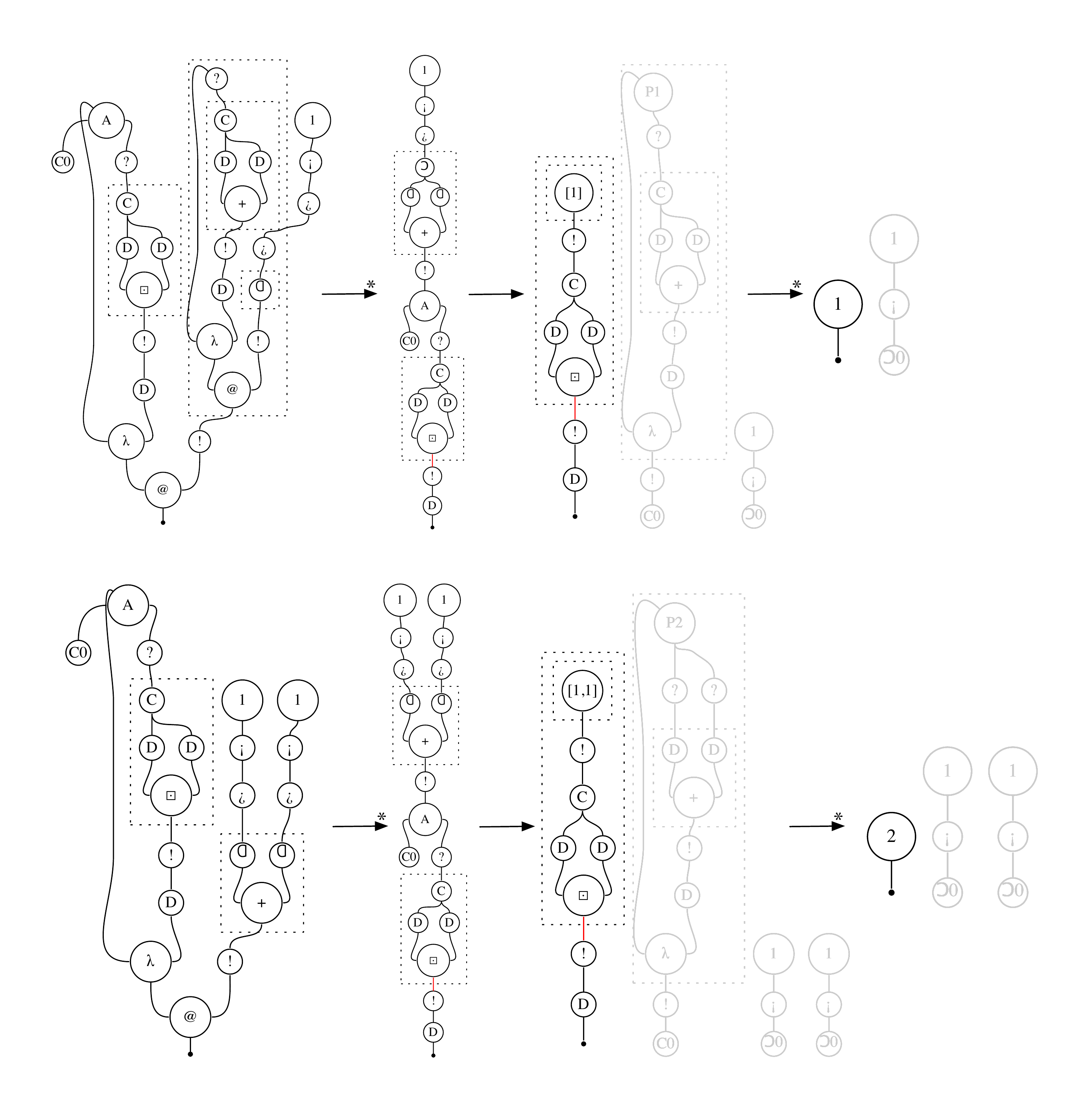}
	\caption{Contraction of provisional constants}
	\label{fig:figex3}
\end{figure}

Similarly, considering $t=(\lambda x.0)\{1\}$ versus $t'=0$ in the same context $\letin{f\,@\,p}{\square}{p\bullet p}$ shows observable differences between the two because of the abductable provisional constant in $t$.

\subsubsection{Learning and meta-learning}
The main motivation of our abduction calculus is supporting learning through parameter tuning. In the on-line visualiser we provide a full example of learning a linear regression model via gradient descent. 

Beyond learning, the native support for parameterised constants also makes it easy to express so-called ``meta-learning'', where the learning process itself is parameterised~\cite{vilalta2002perspective}. For example the rate of change in gradient descent or the rate of mutation in a genetic algorithm can also be tuned by observing and improving the behaviour of the algorithm in concrete scenarios. A stripped-down example of ``meta-learning'' is $\letin{g\,@\,q}{(\letin{f\,@\,p}{\{1\}}{f(\{2\}\times p)})}{g\, q}$, 
because learning (after decoupling) mode uses tunable parameters which are themselves subsequently decoupled. The inner let is where learning happens, whereas the outer let indicates the ``meta-learning.'' Fig.~\ref{fig:figex5} shows the initial graph, before-and-after the first decoupling, before-and-after the second decoupling, and the final results. 

One meta-algorithm which is particularly interesting and widely used is ``adaptive boosting''  (\textit{AdaBoost})~\cite{freund1997decision}, and it is also programmable using the abductive calculus in a less bureaucratic way. A typical boosting algorithm uses ``weak'' learning algorithms combined into a weighted sum that represents the final output of the boosted classifier. A typical implementation of adaptive boosting using abduction is as follows:
\begin{align*}
	&\letin {\textit{model}\,x} {\ldots} { }
	& \text{initial model}\\[-1ex]
	&\letin {\textit{learn}'} {\ldots} { }
	& \text{some learning algorithm}\\[-1ex]
	&\letin {\textit{learn}''} {\ldots} { }
	& \text{different learning algorithm}\\[-1ex]
	&\letin{f\,@\,p}{\textit{model}}{ }
	& \text{abductive decoupling}\\[-1ex]
	&\letin{p'}{\textit{learn}'\,f\,p}{}
	& \text{tune default parameters}\\[-1ex]
	&\letin{p''}{\textit{learn}''\,f\,p'}{}
	& \text{tune new parameters}\\[-1ex]
	&\letin{\textit{model}'}{f\,p'}{}
	& \text{first tuned model}\\[-1ex]
	&\letin{\textit{model}''}{f\,p''}{}
	& \text{second tuned model}\\[-1ex]
	&\letin{\textit{boosted}\,x}{(\textit{model}'\,x+\textit{model}''\,x)/2{}}
	& \text{boosted aggregated model}\\[-1ex]
	& \ldots & \text{main program.}
\end{align*}
\begin{figure}
	\centering
	\includegraphics[scale=0.5]{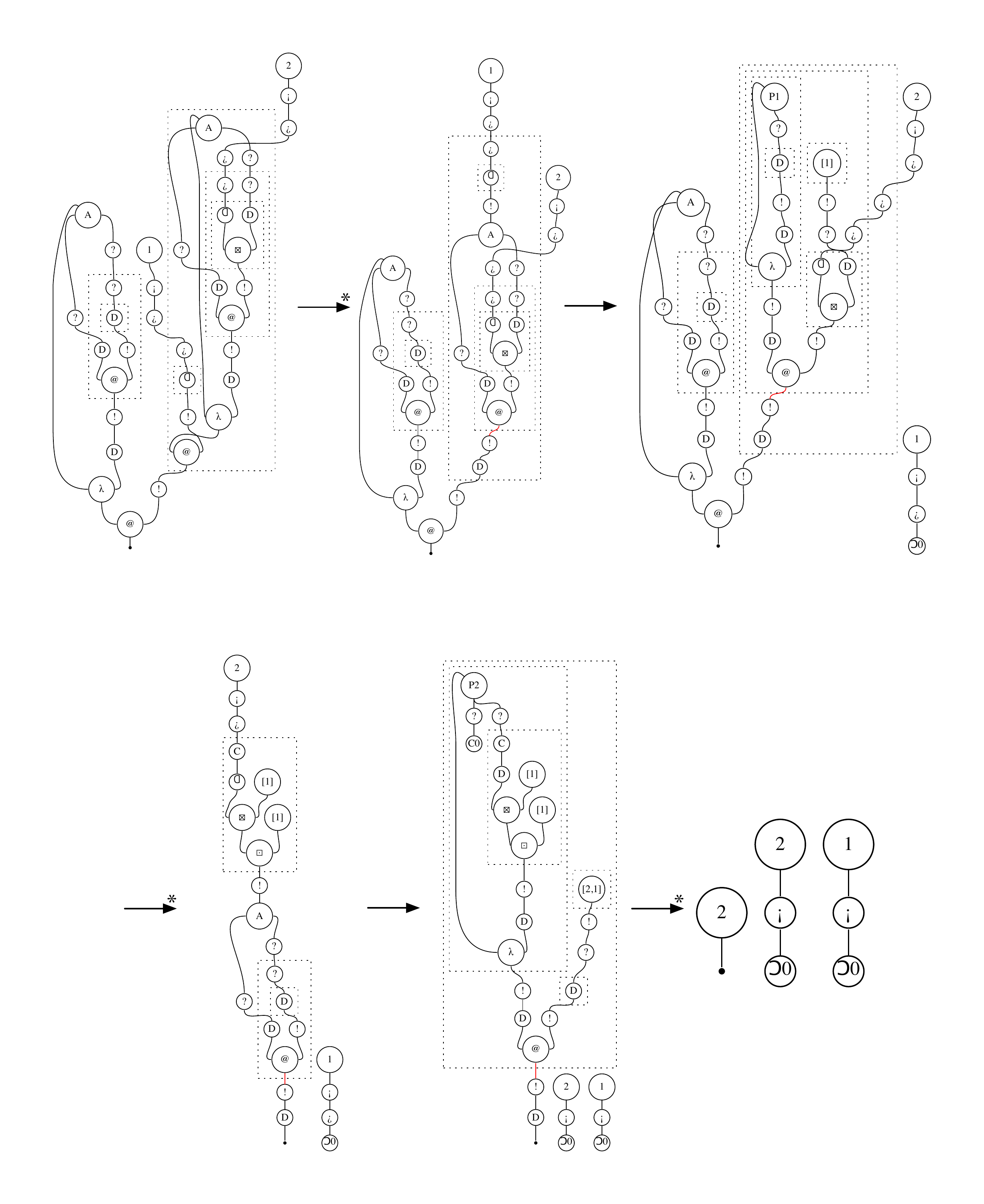}
	\caption{''Meta-learning''}
	\label{fig:figex5}
\end{figure}

\section{Correctness}
\subsection{Soundness}
\label{sec:Soundness}

The main technical result of this paper is soundness, which expresses the fact that well typed programs terminate correctly, which means they do not crash and do not run forever. The challenge is, as expected, dealing with the complex rewriting rule used to model abductive decoupling. 

\begin{theorem}[Soundness]
	For any closed program $t$ such that 
	$- \mid - \mid \vec{p} \vdash t : T$, there exist a graph $G$ and an
	element $X$ of a computation stack such that:
	\begin{align*}
	\Init((- \mid - \mid \vec{p} \vdash t : T)^\ddag) \to^* \Final(G,X).
	\end{align*}
\end{theorem}

In our semantics, the execution involves either a token moving through the graph, or rewrites to the graph. Above, $G$ is the final shape of the graph at the end of the execution, and $X$ is a part of the token data as it ``exits'' the graph $G$. $X$ will always be either a scalar, or a vector, or the symbol $\lambda$ indicating a function-value result. The graph $G$ will involve  the provisional constants in the vector $\vec p$, which are not reduced during execution.

The proof is given in  Appendix~\ref{app:Soundness}.

\subsection{Garbage collection}\label{sec:gc}

Large programs generate subgraphs which are, in a technical sense, unreachable during normal execution, i.e. \textit{garbage}. In the presence of decoupling the precise definition is subtle, and the rules for removing it not immediately obvious. To define garbage collection we first introduce a notion of operational equivalence for graphs, then we show that the rewrite rules corresponding to garbage collection preserve this equivalence.

\section{Garbage collection}
\label{app:GC}

\begin{definition}[Graph equivalence]
	Two definitive graphs $G_1(1,n)$ and $G_2(1,n)$ of ground type are
	\emph{equivalent}, written $G_1 \equiv G_2$, if for any vector
	$\vec{q} \in \F^n$, there exists an element $X$ of a computation
	stack such that the following are equivalent:
	$\Init(G_1 \circ (\vec{p})^\ddag)
	\to^* \Final(G'_1 \circ (\vec{p})^\ddag,X)$
	for some definitive graph $G'_1$, and
	$\Init(G_2 \circ (\vec{p})^\ddag)
	\to^* \Final(G'_2 \circ (\vec{p})^\ddag,X)$
	for some definitive graph $G'_2$. 
\end{definition}
\begin{definition}[Graph-contextual equivalence]
	Two graphs $G_1(n,m)$ and $G_2(n,m)$ are \emph{contextually
		equivalent}, written $G_1 \cong G_2$, if for any graph context
	$\mathcal{G}[\square]$ that is itself a definitive graph of ground
	type, $\mathcal{G}[G_1] \equiv \mathcal{G}[G_2]$ holds.
\end{definition}

The graph equivalence $\equiv$ and the graph-contextual equivalence
$\cong$ are indeed equivalence relations.
Our interest here is what binary relation $R$ on graphs implies
(equivalently, be included by) the graph-contextual equivalence
$\cong$.
\begin{definition}[Lifting]
	Given a binary relation $R$ on graphs of the same interface, its
	\emph{lifting} $\lift{R}$ is a binary relation between states
	defined by:
	$((\mathcal{G}[G_1],\ell),\delta)
	\lift{R} ((\mathcal{G}[G_2],\ell),\delta)$
	where $G_1 \mathrel{R} G_2$, and the position $\ell$ is in the graph
	context $\mathcal{G}[\square]$.
\end{definition}
\noindent
We use the reflexive and transitive closure of a lifting $\lift{R}$,
denoted by $\liftRT{R}$, to deal with duplication of sub-graphs.
\begin{lemma}
	\label{lem:InitLifting}
	If two composite graphs can be decomposed as
	$\mathcal{G}[G_1]$ and $\mathcal{G}[G_2]$ such that
	$G_1 \mathrel{R} G_2$ for some binary relation $R$, initial
	states on them satisfy
	$\Init(\mathcal{G}[G_1]) \lift{R} \Init(\mathcal{G}[G_2])$.
	\qed
\end{lemma}
\begin{proposition}[Sufficient condition of graph-contextual
	equivalence]
	\label{prop:SufficientCondGraphCtxtEquiv}
	A binary relation $R$ on graphs satisfies
	$R \subseteq \mathit{\cong}$, if
	$\Init(H_1) \liftRT{R}(H_2)$ implies the existence of an element $X$
	of a computation stack such that the following are equivalent:
	$\Init(H_1) \to^* \Final(H'_1,X)$ for some graph $H'_1$, and
	$\Init(H_2) \to^* \Final(H'_2,X)$ for some graph $H'_2$.
\end{proposition}
\begin{proof}
	Assume $G_1 \mathrel{R} G_2$.
	By Lem.~\ref{lem:InitLifting}, for any graph context
	$\mathcal{G}[\square](1,n)$ which is itself a definitive graph of
	ground type and any vector $\vec{p} \in \F^n$, we have
	$\Init(\mathcal{G}[G_1] \circ (\vec{p})^\ddag)
	\lift{R} \Init(\mathcal{G}[G_1] \circ (\vec{p})^\ddag)$.
	Therefore by assumption, we have
	$\mathcal{G}[G_1] \equiv \mathcal{G}[G_2]$, and hence $G_1 \cong G_2$.
\end{proof}

The graph-contextual equivalence ensures safety of some forms of
garbage collection, as proved below.
\begin{proposition}[Garbage collection]
	Let $\succ_1$, $\succ_2$ and $\succ_3$ be binary relations on graphs,
	defined by:
	\begin{center}
		\includegraphics[]{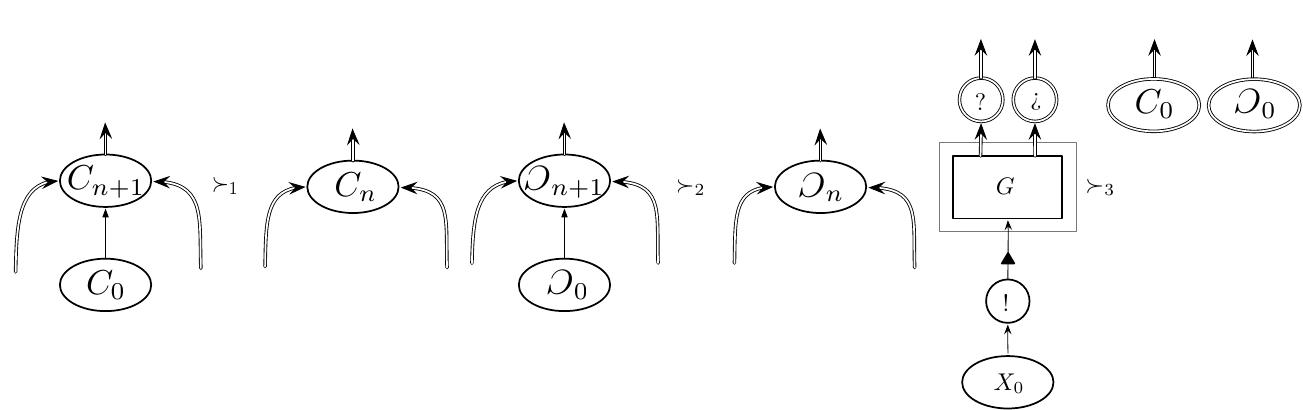}
	\end{center}
	where the $X_0$-node is either a $\lb{C}_0$-node or a
	$\lb{P}_0$-node.
	They altogher imply the graph-contextual equivalence, i.e.\
	$\mathit{(\succ_1 \cup \succ_2 \cup \succ_3)}
	\enspace\subseteq\enspace \mathit{\cong}$.
\end{proposition}
\begin{proof}[Sketch of proof]
	The lifting $\liftRT{\succ_1 \cup \succ_2 \cup \succ_3}$ in fact
	gives a bisimulation.
	Taking the reflexive and transitive closure $\liftRT{}$ primarily
	deals with duplication.
	Taking union of three binary relations $\succ_1$, $\succ_2$ and
	$\succ_3$ is important, because each of them does not lift to a
	bisimulation on its own.
	The decoupling rule turns a $\rot{\lb{C}}$-node to a $\lb{C}$-node,
	which means $\succ_2$ depends on $\succ_1$.
	The deep contraction rule may generate a $\oc$-box whose
	principal door is connected to a $\lb{C}_0$-node, which means
	$\succ_1$ depends on $\succ_3$, and further on $\succ_2$ and
	$\succ_1$ itself, via transitivity.
\end{proof}
\noindent
\textit{Discussion.}
	Congruence of graph equivalence $\equiv$ is more subtle than
	one might expect, due to decoupling.
	Consider the abductive decoupling rule, applied to graphs $\mathcal{G}[G_1]$ and
	$\mathcal{G}[G_2]$ such that $G_1 \equiv G_2$.
	If  graphs $G_1$ and $G_2$ are in the redex of the decoupling rule,
	their output type $\rot{\oc}\F$ is changed to $\oc\F$, which means the
	definition of equivalence does not apply to the graphs subsequent decoupling.
	\begin{conjecture}[Congruence of graph equivalence]
		Graph equivalence $\equiv$ implies  graph-contextual
		equivalence $\cong$, in other words, it is a congruence.
		Formally, for any graphs $G_1$ and $G_2$,
		$G_1 \equiv G_2$ implies $G_1 \cong G_2$.
	\end{conjecture}

\subsection{Program equivalence}
The usual way of equating programs if they produce the same value is not applicable in contexts with decoupling since it can observe differences between, for example $\{1\}+2$, $1+2$ or $1+\{2\}$. 
However, the notion of graph equivalence introduced above is generally appropriate to also define program equivalence. 

Programs, as usual, are closed ground-type terms, and we note that parameters of programs can be permuted once the graphs have been computed. This is not a semantic rule, but only a top-level transformation:
\begin{center}
	\includegraphics[]{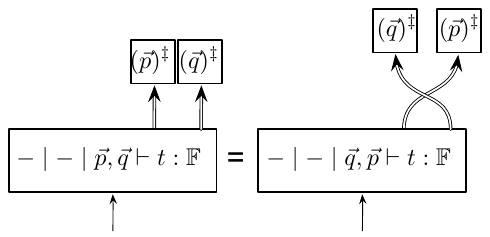}
\end{center}
\begin{definition}[Program equivalence]
	Two programs $(- \mid - \mid \vec{p_i} \vdash t_i : \F)$ are said to be equivalent, written as
	$
	(- \mid - \mid \vec{p_0} \vdash t_0 : \F)
 \dot\approx (- \mid - \mid \vec{p_1} \vdash t_1 : \F)
 $,
 iff there exists a vector $\vec p$ and a permutation $\sigma_i$
 such that
 $\sigma_0\cdot \vec p_0=\sigma_1\cdot \vec p_1=\vec p$, and if
 $(- \mid - \mid \vec{p_i} \vdash t_i)^\ddagger
 = H_i\circ (\vec p)^\ddagger$,
 then $H_0\equiv H_1$.
\end{definition}
This definition can be lifted to open terms in the usual way.
\begin{definition}[Term equivalence]\label{def:equiv}
 $(A \mid \Gamma \mid \vec{q_0} \vdash t : T')
 \approx (A \mid \Gamma \mid \vec{q_1} \vdash u : T')$ iff,
 for any context
 $- \mid - \mid \vec{p},\vec{r} \vdash C\emptyCtxt^{T'} : T$, we have that
 $(- \mid - \mid \vec{p},\vec{q},\vec{r} \vdash \plug{C}{t_0} : T)
 \dot\approx (- \mid - \mid \vec{p},\vec{q},\vec{r}
 \vdash \plug{C}{t_1} : T).$
\end{definition}

\section{Related and further work}

\subsection{Machine learning}

Our belief that there is a significant role for transparent parameterisation of programs in some areas of programming, in particular machine learning, is inspired primarily by \textsc{TensorFlow}~\cite{abadi2016tensorflow}, which already exhibits some of the programming structures we propose. It has support for explicit \textit{learning modes} for its models and it introduces a notion of \textit{variable} which corresponds to our \textit{provisional constants}, so that in learning mode  variables are \textit{implicitly} tuned via training. However, \textsc{TensorFlow} is not a stand-alone programming language but a shallow embedding of a DSL into \textsc{Python} (noting that other language bindings also exist). Our initial aim was to consider it as a proper (functional) programming language by developing a simple and uniform syntax. 

However, what we propose has evolved beyond \textsc{TensorFlow} in several ways. The most significant difference is that \textsc{TensorFlow} only supports gradient descent tuning, by providing built-in support for computing gradients of models via automatic differentiation~\cite{rall1981automatic}. From the point of view of efficiency this is ideal, but it has several drawbacks. It prevents full integration with a normal programming language because computation in the model must be restricted to operations representing derivable functions. We take a black-box approach to models which is less efficient, but fully compositional. It allows for any numerical algorithm to be used for tuning, not just gradient descent, but also simulated annealing or combinatorial optimisations. We also support meta-learning in a way that \textsc{TensorFlow} cannot. 

Semantically, the idea of building a \textit{computational graph} is also present in \textsc{TensorFlow}. The key difference is that our computational graph, implicit in the GoI-style semantics, evolve during computation. This is again potentially less efficient, although efficient compilation of dynamic GoI is an area of ongoing research. However, the efficiency of the functional infrastructure for abduction is dominated by the vector operations, which may involve very large amounts of data. The rules, expressed in the unfoldings in  Fig.~\ref{fig:RewriteTransitionsComput}, can be ``factored out'' of the abductive calculus to a secondary, special-purpose, and  efficient device dedicated to vector operations, in the same way as \textsc{TensorFlow} constructs the model in \textsc{Python} but the heavy-duty computations can be farmed out to GPU back-ends. 

Currently most programming for machine learning is done either in general-purpose languages or in DSLs such as \textsc{Matlab} or R. There is a growing body of work dedicated to programming languages, toolkits and libraries for machine learning. Most of them are not directly relevant to our work as they focus mainly on computational efficiency especially via parallelism. This is an extremely important practical concern, and our vector computation rules can be easily parallelised in practice by using different unfoldings than the sequential ones we use, noting that efficient parallel computation in the GoI semantics requires a more complex, multi-token machine~\cite{dal2014geometry}.

\subsection{Abduction}

Despite being one of the three pillars of inferential reasoning, abduction has been far less influential than deduction and induction as a source of methodological inspiration for programming languages. Abduction has been used as a source of inspiration in logical programming~\cite{kakas1998role} and in verification~\cite{DBLP:journals/jacm/CalcagnoDOY11}. Somewhat related to verification is an interesting perspective that abduction can shed on type inference~\cite{sulzmann2008type}. However, the concept of abduction as a manifestation of the principle of ``discovery of best explanations'' is a powerful one, and our use of ``runtime abduction'' is only a first step towards developing and controlling more expressive or more efficient concepts of programming abduction. Our core calculus is meant to open a new perspective more than providing a definitive solution.

Historically, the relation between abduction and Bayesian inference has been a subject of much discussion among philosophers and logicians in the theory of confirmation. Bayesian inference is established as the dominant methodology, but recently authors have argued that there is a false dichotomy between the two~\cite[Ch.~7]{lipton2003inference} and that the explanatory power of abduction can complement the quantitative Bayesian analysis. Philosophical considerations aside, our hope is that in programming languages abductive decoupling and probabilistic programming for parameter tuning can be combined. Indeed, the theory of probabilistic programming is by now a highly developed reseach area~\cite{carpenter2016stan,gordon2014probabilistic,vajda2014probabilistic} and no striking incompatibilities exist between such languages and abductive decoupling.  

\subsection{Geometry of Interaction}

For the authors, the GoI style was instrumental in making a very complex operational semantics tractable (and implementable). We think that the GoI style semantics can be illuminating for other programming paradigms in which data-flow models are constructed and manipulated, such as self-adjusting computation~\cite{acar2009self} or functional reactive programming~\cite{wan2000functional}. This is part of a larger, on-going, programme of research. Implementation of programming languages from GoI-style semantics is a highly relevant area of research~\cite{Mackie95,FernandezM02} as are parallel GoI machines. It remains to be seen whether such implementation techniques are efficient enough to support an, otherwise highly desirable, semantics-directed compilation or whether completely different approaches are required, such as leveraging more powerful features that imperative programming languages offer. The \textsc{Incremental} library for OCaml\footnote{\url{https://github.com/janestreet/incremental}} is an example of the latter approach.  

But even as a specification formalism only, the GoI style seems to be both expressive and tractable for highly complex semantics. Even though we have introduced a notion of term equivalence in Def.~\ref{def:equiv} it seems more promising  to use equivalence of graphs directly and define program optimisation strategies directly on the graphs, in the style of~\cite{GonthierAL92}. This also remains a subject of further work. 


\bibliographystyle{abbrv}

\bibliography{ref}

\newpage
\appendix

\section{Determinism}
\label{app:Determinism}

The only sources of non-determinism are the choice of fresh names in
replicating a $\oc$-box and the choice of $\wn$-rewrite transitions
(Fig.~\ref{fig:RewriteTransitionsDeep} and
Fig.~\ref{fig:RewriteTransitionsOpenBox}).
Introduction of fresh names has no impact on execution, as we can
prove ``alpha-equivalence'' of graph states.
\begin{proposition}[''alpha-equivalence'' of graph states]
 \label{prop:AlphaEquivBisim}
 The binary relation $\sim_\alpha$ of two graph states, defined by
 $((G,\ell),\delta) \sim_\alpha ((\pi \cdot G,\ell), \delta)$ for any
 name permutation $\pi$, is an equivalence relation and a
 bisimulation.
\end{proposition}
\begin{proof}
 Only rewrite transitions that replicate a $\oc$-box (in
 Fig.~\ref{fig:RewriteTransitionsDeep} and
 Fig.~\ref{fig:RewriteTransitionsClosedBox}) involve name permutation.
 Names are irrelevant in all the other transitions.
\end{proof}

We identify graph states modulo name permutation, namely the binary
relation $\sim_\alpha$ in the above proposition.
Now non-determinism boils down to the choice of $\wn$-rewrites, which
however does not yield non-deterministic overall executions.
\begin{proposition}[Determinism]
 \label{prop:TransitionsDeterminism}
 If there exists a sequence
 $((G,\ell),\delta) \to^* ((G',\ell'),(d',\square,S',B'))$, any
 sequence of transitions from the state $((G,\ell),\delta)$ reaches
 the state $((G',\ell'),(d',\square,S',B'))$, up to name permutation.
\end{proposition}
\begin{proof}
 The applicability condition of $\wn$-rewrite rules ensures that
 possible $\wn$-rewrites at a state do not share any redexes.
 Therefore $\wn$-rewrites are confluent, satisfying the so-called
 diamond property:
 if two different $\wn$-rewrites
 $((G,\ell),\delta) \to ((G_1,\ell_1),\delta_1)$ and
 $((G,\ell),\delta) \to ((G_2,\ell_2),\delta_2)$ and
 are possible from a single state, both of the data $\delta_1$ and
 $\delta_2$ has rewriting flag $\wn$, and there exists a state
 $((G',\ell'),\delta')$ such that
 $((G_1,\ell_1),\delta_1) \to ((G',\ell'),\delta')$ and
 $((G_2,\ell_2),\delta_2) \to ((G',\ell'),\delta')$.
\end{proof}
\begin{corollary}[Prop.~\ref{prop:Determinism}]
 \label{prop:ExecDeterminism}
 For any initial state $\Init(G)$, the final state $\Final(G,X)$ such
 that $\Init(G) \to^* \Final(G,X)$ is unique up to name permutation,
 if it exists.
\end{corollary}

\section{Validity}
\label{app:Validity}

This section investigates a property of graph states, \emph{validity},
which plays a key role in disproving any failure of transitions.
It is based on three criteria on graphs.

In the lambda-calculus one often assumes that bound variables in a
term are distinct, using the alpha-equivalence, so that beta-reduction
does not cause unintended variable capturing.
We start with an analogous criterion on names.
\begin{definition}[Bound/free names]
 \label{def:BoundFreeNames}
 A name $a \in \A$ in a graph is said to be:
 \begin{enumerate}
  \item \emph{bound} by an $\lb{A}$-node, if the $\lb{A}$-node has
	input types $V_a \to T)$ and $\oc V_a$, for some type $T$.
  \item \emph{free}, if a $\val{\vec{p}}$-node has input type $V_a$ or
	a $\lb{P}_n$-node has output type $V_a$.
 \end{enumerate}
\end{definition}
\begin{definition}[Bound-name criterion]
 \label{def:BoundNameCriterion}
 A graph $G$ meets the \emph{bound-name criterion} if any bound name
 $a \in \A$ in the graph $G$ satisfies the following.
 \begin{description}
  \item[Uniqueness.]
	     The name $a$ is not free, and is bound by exactly one
	     $\lb{A}$-node.
  \item[Scope.]
	     Bound names do not appear in types of input links of the
	     graph $G$.
	     Moreover, if the $\lb{A}$-node that binds the name $a$ is
	     in a $\oc$-box, the name $a$ appears only strictly inside
	     the $\oc$-box (i.e.\ in the $\oc$-box, but not on its
	     interfaces).
 \end{description}
\end{definition}

The name permutation action accompanying rewrite transitions
(Fig.~\ref{fig:RewriteTransitionsDeep} and
Fig.~\ref{fig:RewriteTransitionsClosedBox}) is an explicit way to
preserve the above requirement in transitions.
\begin{proposition}[Preservation of bound-name criterion]
 \label{prop:PreserveBoundNameCriterion}
 In any transition, if an old state meets the bound-name criterion, so
 does a new state.
\end{proposition} 
\begin{proof}
 In a $\wn$-rewrite transition that eliminates a $\lb{A}$-node, the
 name $a \in \A$ bound by the $\lb{A}$-node turns free.
 As the name $a$ is not bound by any other $\lb{A}$-nodes, it does not
 stay bound after the transition.
 The transition does not change the status of any other names, and
 therefore preserves the uniqueness and scope of bound variables.

 Duplication of a $\oc$-box, in a rewrite transition involving a
 $\lb{C}_n$-node or a $\lb{P}_n$-node applies name permutation.
 The scope of bound names is preserved by the transition, because
 if an $\lb{A}$-node is duplicated, all links in which the name bound
 by the $\lb{A}$-node appears are duplicated together.
 The scope also ensures that, if an $\lb{A}$-node is copied,
 the name permutation makes each copy of the node bind distinct names.
 Therefore the uniqueness of bound names is not broken by the
 transition.

 Any other transitions do not change the status of names.
\end{proof}

The second criterion is on free names, which ensures each free name
indicates a unique vector space $\F^n$.
\begin{definition}[Free-name criterion]
 \label{def:FreeNameCriterion}
 A graph $G$ meets the \emph{free-name criterion} if it comes with a
 ``validation'' map $v \colon \mathrm{FR}_G \to \N$, from
 the set $\mathrm{FR}_G$ of free names in the graph $G$ to the set
 $\N$ of natural numbers, that satisfies the following.
 \begin{itemize}
  \item If a $\val{\vec{p}}$-node has input type $V_a$, the
	vector $\vec{p}$ has the size $v(a)$,
	i.e.\ $\vec{p} \in \F^{v(a)}$
  \item If a $\lb{P}_n$-node has output type $\oc V_a$, it has $v(a)$
	input links, i.e.\ $n = v(a)$.
 \end{itemize}
\end{definition}
The validation map is unique by definition.
We refer to the combination of the bound-name criterion and the
free-name criterion as ``name criteria.''
\begin{proposition}[Preservation of name criteria]
 \label{prop:PreserveNameCriteria}
 In any transition, if an old state meets both the bound-name
 criterion and the free-name criterion, so does a new state.
\end{proposition}
\begin{proof}
 With Prop.~\ref{prop:PreserveBoundNameCriterion} at hand, we here
 show that the new state fulfills the free-name criterion.

 A free name is introduced by a $\wn$-rewrite transition that
 eliminates a $\lb{A}$-node.
 The name was bound by the $\lb{A}$-node and not free before the
 transition, because of the bound-name criterion (namely the
 uniqueness property).
 Therefore the validation map can be safely extended.

 The name permutation, in rewrite transitions that duplicate a
 $\oc$-box, applies for both bound names and free names.
 It introduces fresh free names, without changing the status of names,
 and therefore the validation map can be extended accordingly.

 Some computational rewrite rules
 (Fig.~\ref{fig:RewriteTransitionsComput}) act on links with vector
 type $V_a$, however they have no impact on the validation map.
 Any other transitions also do not affect the validation map.
\end{proof}

The last criterion is on the shape of graphs.
It is inspired by Danos and Regnier's correctness criterion
\cite{DanosR89} for proof nets.
\begin{definition}[Covering links]
 In a graph $G(1,n)$, a link $\ell$ is \emph{covered} by another link
 $\ell'$, if any box-path (see Def.~\ref{def:BoxReachability}) from
 the root of the graph $G$ to the link $\ell$ contains the covering
 link $\ell'$.
\end{definition}
\begin{definition}[Graph criterion]
 \label{def:GraphCriterion}
 A graph $G(1,n)$ fulfills the \emph{graph criterion} if it satisfies
 the following.
 \begin{description}
  \item[Acyclicity] Any box-path, in which all links have (not
	     necessarily the same) argument types, is acyclic, i.e.\
	     nodes or links appear in the box-path at most once.
	     Similarly, any directed path whose all links have the
	     provisional type $\rot{\oc} \F$ is acyclic.
  \item[Covering] At any $\lambda$-node, its incoming output link is
	     covered by its input link.
	     Any $\lb{A}$-node or $\lb{P}$-node is covered by a
	     $\wn$-node.
 \end{description}
\end{definition}
\begin{proposition}[Preservation of graph criterion]
 \label{prop:PreserveGraphCriterion}
 In any transition, if an old state meets the graph criterion, so does
 a new state.
\end{proposition}
\begin{proof}
 An $@$-rewrite transition eliminates a pair of a $\lambda$-node and
 an $@$-node, and connects two acyclic box-paths of argument types.
 The resulting box-path being a cycle means that there
 existed a box-path from the free (i.e.\ not
 connected to the $\lambda$-node) output link of the $@$-node to the
 incoming output link of the $\lambda$-node before the transition.
 This cannot be the case, as the incoming output link must have been
 covered by the input link of the $\lambda$-node.
 Therefore the $@$-rewrite does not break the acyclicity condition.
 The condition can be easily checked in any other transitions.

 The covering condition is also preserved.
 Only notable case for this condition is the decoupling rule that
 introduces a $\lambda$-node and a $\lb{P}$-node.
\end{proof}

Finally the validity of graph states is defined as below.
The validation map of a graph is used to check if the token carries
appropriate data to make computation happen.
\begin{definition}[Queries and answers]
 Let $d \colon A \to \N$ be a map from a finite set
 $A \finsubset \A$ of names to the set $\N$ of natural numbers.
 For each type $\tilde{T}$, two sets $\Qry_{\tilde{T}}$ and
 $\Ans_{\tilde{T}}^d$ are defined inductively as below.
 \begin{align*}
  \Qry_\F = \Qry_{\rot\oc  \F} &:= \{ \star \},
  & \Ans_\F^d = \Ans_{\rot\oc  \F}^d &:= \F \\
  \Qry_{V_a} &:= \{ \star \},
  & \Ans_{V_a}^d &:=
  \begin{cases}
   \F^{d(a)} & \text{(if $a \in A$)}\\
   \emptyset & \text{(otherwise)}
  \end{cases} \\
  \Qry_{T_1 \to T_2} &:= \{ \star, @ \},
  & \Ans_{T_1 \to T_2}^d &:= \{ \lambda \} \\
  \Qry_{\oc T} &:= \Qry_T,
  & \Ans_{\oc T}^d &:= \Ans_T^d.
 \end{align*}
\end{definition}
\begin{definition}[Valid states]
 \label{def:ValidStates}
 A state $((G,\ell),(d,f,S,B))$ is \emph{valid} if the following
 holds.
 \begin{enumerate}
  \item The graph $G$ fulfills the name criteria and the graph
	criterion.
  \item If $d = \up$ and the position $\ell$ has type $\rho$, the
	computation stack $S$ is in the form of $X \cl S'$ such that
	$X \in \Qry_\rho$.
  \item Let $v$ be the validation map of the graph $G$.
	If $d = \dn$ and the position $\ell$ has type $\rho$, the
	set $\Ans_\rho^v$ is not empty, and the
	computation stack $S$ is in the form of	$X \cl S'$ such that
	$X \in \Ans_\rho^v$.
 \end{enumerate}
\end{definition}
\begin{proposition}[Preservation of validity]
 \label{prop:PreserveValidity}
 In any transition, if an old state is valid, so is a new state.
\end{proposition}
\begin{proof}
 Using Prop.~\ref{prop:PreserveNameCriteria} and
 Prop.~\ref{prop:PreserveGraphCriterion}, the proof boils down to
 check the bottom two conditions of validity.
 Note that no rewrite transitions change the direction and the
 computation stack.
 When the token passes a $\$^0$-node downwards, application of the
 primitive operation $\$^0$ preserves the last condition of validity.
 All the other pass transitions are an easy case.
\end{proof}

In an execution, validity of intermidiate states can be reduced to the
criteria on its initial graph.
\begin{proposition}[Validity condition of executions]
 \label{prop:ExecValidityCondition}
 For any execution $\Init(G_0) \to^* ((G,\ell),\delta)$, if the
 initial graph $G_0$ meets the name criteria and the graph criterion,
 the state $((G,\ell),\delta)$ is valid.
\end{proposition}
\begin{proof}
 The initial state $\Init(G_0)$ has the direction $\up$, and its
 computation stack has the top element $\star$.
 Since any type $\rho$ satisfies $\star \in \Qry_\rho$, the criteria
 implies validity at the initial state $\Init(G_0)$.
 Therefore the property is a consequence of
 Prop.~\ref{prop:PreserveValidity}.
\end{proof}

\section{Stability}
\label{app:Stability}

This section studies executions in which the underlying graph is never
changed.
\begin{definition}[Stable executions/states]
 An execution $\Init(G) \to^* ((G,\ell),\delta)$ is \emph{stable} if
 the graph $G$ is never changed in the execution.
 A state is \emph{stable} is there exists a stable execution to the
 state itself.
\end{definition}
A stable execution can include pass transitions, and rewrite
transitions that just lower the rewrite flag, as well.
Since the only source of non-determinism is rewrite transitions that
actually change a graph, a stable state comes with a unique stable
execution to the state itself.

The stability property enables us to backtrack an execution in certain
ways, as stated below.
\begin{proposition}[Factorisation of stable executions]
 \label{prop:FactoriseStableExec}
 \noindent
 \begin{enumerate}
  \item If an execution $\Init(G) \to^* ((G,\ell),\delta)$ is stable,
	it can be factorised as
	$\Init(G) \to^* ((G,\ell'),\delta') \to^* ((G,\ell),\delta)$
	where the link $\ell'$ is any link covering the link $\ell$.
  \item If an execution
	$\Init(G) \to^* ((G,\ell),(\dn,\square,X \cl S,B))$ is stable,
	it can be factorised as
	$\Init(G) \to^* ((G,\ell),(\up,\square,\star \cl S,B))
	\to^* ((G,\ell),(\dn,\square,X \cl S,B))$.
 \end{enumerate}
\end{proposition}
\begin{proof}[Proof of Prop.~\ref{prop:FactoriseStableExec}.1]
 The proof is by induction on the length of the stable execution
 $\Init(G) \to^* ((G,\ell),\delta)$.
 When the execution has null length, the last position $\ell$ is the
 root of the graph $G$, and the only link that can cover it is the
 root itself.

 When the execution has a positive length, we examine each possible
 transition.
 Rewrite transitions that only lower the rewriting flag are trivial
 cases.
 Cases for pass transitions are the straightforward use of induction
 hypothesis, because for any link and a node, the following are
 equivalent: (i) the link covers one of outgoing output links of the
 node, and (ii) the link covers all input links of the node.
\end{proof}
\begin{proof}[Proof of Prop.~\ref{prop:FactoriseStableExec}.2]
 The proof is by induction on the length $n$ of the stable execution
 $\Init(G) \to^n ((G,\ell),(\dn,\square,X \cl S,B))$.
 
 As the first state and the last state cannot be equal, base cases are
 for single transitions, i.e.\ when $n = 1$.
 Only possibilities are pass transitions over a $\lambda$-node, a
 $\val{p}$-node or a $\val{\vec{p}}$-node, all of which is in the form
 of $((G,\ell), (\up,\square,\star \cl S,B)) \to
 ((G,\ell), (\dn,\square,X \cl S,B))$.

 In inductive cases, we will use induction hypothesis for any
 length that is less than $n$.
 If the last transition is a pass transition over a $\lambda$-node, a
 $\val{p}$-node or a $\val{\vec{p}}$-node, the discussion goes in the
 same way as in base cases.
 All the other possible last transitions are: pass transitions over a
 node labelled with $\oc$, $\rot{\oc}$, $\rot{\wn}$,
 $\rot{\lb{D}}$ or $\rot{\lb{C}}_n$; and rewrite transitions that do
 not change the underlying graph but discard the rewriting flag
 $\$^0$.

 If the last transition is a pass transition over a $Z$-node such that
 $Z \in \{ \oc, \rot{\oc}, \rot{\wn},
 \rot{\lb{D}}, \rot{\lb{C}}_n \}$,
 the last position (referred to as $\In(Z)$) is input to the
 $Z$-node, and the second last position (referred to as $\Out(Z)$) is
 output of the $Z$-node.
 Induction hypothesis (on $n-1$) implies the factorisation below,
 where $n = m + l + 1$:
 \begin{align*}
  \Init(G)
  &\to^m ((G,\Out(Z)), (\up,\square,\star \cl S,B')) \\
  &\to^l ((G,\Out(Z)), (\dn,\square,X \cl S,B'))
  \to ((G,\In(Z)), (\dn,\square,X \cl S,B)).
 \end{align*}
 Moreover the state $((G,\Out(Z)),(\up,\square,\star \cl S,B'))$
 must be the result of a pass transition over the $Z$-node.
 This means we have the following further factorisation if
 $Z \neq \oc$,
 \begin{align*}
  \Init(G)
  &\to^{m-1} ((G,\In(Z)), (\up,\square,\star \cl S,B))
  \to ((G,\Out(Z)), (\up,\square,\star \cl S,B')) \\
  &\to^l ((G,\Out(Z)), (\dn,\square,X \cl S,B'))
  \to ((G,\In(Z)), (\dn,\square,X \cl S,B))
 \end{align*}
 and the one below if $Z = \oc$.
 \begin{align*}
  \Init(G)
  &\to^{m-3} ((G,\In(Z)), (\up,\square,\star \cl S,B))
  \to ((G,\Out(Z)), (\up,\wn,\star \cl S,B)) \\
  &\to ((G,\Out(Z)), (\up,\oc,\star \cl S,B))
  \to ((G,\Out(Z)), (\up,\square,\star \cl S,B')) \\
  &\to^l ((G,\Out(Z)), (\dn,\square,X \cl S,B'))
  \to ((G,\In(Z)), (\dn,\square,X \cl S,B))
 \end{align*}

 If the last transition is a rewrite transition that discards the
 rewriting flag $\$^0$, it must follow a pass transition over a
 $\$^0$-node.
 Let $\In$, $\Out_1$ and $\Out_2$ denote input, left output and right
 output, respectively, of the $\$^0$-node.
 We obtain the following factorisation where $n = m + l_2 + l_1 + 3$,
 using induction hypothesis twice (on $n-2$ and $n-l_1-3$).
 \begin{align*}
  \Init(G)
  &\to^{m-1} ((G,\In), (\up,\square,\star \cl S,B))
  \to ((G,\Out_2), (\up,\square,\star \cl \star \cl S,B)) \\
  &\to^{l_2} ((G,\Out_2), (\dn,\square,\val{k_2} \cl \star \cl S,B))
  \to ((G,\Out_1),
  (\up,\square,\star \cl \val{k_2} \cl \star \cl S,B)) \\
  &\to^{l_1} ((G,\Out_1),
  (\dn,\square,\val{k_1} \cl \val{k_2} \cl \star \cl S,B)) \\
  &\to ((G,\In), (\dn,\$^0,\val{k_1 \$^0 k_2} \cl S,B))
  \to ((G,\In), (\dn,\square,\val{k_1 \$^0 k_2} \cl S,B))
 \end{align*}
\end{proof}

Inspecting the proof of Prop.~\ref{prop:FactoriseStableExec}.2 gives
some (graphically-)intensional characterisation of graphs in stable
executions.
We say a transition ``involves'' a node, if it is a pass transition
over the node or it is a rewrite transition whose (main-)redex
contains the node.
\begin{proposition}[Stable executions, intensionally]
 \label{prop:StableExecIntensionally}
 Any stable execution of the form
 $\Init(G) \to^h ((G,\ell),(\up,\square,\star \cl S,B))
 \to^k ((G,\ell),(\dn,\square,X \cl S,B))$
 satisfies the following.
 \begin{itemize}
  \item If the position $\ell$ has a ground type or the provisional
	type $\rot{\oc}\F$, the last $k$ transitions
	of the stable execution involve nodes labelled with only
	$\{ p,\vec{p},\$^0, \rot{\oc},\rot{\lb{D}},\rot{\lb{C}}_m
	\mid p \in \F,\ \vec{p} \in \F^n,\ \$^0 \in \Sigma,\ 
	n \in \N,\ m \in \N \}$.
  \item If the position $\ell$ has a function type,
	i.e.\ $T_1 \to T_2$, it is the input of a $\lambda$-node, and
	$k = 1$.
 \end{itemize}
\end{proposition}
\begin{proof}
 The proof is by looking at how factorisation is given in the proof of
 Prop.~\ref{prop:FactoriseStableExec}.2.
 Note that, since we are ruling out argument types, i.e.\ enriched
 types of the form of $\oc T$, the factorisation never encounters
 $\oc$-nodes (hence nor $\rot{\wn}$-nodes).
\end{proof}

The fundamental result is that stability of states is preserved by any
transitions.
This means, in particular, rewrites triggerd by the token in an
execution can be applied beforehand to the initial graph without
changing the end result.
Another (rather intuitive) insight is that, in an execution, the token
leaves no redexes behind it.
\begin{proposition}[Preservation of stability]
 \label{prop:PreserveStability}
 In any transition, if an old state is stable, so is a new state.
\end{proposition}
\begin{proof}
 If the transition does not change the underlying graph, it clearly
 preserves stability.
 If not, the preservation is a direct consequence of
 Lem.~\ref{lem:StableExecInCtxt} and
 Lem.~\ref{lem:StabiliseActualRewrites} below.
\end{proof}
\begin{lemma}[Stable executions in graph context]
 \label{lem:StableExecInCtxt}
 If all positions in a stable execution
 $\Init(\mathcal{G}[G]) \to^* ((\mathcal{G}[G],\ell),\delta)$ are in
 the graph context $\mathcal{G}$, there exists a stable execution
 $\Init(\mathcal{G}[G']) \to^* ((\mathcal{G}[G'],\ell),\delta)$ for
 any graph $G'$ with the same interfaces as the graph $G$.
\end{lemma}
\begin{proof}
 The proof is by induction on the length of the stable
 execution
 $\Init(\mathcal{G}[G]) \to^* ((\mathcal{G}[G],\ell),\delta)$.
 The base case for null length is trivial.
 Inductive cases are respectively for all possible last transitions.
 When the last transition is a pass transition, the single node
 involved by the transition must be in the graph context
 $\mathcal{G}[\square]$.
 Therefore the last transition is still possible when the graph $G$ is
 replaced, which enables the straightforward use of induction
 hypothesis.

 When the last transition is a ``stable'' rewrite transition which
 simply changes the rewriting flag $f$ to $\square$, we need to
 inspect its redex.
 Whereas a part of the redex may not be inside the graph context
 $\mathcal{G}[\square]$,
 we confirm below that the same last transition is possible for any
 substitution of the hole, by case analysis of the rewriting flag $f$.
 Once this is established, the proof boils down to the straightforward
 use of induction hypothesis.
 Possible rewriting flags are primitive operations $\$^0$, and symbols
 $\wn$ and $\oc$ for $\oc$-box rewrites.

 If the rewriting flag is $\$^0$, the redex consists of one
 $\$^0$-nodes with two nodes connected to its output.
 The flag must have been raised by a pass transition over the
 $\$^0$-node, which means the $\$^0$-node is in the graph context
 $\mathcal{G}[\square]$.
 Moreover, by Lem.~\ref{prop:FactoriseStableExec}.2, the two other
 nodes in the redex are also in the graph context
 $\mathcal{G}[\square]$.
 Therefore the stable rewrite transition, for the flag $\$^0$, is not
 affected by substitution of the hole.

 If the rewriting flag is $\wn$, the redex is a $\oc$-box with all its
 doors.
 Since the rewriting flag must have been raised by the pass transition
 over the principal door, the principal has to be in the graph context
 $\mathcal{G}[\square]$.
 All the auxiliary doors of the same $\oc$-box are also in the graph
 context $\mathcal{G}[\square]$, by definition of graphs.
 The stable rewrite transition for the flag $\wn$ is hence possible,
 regardless of any substitution of the hole, while the $\oc$-box
 itself may be affected by the substitution.
 If the rewriting flag is $\oc$, the redex is a $\oc$-box, all its
 doors, and a node connected to its principal door.
 This case is similar to the last case.
 The connected node, to the principal door, is also in the graph
 context because the token must have visited the node before passing
 the principal door.
\end{proof}
\begin{lemma}[Stabilisation of actual rewrites]
 \label{lem:StabiliseActualRewrites}
 Let
 $((\mathcal{G}[G],\ell),\delta)
 \to ((\mathcal{G}[G'],\ell'),\delta')$
 be a rewrite transition, where $G$ is the redex replaced with a
 different graph $G'$.
 If the rewrite transition follows the stable execution
 $\Init(\mathcal{G}[G]) \to^* ((\mathcal{G}[G],\ell),\delta)$,
 there exist an input link $\ell_0$ of the hole $\square$
 (equivalently of $G$ and $G'$) and token data $\delta_0$ such that:
 \begin{itemize}
  \item the stable execution can be factorised as
	$\Init(\mathcal{G}[G])
	\to^* ((\mathcal{G}[G],\ell_0),\delta_0)
	\to^* ((\mathcal{G}[G],\ell),\delta)$,
	where all positions in the first half sequence are in the
	graph context $\mathcal{G}[\square]$
  \item there exists a sequence
	$((\mathcal{G}[G'],\ell_0),\delta_0)
	\to^* ((\mathcal{G}[G'],\ell'),\delta')$
	in which the graph $\mathcal{G}[G']$ is never changed.
 \end{itemize}
\end{lemma}
\begin{proof}
 The proof is by case analysis of the rewriting flag of the data
 $\delta$.
 Note that we only look at rewrites that actually change the graph.

 When the rewriting flag $f$ is $\lambda$, the redex contains a
 connected pair of an $@$-node and a $\lambda$-node.
 We represent the outgoing output of the $\lambda$-node by
 $\Out(\lambda)$, one output of the $@$-node connected to the
 $\lambda$-node by $\In(\lambda)$, the other output of the
 $@$-node by $\Out(@)$, and the input of the $@$-node by $\In(@)$.
 Lem.~\ref{prop:FactoriseStableExec}.2 implies that the stable
 execution
 $\Init(\mathcal{G}[G]) \to^* ((\mathcal{G}[G],\ell),\delta)$
 can be factorised as below.
 \begin{align*}
  \Init(\mathcal{G}[G])
  &\to^* ((\mathcal{G}[G],\In(@)),(\up,\square,S,B)) \\
  &\to ((\mathcal{G}[G],\Out(@)),(\up,\square,\star \cl S,B)) \\
  &\to^* ((\mathcal{G}[G],\Out(@)),(\dn,\square,X \cl S,B)) \\
  &\to ((\mathcal{G}[G],\In(\lambda))),(\up,\square,@ \cl S,B) \\
  &\to ((\mathcal{G}[G],\Out(\lambda)),(\up,\lambda,S,B))
 \end{align*}
 The four links $\Out(\lambda)$, $\In(\lambda)$, $\Out(@)$ and
 $\In(@)$ cannot happen in the stable prefix execution
 $\Init(\mathcal{G}[G])
 \to^* ((\mathcal{G}[G],\In(@)),(\up,\square,S,B))$,
 except for the last state, otherwise the rewriting flag $\lambda$
 must have been raised in this execution, causing the change of the
 graph.
 The other link in the redex, the incoming output of the
 $\lambda$-node, neither appears in the prefix execution, as no pass
 transition is possible at the link.
 Therefore the prefix execution contains only links in the graph
 context $\mathcal{G}[\square]$, and we can take $\ell_0$ as $\In(@)$
 and $(\up,\square,S,B)$ as $\delta_0$.
 The rewrite yields the state
 $((\mathcal{G}[G'],\ell'),(\up,\square,S,B))
 = ((\mathcal{G}[G'],\ell_0),\delta_0)$.

 When the rewriting flag $f$ is $\$^0$, the redex is a $\$^0$-node
 with two constant nodes ($\val{k_1}$ and $\val{k_2}$) connected.
 Let $\In(\$^0)$, $\In(k_1)$ and $\In(k_2)$ denote the unique input of
 these three nodes, respectively.
 The stable execution
 $\Init(\mathcal{G}[G]) \to^* ((\mathcal{G}[G],\ell),\delta)$ is
 actually in the following form.
 \begin{align*}
  \Init(\mathcal{G}[G])
  &\to^* ((\mathcal{G}[G],\In(\$^0)),(\up,\square,\star \cl S',B')) \\
  &\to ((\mathcal{G}[G],\In(k_2)),
  (\up,\square,\star \cl \star \cl S',B')) \\
  &\to ((\mathcal{G}[G],\In(k_2)),
  (\dn,\square,\val{k_2} \cl \star \cl S',B')) \\
  &\to ((\mathcal{G}[G],\In(k_1)),
  (\up,\square,\star \cl \val{k_2} \cl \star \cl S',B')) \\
  &\to ((\mathcal{G}[G],\In(k_1)),
  (\dn,\square,\val{k_1} \cl \val{k_2} \cl \star \cl S',B')) \\
  &\to ((\mathcal{G}[G],\In(\$^0)),
  (\dn,\$^0,\val{k_1 \$^0 k_2} \cl S',B))
 \end{align*}
 The links $\In(\$^0)$, $\In(k_1)$ and $\In(k_2)$ cannot appear in the
 stable prefix execution
 $\Init(\mathcal{G}[G])
 \to^* ((\mathcal{G}[G],\In(\$^0)),(\up,\square,\star \cl S',B'))$,
 except for the last state,
 otherwise the rewriting flag $\$^0$ must have been raised and have
 triggered the change of the graph.
 As the links $\In(k_1)$ and $\In(k_2)$ are the only ones outside the
 graph context $\mathcal{G}[G]$ and the link $\In(\$^0)$ is input of
 the redex $G$, the prefix execution is entirely in the graph context
 $\mathcal{G}[\square]$.
 We can take $\In(\$^0)$ as $\ell_0$ and
 $(\up,\square,\star \cl S',B')$ as $\delta_0$.
 The rewrite of the redex does not change the position, which means
 $\ell_0 = \ell' = \In(\$^0)$.
 The resulting graph $G'$ consists of one constant node
 ($\val{k_1 \$^0 k_2}$), and we have a single transition
 $((\mathcal{G}[G'],\ell_0)),(\up,\square,\star \cl S',B'))
 \to ((\mathcal{G}[G'],\ell)),
 (\dn,\square,\val{k_1 \$^0 k_2} \cl S',B))$
 to the result state of the rewrite.

 When the rewriting flag is $\$^1$, the redex is a $\$^1$-node with
 one node connected to one of its output links.
 By Lem.~\ref{prop:FactoriseStableExec}.2, the stable execution
 $\Init(\mathcal{G}[G]) \to^* ((\mathcal{G}[G],\ell),\delta)$
 is in the form of:
 \begin{align*}
  \Init(\mathcal{G}[G])
  &\to^* ((\mathcal{G}[G],\In),(\up,\square,\star \cl S',B')) \\
  &\to ((\mathcal{G}[G],\Out_2),
  (\up,\square,\star \cl \star \cl S',B')) \\
  &\to ((\mathcal{G}[G],\Out_2),
  (\dn,\square,X_2 \cl \star \cl S',B')) \\
  &\to ((\mathcal{G}[G],\Out_1),
  (\up,\square,\star \cl X_2 \cl \star \cl S',B')) \\
  &\to ((\mathcal{G}[G],\Out_1),
  (\dn,\square,X_1 \cl X_2 \cl \star \cl S',B')) \\
  &\to ((\mathcal{G}[G],\In),
  (\up,\$^1(n),\star \cl S',B))
 \end{align*}
 where $\In$, $\Out_2$ and $\Out_1$ denote the input, the right output
 and the left output of the $\$^1$-node, respectively.
 These three links are the only links in the redex, and they do not
 appear in the prefix execution
 $\Init(\mathcal{G}[G])
 \to^* ((\mathcal{G}[G],\In),(\up,\square,\star \cl S',B'))$
 except for the last.
 The last state of the prefix execution has the same token position
 and token data as the result of the rewrite.
 
 The rewriting flag $\wn$ is raised by a pass transition over a
 $\oc$-node, principal door of a $\oc$-box.
 The pass transition is the last one of the stable execution
 $\Init(\mathcal{G}[G]) \to^* ((\mathcal{G}[G],\ell),\delta)$, i.e.\
 \begin{align*}
  \Init(\mathcal{G}[G])
  &\to^* ((\mathcal{G}[G],\In),\delta)
  \to ((\mathcal{G}[G],\Out),\delta)
 \end{align*}
 where $\In$ and $\Out$ are respectively the input and output of the
 $\oc$-node.
 Since any $\wn$-rewrite leaves the $\oc$-node in place and keeps the
 position and data of the token, we have a pass transition
 $((\mathcal{G}[G'],\In),\delta)
 \to ((\mathcal{G}[G'],\Out),\delta)$
 to the resulting state of the rewrite.
 It remains to be seen whether the stable prefix execution
 $\Init(\mathcal{G}[G]) \to^* ((\mathcal{G}[G],\In),\delta)$ is
 entirely in the graph context $\mathcal{G}[\square]$.

 First, the links $\In$ and $\Out$ cannot appear in the prefix
 execution except for the last, otherwise there must have been a
 non-stable $\wn$-rewrite.
 When the rewriting flag $\wn$ triggers the contraction rule
 (bottom-right in Fig.~\ref{fig:RewriteTransitionsDeep}) or the
 absorption rule (left in Fig.~\ref{fig:RewriteTransitionsOpenBox}),
 any links in the redex are covered by the link $\In$, by definition
 of graphs.
 Therefore by Lem.~\ref{prop:FactoriseStableExec}.1, these links
 neither appear in the prefix execution.
 When the projection rule (top-right in
 Fig.~\ref{fig:RewriteTransitionsDeep}) occurs, the interface links of
 the $\lb{P}_n$-node do not appear in the prefix execution, as there
 is no pass transition over the $\lb{P}_n$-node.
 Since the input link of the $\lb{P}_n$-node covers all the other
 links in the redex, by Lem.~\ref{prop:FactoriseStableExec}.1, no
 links in the redex have been visited by the token.
 The case of decoupling rule (left in
 Fig.~\ref{fig:RewriteTransitionsDeep}) is similar to the projection
 case.
 Recall that the redex for the decoupling rule excludes the sub-graph
 ($G \circ (\vec{p})^\ddag$ in the figure) that stays the same.
 In decoupling case, all links in the redex do not appear in the
 prefix execution, while the unchanged sub-graph is included by the
 graph context $\mathcal{G}[\square]$ by assumption.

 Finally for the rewriting flag $\oc$, the stable execution
 $\Init(\mathcal{G}[G]) \to^* ((\mathcal{G}[G],\ell),\delta)$
 ends with several pass transitions, including one over a $\oc$-node,
 and a rewrite transition that sets the flag:
 \begin{align*}
  \Init(\mathcal{G}[G])
  &\to^* ((\mathcal{G}[G],\ell_0),\delta_0)
  \to^* ((\mathcal{G}[G],\In),(\up,\square,S,B)) \\
  &\to ((\mathcal{G}[G],\Out),(\up,\wn,S,B))
  \to ((\mathcal{G}[G],\Out),(\up,\oc,S,B))
 \end{align*}
 where $\ell_0$ is an input link of the redex $G$, and
 $\In$ and $\Out$ are respectively the input and output of the
 $\oc$-node.
 Inspecting each $\oc$-rewrites yields a stable sequence
 $((\mathcal{G}[G'],\ell_0),\delta_0)
 \to^* ((\mathcal{G}[G'],\Out),(\up,\oc,S,B))$
 to the result state of the rewrite.
 The inspection also confirms that the stable prefix execution
 $\Init(\mathcal{G}[G]) \to^* ((\mathcal{G}[G],\ell_0),\delta_0)$ is
 entirely in the graph context $\mathcal{G}$, as below.

 In rewrites in the first row of
 Fig.~\ref{fig:RewriteTransitionsClosedBox}, the
 link $\ell_0$ is the only input of the redex and it covers the whole
 redex.
 As the link $\ell_0$ cannot appear in the stable prefix execution,
 neither any link in the redex, by
 Lem.~\ref{prop:FactoriseStableExec}.1.
 In the other rewrites (bottom row of
 Fig.~\ref{fig:RewriteTransitionsClosedBox}), the input link $\ell_0$
 of the redex covers the whole redex except for the other input links.
 The uncovered input links, in fact, must have not been visited by the
 token, otherwise the token has proceeded to a $\oc$-node and
 triggered copying.
\end{proof}

Since stability is trivial for initial states, we can always assume
stability at any states in an execution.
\begin{proposition}[Stability of executions]
 \label{prop:ExecStability}
 In any execution $\Init(G_0) \to^* ((G,\ell),\delta)$, the state
 $((G,\ell),\delta)$ is stable.
\end{proposition}
\begin{proof}
 This is a consequence of Prop.~\ref{prop:PreserveStability}, since
 any initial states are trivially stable.
\end{proof}

\section{Productivity and safe termination}
\label{app:Productivity}

By assuming both validity and stability, we can prove
\emph{productivity}: namely, a transition is always possible at a
valid and stable intermediate state.
\begin{proposition}[Productivity]
 \label{prop:Productivity}
 If a state is valid, stable and not final, there exists a possible
 transition from the state.
\end{proposition}
\begin{proof}[Proof in Sec.~\ref{app:proof:Productivity}]
\end{proof}

We can obtain a sufficient condition for the safe termination of an
execution, which is satisfied by the translation of any program.
\begin{proposition}[Safe termination]
 \label{prop:SafeTermination}
 Let $\Init(G_0)$ be an initial state whose graph $G_0$ meets the name
 criteria and the graph criterion.
 If an execution $\Init(G_0) \to^* ((G,\ell),\delta)$ can be followed
 by no transition, the last state $((G,\ell),\delta)$ is a final
 state.
\end{proposition}
\begin{proof}
 This is a direct consequence of
 Prop.~\ref{prop:ExecValidityCondition},
 Prop.~\ref{prop:ExecStability} and Prop.~\ref{prop:Productivity}.
\end{proof}
\begin{proposition}[Safe termination of programs]
 For any closed program $t$ such that
 $- \mid - \mid \vec{p} \vdash t : T$, if an execution on the
 translation $(- \mid - \mid \vec{p} \vdash t : T)^\ddag$ can be
 followed by no transition, the last state of the execution is a final
 state.
\end{proposition}
\begin{proof}
 The translation $(- \mid - \mid \vec{p} \vdash t : T)^\ddag$ fulfills
 the name criteria and the graph criterion, which can be checked
 inductively.
 Note that all names in the translation are bound.
 This proposition is hence a consequence of
 Prop.~\ref{prop:SafeTermination}.
\end{proof}

\subsection{Proof of Prop.~\ref{prop:Productivity}}
\label{app:proof:Productivity}
 First we assume that the state has rewriting flag $\square$.
 Failure of pass transitions can be caused by either of the following
 situations: (i) the position is eligible but the token data is not
 appropriate, or (ii) the position is not eligible.

 The situation (i) is due to the wrong top elements of a
 computation/box stack.
 In most cases, it is due to the single top element of the
 computation stack, or top elements of the box stack, which can be
 desproved easily by validity, or respectively, stability.
 The exception is when the token points downwards at the left output
 of a primitive operation node ($\$$), and the top three elements of
 the computation stack have to be checked.
 Let $\In$, $\Out_1$ and $\Out_2$ denote the input, the left output
 and the right output of the $\$$-node, respectively.
 By stability and Lem.~\ref{prop:FactoriseStableExec}.2, the state is
 the last state of the following stable execution.
 \begin{align*}
  \Init(G)
  \to^* ((G,\Out_1), (\up,\square,\star \cl S,B))
  &\to^* ((G,\Out_1), (\dn,\square,X_1 \cl S,B))
 \end{align*}
 The intermediate state has to be the result of a pass transition,
 i.e.\
 \begin{align*}
  \Init(G)
  \to^* ((G,\Out_2), (\dn,\square,S,B))
  \to ((G,\Out_1), (\up,\square,\star \cl S,B))
  &\to^* ((G,\Out_1), (\dn,\square,X_1 \cl S,B)).
 \end{align*}
 Since the last state is valid, the graph $G$ fulfills the criteria
 (Def.~\ref{def:BoundNameCriterion}, Def.~\ref{def:FreeNameCriterion}
 and Def.~\ref{def:GraphCriterion}) and any states in this
 execution is valid by Prop.~\ref{prop:ExecValidityCondition}.
 Therefore the computation stack $S$ is in the form of
 $S = X_2 \cl S'$, and using Lem~\ref{prop:FactoriseStableExec}.2
 again yields:
 \begin{align*}
  \Init(G)
  &\to^* ((G,\Out_2), (\up,\square,\star \cl S',B)) \\
  &\to^* ((G,\Out_2), (\dn,\square,X_2 \cl S',B))
  \to ((G,\Out_1), (\up,\square,\star \cl X_2 \cl S',B)) \\
  &\to^* ((G,\Out_1), (\dn,\square,X_1 \cl X_2 \cl S',B)).
 \end{align*}
 The first intermediate state, again, has to be the result of a pass
 transition, i.e.\
 \begin{align*}
  \Init(G)
  &\to^* ((G,\In), (\up,\square,S',B))
  \to ((G,\Out_2), (\up,\square,\star \cl S',B)) \\
  &\to^* ((G,\Out_2), (\dn,\square,X_2 \cl S',B))
  \to ((G,\Out_1),
  (\up,\square,\star \cl X_2 \cl S',B)) \\
  &\to^* ((G,\Out_1),
  (\dn,\square,X_1 \cl X_2 \cl S',B)).
 \end{align*}
 Since the first intermediate state of the above execution is valid,
 the computation stack $S'$ is in the form of $S' = \star \cl S''$,
 which means $S = X_1 \cl X_2 \cl \star \cl S''$.
 Moreover validity ensures that the elements
 $X_1$ and $X_2$ are eligible for a pass transition from the last
 state; in particular vector operations $+$, $\times$ and $\cdot$
 are always given two vectors of the same size.

 We move on to the situation (ii), where the token position is not
 eligible to pass transitions.
 To disprove this situation, we assume a valid, stable and non-final
 state from which no pass transition is possible, and derive
 contradiction.

 The first case is the state $((G,\ell),(\dn,\square,S,B))$ where the
 position $\ell$ is the incoming output of a $\lambda$-node.
 By the graph criterion, the position is covered by the
 outgoing output $\Out(\lambda)$ of the $\lambda$-node, and
 Lem.~\ref{prop:FactoriseStableExec} implies the following stable
 execution.
 \begin{align*}
  \Init(G) \to^* ((G,\Out(\lambda)),(\up,f,S',B'))
  \to^* ((G,\ell),(\dn,\square,S,B))
 \end{align*}
 Due to stability, the intermediate state
 $((G,\Out(\lambda)),\delta)$ must be the
 result of a pass transition over the $\lambda$-node.
 However the transition sets $\lambda$ as the rewriting flag $f$,
 which triggers elimination of the $\lambda$-node and contradicts
 stability.

 The second case is the state $((G,\ell),(\dn,\square,S,B))$ where the
 position $\ell$ is the left output of an $@$-node.
 By validity the computation stack $S$ is in the form of
 $S = \lambda \cl S'$, and Lem.~\ref{prop:FactoriseStableExec}.2 gives
 the stable execution
 \begin{align*}
  Init(G) \to^* ((G,\ell),(\up,\square,\star \cl S',B))
 \to^* ((G,\ell),(\dn,\square,\lambda \cl S',B))
 \end{align*}
 to the state.
 The only transitions that can yield the intermediate state, at the
 left output of the $@$-node, are rewrite transitions that change the
 graph, which is contradiction.

 The third case is the state $((G,\ell),(\up,\square,S,B))$ where the
 position $\ell$ is the input of a $\wn$-node, an auxiliary door of a
 $\oc$-box.
 Since the link $\ell$ is covered by the root of the $\oc$-box, by
 Lem.~\ref{prop:FactoriseStableExec}.1, the token has visited its
 principal door, i.e.\ $\oc$-node.
 This visit must have raised rewriting flag $\wn$.
 Because of the presence of the $\wn$-node, the flag must have
 triggered a rewrite that eliminates the $\wn$-node, which is
 contradiction.

 The fourth case is when the position $\ell$ is one of the interface
 (i.e.\ either input or output) links of an $\lb{A}$-node or a
 $\lb{P}$-node.
 By the covering condition of the graph criterion, this case reduces
 to the previous case.

 The last case is the state $((G,\Out(X)),(\dn,\square,S,B))$ where
 the position $\Out(X)$ is the output of an $X$-node, for
 $X \in \{ \lb{D}, \lb{C}_n, \wn \mid n \in \N\}$.
 If $X = \wn$, i.e.\ the node is an auxiliary door of a $\oc$-box, the
 position is covered by the root of the $\oc$-box.
 This reduces to the previous case.
 If not, i.e.\  $X \in \{ \lb{D}, \lb{C}_n \mid n \in \N\}$,
 Lem.~\ref{prop:FactoriseStableExec}.2 gives the stable execution
 \begin{align*}
  Init(G) \to^* ((G,\Out(X)),(\up,\square,S,B))
 \to^* ((G,\Out(X)),(\dn,\square,S,B))
 \end{align*}
 to the state.
 The graph criterion (Def.~\ref{def:GraphCriterion}) implies the
 $X$-node belongs to an acyclic box-path of argument types.
 By typing, any maximal acyclic box-path ends with either a
 $\lambda$-node or a $\oc$-node, and the token must visit this node
 in the second half of the stable execution.
 This implies contradiction as follows.
 If the last node of the maximal box-path is a $\lambda$-node, the
 token must visit the incoming output link of the $\lambda$-node, from
 which the token cannot been proceeded.
 If the last node is a $\oc$-node, the token must pass the $\oc$-node
 and trigger rewrites that eliminate the $X$-node.
 
 This completes the first half of the proof, where we assume the state
 has rewriting flag $\square$.
 In the second half, we assume that the state has a rewriting flag
 which is not $\square$, and show the graph of the state contains an
 appropriate redex for the rewriting flag.

 When the flag is $\lambda$, by stability, the token is at the
 outgoing output of a $\lambda$-node which is connected to the left
 output of an $@$-node.
 Since the graph of the state has no output
 (Def.~\ref{def:GraphStates}), both output links of the $\lambda$-node
 are connected to some nodes.
 Therefore the $\lambda$-rewrite is possible.
 We can reason in the same way when the flag is $\$^1$.
 Rewrite transitions for rewrite flags $\$^0$ are exhaustive.

 When the flag is $\oc$, by stability, the state is the result of the
 $\wn$-rewrite that only changes the flag.
 This means the token is at the root of a $\oc$-box with no definitive
 auxiliary doors ($\wn$-nodes).
 Rewrite transitions for flag $\oc$ are exhaustive for the closed
 $\oc$-box.

 When the flag is $\wn$, the token is at the root of a $\oc$-box,
 which we here call ``inhabited $\oc$-box.''
 By typing, output links of definitive auxiliary doors of the
 inhabited $\oc$-box can be connected to $\lb{C}_n$-nodes,
 $\lb{P}_n$-nodes, $\lb{A}$-nodes, $\oc$-nodes, $\wn$-nodes or
 $\lambda$-nodes.
 However $\wn$-nodes and $\lambda$-nodes are not the case, as we see
 below.

 First, we assume that a definitive auxiliary door of the inhabited
 $\oc$-box is connected to another $\wn$-node.
 This means that the inhabited $\oc$-box is inside
 another $\oc$-box, and therefore the token position is covered by the
 root of the outer $\oc$-box.
 By Lem.~\ref{prop:FactoriseStableExec}.1, the token must
 have visited the principal door of the outer $\oc$-box and triggered
 the change of the graph, which contradicts stability.

 Second, we assume that a definitive auxiliary door of the inhabited
 $\oc$-box is connected to the incoming output link of a
 $\lambda$-node.
 This $\lambda$-node cannot be inside the inhabited $\oc$-box, since
 no $\oc$-box has incoming output.
 Clearly there exists a box-path from the token position to the root
 of the $\oc$-tree.
 Therefore the token position, the unique input of the $\oc$-box, is
 covered by the input of the $\lambda$-node; otherwise the graph
 criterion is violated.
 This covering implies that the token must have passed the
 $\lambda$-node upwards and triggered its elimination, by
 Lem.~\ref{prop:FactoriseStableExec}.1, which contradicts stability.

 The last remark for the rewriting flag $\wn$ is about the
 replacement of nodes in the decoupling rule.
 Typing of links ensures that the replacement never fails and produces
 a correct graph.
 In particular the sub-graph $G$ in
 Fig.~\ref{fig:RewriteTransitionsDeep} only consists of
 $\rot{\lb{C}}_n$-nodes and $\rot{\wn}$-nodes.
 Finally, in conclusion, the state with rewriting flag $\wn$ is always
 eligible for at least one of the rules in
 Fig.~\ref{fig:RewriteTransitionsDeep} and
 Fig.~\ref{fig:RewriteTransitionsOpenBox}.
 \qed

\section{Provisional contexts and congruence of execution}
\label{app:ProvisionalCtxtAndCongruence}

To deal with shared provisional constants that arise in an execution,
we introduce another perspective on composite graphs which takes
$\rot{\lb{C}}$-nodes into account.
\begin{definition}[Provisional contexts]
 A graph context of the form
 \begin{center}
  \includegraphics{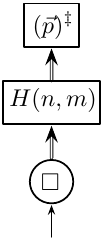}
 \end{center}
 denoted by $\mathcal{P}[\square]$, is a
 \emph{provisional context} if it satisfies the following.
 \begin{itemize}
  \item The graph $H(n,m)$ consists solely of $\rot{\lb{C}}$-nodes,
	and $\vec{p} \in \F^m$.
  \item For any graph $G(1,n)$ that fulfills the graph criterion, the
	graph $\mathcal{P}[G]$ also fulfills the graph criterion.
 \end{itemize}
\end{definition}

In the above definition, the second condition implies that the graph
$H$ contains no loops.
We sometimes write $\mathcal{P}[\square]_{\tilde{T}}^n$ to make
explicit the input type and the number of output links of the hole.
Note that a graph $\mathcal{P}[G]$, where $\mathcal{P}[\square]$ is a
provisional context and $G$ is a definitive graph, is a composite
graph.
As an extension of Prop.~\ref{prop:FormPreservation}, we can see a
provisional context is preserved by transitions.
\begin{proposition}[Provisional context preservation]
 \label{prop:ProvisionalCtxtPreservation}
 When a transition sends a graph $G$ to a graph $G'$, if
 the old graph $G$ can be decomposed as $\mathcal{P}[H]$ where
 $\mathcal{P}[\square]$ is a provisional context and $H(1,n)$ is a
 definitive graph, the new graph $G'$
 can be also decomposed as $\mathcal{P}[H']$ for some definitive graph
 $H'(1,n)$, using the same provisional context.
\end{proposition}
\begin{proof}
 In addition to Prop.~\ref{prop:FormPreservation}, no transition
 changes existing $\rot{\lb{C}}$-nodes.
 Therefore a provisional context is preserved in any transition.
\end{proof}

Since our operational semantics is based on low-level graphical
representation and local token moves, rather than structured
syntactical representation, any structural reasoning requires extra
care.
For example, evaluation of a term of function type is not exactly the
same, depending on whether the term appears in the argument position
or the function position of function application $\app{t}{u}$.
The token distinguishes the evaluation using elements $\star$ and $@$
of a computation stack, which is why we explicitly require termination
in definition of $P_{T_1 \to T_2}$.
Moreover congruence of execution is not trivial.

To prove a specific form of congruence, we begin with ``extracting'' a
provisional context out of a graph context.
Let $\mathcal{G}[\square]$ be a graph context, such that for any
definitive graph $G(1,n)$ of type $T$, the graph $\mathcal{G}[G]$ is
a composite graph.
We can decompose the graph context $\mathcal{G}[\square]$ as
\begin{align*}
 \mathcal{G}[\square]
 =
 \includegraphics{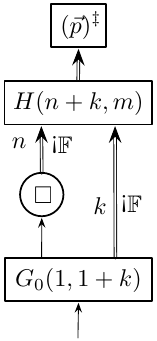}
\end{align*}
where the graph $H(n+k,m)$ consists of all reachable nodes from
output links of the hole $\square$.
By the assumption on the graph context $\mathcal{G}[\square]$, all
the output links of the hole $\square$ have the provisional type
$\rot{\oc}\F$.
Therefore typing ensures the graph $H$ in fact consists of only
$\rot{\lb{C}}$-nodes and $\rot{\wn}$-nodes.
Therefore, we can turn the graph $H \circ (\vec{p}^\ddag)$ to a
provisional context
$\mathcal{P}[\square]$, by dropping all $\rot{\wn}$-nodes and
adding $k$ $\rot{\lb{C}}_0$-nodes, as below.
\begin{align*}
 \mathcal{P}[\square]
 =
\includegraphics{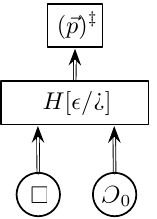}
\end{align*}
We say the provisional context $\mathcal{P}[\square]$ is ``induced''
by the graph context $\mathcal{G}[\square]$.
\begin{proposition}[Congruence of execution]
 \label{prop:ExecutionCongruence}
 Let $\mathcal{G}[G]$ be a composite graph where $G(1,n)$ is a
 definitive graph of type $T$, and $e$ be the root of the
 hole of the graph context $\mathcal{G}[\square]$.
 Assume an execution
 $\Init(\mathcal{P}[G])
 \to^* ((\mathcal{P}[G'],\ell'),(d,f,S',B'))$,
 where the provisional context $\mathcal{P}[\square]$ is induced by
 the graph context $\mathcal{G}[\square]$.
 Then, for any stacks $S$ and $B$, there exists a sequence
 \begin{align}
  ((\mathcal{G}[G],\ell),(\up,\square,\star \cl S,B))
  \to^* ((\mathcal{G}[G'],\ell'),(d,f,S'',B''))
  \label{eq:1}
 \end{align}
 of transitions, where $S'' = S' {::} S$ and $B'' = B' {::} B$.
 The decomposition ${::}$ replaces the bottom element $\square$ of
 the first stack with the second stack.
 Moreover, if $T$ is a function type, there also exists a sequence
 \begin{align}
  ((\mathcal{G}[G],\ell),(\up,\square,@ \cl S,B))
  \to^* ((\mathcal{G}[G'],\ell'),\delta')
  \label{eq:2}
 \end{align}
 of transitions, for some token data $\delta'$.
\end{proposition}
\begin{proof}
 By the way the provisional context $\mathcal{P}[\square]$ is induced
 by the graph context $\mathcal{G}[\square]$,
 any link of the graph $\mathcal{P}[G]$ has at least one (canonically)
 corresponding link in the graphs $G$, $H$ and $(\vec{p})^\ddag$.
 A link may have several corresponding links, because of
 $\rot{\wn}$-nodes dropped in the induced provisional context
 $\mathcal{P}[\square]$.
 The sequence (\ref{eq:1}) is given as a consequence of the following
 weak simulation result, where weakness is due to $\rot{\wn}$-nodes.
 \begin{description}
  \item[Weak simulation.]
	     Let
	     $((\mathcal{P}[H],\ell),(d,f,S,B))
	     \to ((\mathcal{P}[H'],\ell'),(d',f',S',B'))$
	     be a single transition of the assumed execution
	     $\Init(\mathcal{P}[G])
	     \to^* ((\mathcal{P}[G'],\ell'),(d,f,S',B'))$.
	     For any computation stack $S_0$ and any box stack $B_0$,
	     there exists a sequence
	     $((\mathcal{G}[H],\ell),
	     (d,f,S \mathrel{::} S_0,B \mathrel{::} B_0))
	     \to^* ((\mathcal{G}[H'],\ell''),
	     (d',f',S' \mathrel{::} S_0,B' \mathrel{::} B_0))$
	     of transitions from a stable state, where the position
	     $\ell''$ is one of those corresponding to $\ell'$.
 \end{description}
 The proof of the weak simulation follows the fact that the
 presence of the graph $G_0$ below the graph $G$ does not raise any
 extra rewriting, so long as the token data is taken from the
 execution on the graph $\mathcal{P}[G]$.

 Finally, if $T$ is a function type, replacing the element $\star$
 with the element $@$ in the sequence (\ref{eq:1}) only affects a pass
 transition over a $\lambda$-node, which sends the computation stack
 $\star \cl S$ to $\lambda \cl S$.

 The execution
 $\Init(\mathcal{P}[G])
 \to^* ((\mathcal{P}[G'],\ell'),(d,f,S',B'))$,
 in fact, can contain at most one pass transition over a
 $\lambda$-node which changes the computation stack
 $\star \cl \square$ to $\lambda \cl \square$.
 To make the second such transition happen, some other transition has
 to remove the top element of the computation stack
 $\lambda \cl \square$, however by stability
 (Prop.~\ref{prop:ExecStability}), no transition can do this.
 Moreover, such pass transition can be only the last transition of the
 execution.
 Any transitions that can possibly follow the pass transition, which
 sets direction $\dn$ and computation stack $\lambda \cl \square$, are
 pass transitions over $\oc$-nodes, $\rot{\oc}$-nodes,
 $\rot{\wn}$-nodes or $\rot{\lb{D}}$-nodes; the existence of these
 nodes contradicts with the type $T = T_1 \to T_2$ of the underlying
 graph.

 Since the sequence (\ref{eq:1}) weakly simulates the execution, where
 the weakness comes from only $\rot{\wn}$-nodes, we can conclude that
 there is no occurrence, or exactly one occurrence at the last, in the
 sequence (\ref{eq:1}), of the pass transition which is
 affected by the replacement of the element $\star$ with the element
 $@$.
 Therefore if the sequence (\ref{eq:1}) contains no such pass
 transition, the sequence (\ref{eq:2}) can be directly obtained by
 replacing the element $\star$ with the element $@$.
 Otherwise, cutting the last transition of the sequence (\ref{eq:1})
 just does the job, as the transition does not change the underlying
 graph and the token position.
\end{proof}

\section{Data-flow graphs}
\label{app:DataFlowGraphs}

This section looks at graphs consisting of specific nodes.
The restriction on nodes rules out some rewrites, especially
$@$-rewrites for function application and the decoupling rule.
\begin{definition}[Data-flow graphs]
 A \emph{data-flow} graph is a graph with no $\rot{\oc}$-nodes, that
 satisfies the following.
 \begin{itemize}
  \item All its input links have ground types.
  \item Any reachable (in the normal graphical sense) nodes from
	input links are labelled with
	$\{ p,\vec{p},\$^0,\oc,\wn,\rot{\wn},\lb{D},\lb{C}_m,
	\rot{\lb{D}},\rot{\lb{C}}_m
	\mid p \in \F,\ \vec{p} \in \F^n,\ \$^0 \in \Sigma,\ 
	n \in \N,\ m \in \N \}$.
 \end{itemize}
 In particular, a data-flow graph is called \emph{pure} if these
 reachable nodes are not labelled with
 $\{ \oc, \wn, \rot{\wn}, \lb{D}, \lb{C}_m \mid m \in \N \}$.
\end{definition}

Data-flow graphs intensionally characterises graphs of final states.
Graphs of final states play the role of ``values,'' since our
semantics implements (right-to-left) call-by-value evaluation.
\begin{proposition}[Final graphs intensionally]
 \label{prop:FinalGraphsIntensionally}
 Let $G \circ (\vec{p})^\ddag$ be a composite graph of (non-enriched)
 type $T$.
 If a final state $\Final(G \circ (\vec{p})^\ddag,X)$ is stable, the
 definite graph $G$ satisfies the following.
 \begin{enumerate}
  \item When $T$ is a ground type, the graph $G$ is a pure
	data-flow graph.
  \item When $T$ is a function type, i.e.\ $T = T_1 \to T_2$, the root
	of the graph $G$ is the input of a $\lambda$-node.
 \end{enumerate}
\end{proposition}
\begin{proof}
 The second case, where $T = T_1 \to T_2$, is a direct consequence of
 Prop.~\ref{prop:StableExecIntensionally}.
 For the first case, where $T$ is a ground type,
 Prop.~\ref{prop:StableExecIntensionally} tells us that the stable
 execution
 $\Init(G \circ (\vec{p})^\ddag)
 \to^* \Final(G \circ (\vec{p})^\ddag,X)$
 only involves nodes labelled with
 $\{ p,\vec{p},\$^0, \rot{\oc},\rot{\lb{D}},\rot{\lb{C}}_m
 \mid p \in \F,\ \vec{p} \in \F^n,\ \$^0 \in \Sigma,\ 
 n \in \N,\ m \in \N \}$.
 It boils down to show that any reachable node from the root of the
 graph $G$ is involved by the stable execution.
 We can show this by induction on the maximum length of paths from the
 root to a reachable node.
 The base case is trivial, as the root of the graph $G$ coincides with
 an input link of the reachable node.
 In the inductive case, induction hypothesis implies that any
 reachable node is connected to a reachable node which is involved by
 the stable execution.
 By Prop.~\ref{prop:FactoriseStableExec}.2 and labelling of the
 involved node, the stable execution contains a transition that passes
 the token upwards over the involved node, and hence makes the
 reachable node of interest involved by the following transition.
\end{proof}

We can directly prove soundness of data-flow graphs.
\begin{proposition}
 \label{prop:GenSoundnessDataFlowGraphs}
 Let $G(1,n)$ be a data-flow graph, with a link $\ell$ of ground type
 which is reachable from the root of $G$.
 For any vector $\vec{p} \in \F^n$,
 if a state $((G \circ (\vec{p})^\ddag,\ell),(\up,\square,S,B))$ is
 stable and valid, there exists a
 data-flow graph $G'(1,n)$ that agrees with $G$ on the link $\ell$,
 and a computation stack $S'$, such that
 $((G \circ (\vec{p})^\ddag,\ell),(\up,\square,S,B))
 \to^* ((G' \circ (\vec{p})^\ddag,\ell),(\dn,\square,S',B))$.
\end{proposition}
\begin{proof}
 The first observation is that any transition sends a data-flow graph,
 composed with a graph $(\vec{p})^\ddag$, to a data-flow graph with
 the same graph $(\vec{p})^\ddag$.

 Given a composite graph $G \circ (\vec{p})^\ddag$ where $G$ is a
 data-flow graph, we define a partial \emph{ranking} map
 $\rho$ which assigns natural numbers to some links of $G$.
 The ranking is only defined on links which are reachable from the
 root of $G$ and labelled with either a ground type or an argument
 type, as below.
 \begin{align*}
  \rho(\ell) &:= 0
  \tag{if $\ell$ is input of a $p$-node ($p \in \F$), a
  $\vec{q}$-node ($\vec{q} \in \F^k$) or a $\rot{\lb{D}}$-node} \\
  \rho(\ell) &:= \mathit{max}(\rho(\ell_1),\rho(\ell_2))+1
  \tag{if $\ell$ is input of a $\$^0$-node
  ($\$^0 \in \Sigma$), and $\ell_1$ and $\ell_2$ are output links of
  the $\$^0$-node} \\
  \rho(\ell) &:= \rho(\ell')+1
  \tag{if $\ell$ is input of a $\oc$-node, a $\wn$-node, a
  $\lb{D}$-node or a $\lb{C}_n$-node, and $\ell'$ is the corresponding
  output link}
 \end{align*}
 This ranking on reachable links is well-defined, as the composite
 graph $G \circ (\vec{p})^\ddag$ meets the graph criterion.

 Since the state $((G \circ (\vec{p})^\ddag,\ell),(\up,\square,S,B))$
 is stable and the position $\ell$ has ground type, the ranking $\rho$
 of the composite graph $G \circ (\vec{p})^\ddag$ is defined on the
 position $\ell$.
 The proof is by induction on the rank $\rho(\ell)$.

 Base cases are when $\rho(\ell) = 0$.
 If the position $\ell$ is input of a $\rot{\lb{D}}-node$, the graph
 criterion implies an acyclic directed path from the
 $\rot{\lb{D}}-node$ to a $\rot{\oc}$-node.
 Intermediate nodes of this path are only $\rot{\lb{C}}_n$-nodes, and
 the proof is by induction on the number of these
 $\rot{\lb{C}}_n$-nodes.
 Otherwise, the position $\ell$ is input of a constant node ($p$ or
 $\vec{q}$), and the proof is by one pass transition over the node.

 In inductive cases, induction hypothesis is for any natural number
 that is less than $\rho(\ell)$.
 When the position $\ell$ is input of a $\lb{D}$-node, the graph
 criterion implies an acyclic directed path from the $\lb{D}$-node to
 a $\oc$-node, with only $\lb{C}_n$-nodes as intermidiate nodes.
 This means, from the state
 $((G \circ (\vec{p})^\ddag,\ell),(\up,\square,S,B))$, the token goes
 along the directed path, reaches the $\oc$-node, and trigger
 rewrites.
 These rewrites eliminate all the nodes in the path, and possibly
 include deep rules that eliminate other $\wn$-nodes and
 $\lb{C}$-nodes.
 When these rewrites are completed, the position $\ell$ and its type
 are unchanged, but its rank $\rho(\ell)$ is strictly decreased.
 Therefore we can use induction hypothesis to prove this case.
 The last case, when the position $\ell$ is input of a $\$^0$-node,
 boils down to repeated but straightforward use of induction
 hypothesis, which may be followed by a $\$^0$-rewrite.
\end{proof}
\begin{corollary}[Soundness of data-flow graphs]
 \label{cor:SoundnessDataFlowGraphs}
 If a data-flow graph $G(1,n)$ meets the name criteria and the graph
 criterion, for any
 vector $\vec{p} \in \F^n$, there exists a data-flow graph $G'(1,n)$
 and an element $X$ of a computation stack such that
 $\Init(G \circ (\vec{p})^\ddag)
 \to^* \Final(G' \circ (\vec{p})^\ddag,X)$.
\end{corollary}

Roughly speaking,  abduction enables us to replace
data-flow graphs with other data-flow graphs, which resembles
parameter updating.
This replacement is not at all simple; in an execution, it can happen
inside a $\oc$-box, or happen outside a $\oc$-box before the
resulting graph gets absorbed by the $\oc$-box.
Moreover, it can change the number of provisional constants extracted
by decoupling.
Our starting point to formalise this idea of replacement is the notion
of ``data-flow chain''.
It is a sequence of sub-graphs, which are partitioned by
auxiliary doors and essentially representing data flow.
\begin{definition}[Data-flow chains]
 In a graph $G$, a \emph{data-flow chain}
 $\mathbf{D}$ is given by a sequence
 $D_0(n_0,n_1),\ldots,D_k(n_k,n_{k+1})$ of $k+1$ sub-graphs, where
 $k$ is a natural number, that satisfies the following.
 \begin{itemize}
  \item The first sub-graph $D_0(n_0,n_1)$ is a data-flow graph.
  \item If $k > 0$, there exists a unique number $h$ such that
	$0 < h \leq k$.
	
	For each $i = 1,\ldots,h-1$, the sub-graph $D_i(n_i,n_{i+1})$
	can contain only
	$\lb{C}$-nodes, $\lb{P}$-nodes, $\rot{\lb{C}}$-nodes or
	$\oc$-boxes with their doors, where these $\oc$-boxes
	are data-flow graphs.
	Input links of the sub-graph $D_i(n_i,n_{i+1})$ are connected
	to output links of the previous sub-graph
	$D_{i-1}(n_{i-1},n_i)$, via $n_i$ parallel $\wn$-nodes.
	These delimiting parallel doors ($\wn$) belong to the same
	$\oc$-box, whose principal door ($\oc$) is not included in the
	whole sequence of sub-graphs.

	For each $j = h,\ldots,k$, the sub-graph $D_j(n_j,n_{j+1})$
	solely consists of $\rot{\lb{C}}$-nodes.
	Input links of the sub-graph $D_j(n_j,n_{j+1})$ are connected
	to output links of the previous sub-graph
	$D_{j-1}(n_{j-1},n_j)$, via $n_j$ parallel $\rot{\wn}$-nodes.
	These delimiting parallel doors ($\rot{\wn}$) belong to the
	same $\oc$-box, whose principal door ($\oc$) is not included
	in the whole sequence of sub-graphs.
  \item The final sub-graph $D_k(n_k,n_{k+1})$ satisfies either one of
	the following: (i) all its output links have the provisional
	type $\rot{\oc}\F$ and connected to $\rot{\oc}$-nodes, or (ii)
	all its output links are input links of a single
	$\lb{P}_{n_{k+1}}$-node, whose output link is connected to a
	$\lambda$-node.
  \item If a node of the graph $G$ is box-reachable from an input link
	of the first sub-graph $G_0$, it is either (i) in the
	sub-graphs $D_0,\ldots,D_k$, (ii) in auxiliary doors
	partitioning them, or (iii) box-reachable from an output link
	of the last sub-graph $G_k$.
 \end{itemize}
\end{definition}
\noindent
We refer to input of the first sub-graph $D_0$ as input of the
data-flow chain $\mathbf{D}$, and output of the last sub-graph $D_k$
as output of the data-flow chain $\mathbf{D}$.
The following illustrates some forms of data-flow chains.
\begin{center}
 \includegraphics{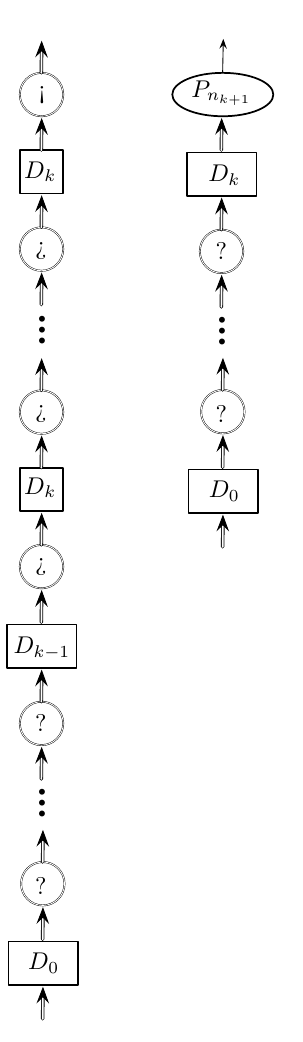}
\end{center}

We define a binary relation $\propto$ between definitive graphs that
applies single replacement of a data-flow chain.
It is lifted to a binary relation $\overline{\propto}$ on some states.
\begin{definition}[Data-flow replacement of graphs]
 Let $G(1,n)$ and $H(1,m)$ be two definitive graphs, that contain
 data-flow chains $\mathbf{D}_G$ and $\mathbf{D}_H$, respectively.
 Two definitive graphs $G(1,n)$ and $H(1,m)$ satisfies $G \propto H$
 if the following holds.
 \begin{itemize}
  \item Two data-flow chains have the same number of input.
	The data-flow chain $\mathbf{D}_H$ has no more length than the
	data-flow chain $\mathbf{D}_G$.
	The data-flow chain $\mathbf{D}_H$ can have an artibrary
	number of output, whereas the data-flow chain $\mathbf{D}_H$
	has at least one output.
  \item Exactly the same set of names appears in both graphs $G$ and
	$H$.
  \item The graphs $G$ and $H$ only differ in the data-flow chains
	$\mathbf{D}_G$ and $\mathbf{D}_H$, and their partitioning
	auxiliary doors.
 \end{itemize}
\end{definition}
\begin{definition}[Data-flow replacements of states]
 Two stable and valid states
 $((G \circ (\vec{p})^\ddag,\ell),(d,f_1,S_1,B_1))$ and
 $((H \circ (\vec{q})^\ddag,\ell),(d,f_2,S_2,B_2))$ are related by a
 binary relation $\overline{\propto}$ if the following holds.
 \begin{itemize}
  \item The definitive graphs satisfy $G \propto H$.
  \item The position $\ell$ is either an input link of the data-flow
	chain replaced by $\propto$, or (strictly) outside the
	data-flow chain.
  \item Two stable executions to these states give the exactly same
	sequence of positions.
  \item The rewriting flags $f_1$ and $f_2$ may only differ in numbers
	$n$ of $\$^1(n)$.
  \item The validation maps $v_G$ and $v_H$ of these states,
	respectively, satisfy that: $v_G(a) = 0$ implies $v_H(a) = 0$,
	for any $a \in \A$ on which they both are defined.
 \end{itemize}
\end{definition}

As usual, we use $\overline{\propto}^*$ to denote the reflexive and
transitive closure of the relation $\overline{\propto}$.
\begin{proposition}
 \label{prop:DataFlowReplacement}
 If
 $((G \circ (\vec{p})^\ddag,\ell),(\up,f_1,S_1,B_1))
 \overline{\propto}^*
 ((H \circ (\vec{q})^\ddag,\ell),(\up,f_2,S_2,B_2))$
 holds, a sequence
 \begin{align}
  ((G \circ (\vec{p})^\ddag,\ell),(\up,f_1,S_1,B_1))
  \to^* ((G' \circ (\vec{p}^\ddag),\ell),(\dn,\square,S'_1,B'_1)),
  \label{eq:5}
 \end{align}
 implies a sequence
 \begin{align}
  ((H \circ (\vec{q})^\ddag,\ell),(\up,f_2,S_2,B_2))
  \to^* ((H' \circ (\vec{q}^\ddag),\ell),(\dn,\square,S'_2,B'_2))
  \label{eq:6}
 \end{align}
 such that the resulting states are again related by
 $\overline{\propto}^*$.
\end{proposition}
\begin{proof}
 The proof is by induction on the length of the sequence (\ref{eq:5}).
 Base cases are when the sequence (\ref{eq:5}) consists of one pass
 transition over a constant node (scalar or vector) or a
 $\lambda$-node, and hence $f = \square$.
 If the transition is over a $\lambda$-node, the same transition is
 possible at the state
 $((H \circ (\vec{q})^\ddag,\ell),(\up,f,S_2,B_2))$.
 If the transition is over a constant node, the constant node may be a
 part of a data-flow chain replaced by a data-flow chain $\mathbf{D}$
 in the graph $H$.
 By stability of the state
 $((H \circ (\vec{q})^\ddag,\ell),(\up,f,S_2,B_2))$, we can conclude
 that any partitioning auxiliary doors of the data-flow chain
 $\mathbf{D}$ are $\rot{\wn}$-nodes, if they are box-reachable from
 the position $\ell$.
 This can be confirmed by contradiction as follows:
 otherwise the position $\ell$ must be in a $\oc$-box with definitive
 auxiliary doors (i.e.\ $\wn$-nodes), which contradicts with stability
 and the fact $f = \square$.
 This concludes the proof for base cases.

 First inductive case is when the position $\ell$ is an input link of
 a data-flow chain $\mathbf{D}_G$, which is replaced with a data-flow
 chain $\mathbf{D}_H$ in the graph $H$.
 If the rewriting flag is $f = \square$, similar to base cases, we can
 see that the data-flow chain $\mathbf{D}_H$ is in fact not
 partitioned by $\wn$-nodes (but possibly $\rot{\wn}$-nodes).
 Moreover by stability and the graph criterion, output links of the
 data-flow chain $\mathbf{D}_H$ are not connected to a $\lb{P}$-node.
 This implies that the data-flow chain $\mathbf{D}_H$ with
 partitioning auxiliary doors altogether gives a data-flow graph.
 Therefore the sequence (\ref{eq:6}) can be obtained by
 Prop.~\ref{cor:SoundnessDataFlowGraphs} and
 Prop.~\ref{prop:ExecutionCongruence}.

 If the rewriting flag $f$ is not equal to $\square$, possibilities
 are $f = \lambda,\wn,\oc$.
 The $\lambda$-rewrite in the graph $G$ sets the flag $\square$ and
 does not change the position $\ell$.
 The same $\lambda$-rewrite is also possible in the graph $H$, and we
 can use induction hypothesis.
 If $f = \wn$, there will be at least one rewrite transitions, in both
 graphs $G$ and $H$, until the flag is changed to $\oc$.
 These $\wn$-rewrites may affect the data-flow chains $\mathbf{D}_G$
 and $\mathbf{D}_H$.
 Since the position $\ell$ is an input of the data-flow chains, these
 $\wn$-rewrites can only eliminate $\lb{C}$-nodes, $\lb{P}$-nodes, or
 $\wn$-nodes that partition the data-flow chains.
 Elimination of $\lb{C}$-nodes and $\lb{P}$-nodes is where the
 transitive closure $\propto^*$ plays a role.
 It does not change the partitioning structure of the data-flow chains
 $\mathbf{D}_G$ and $\mathbf{D}_H$.
 Elimination of $\wn$-nodes in the graph $G$ must introduce
 $\rot{\wn}$-nodes, because the replacement $\propto$ requires the
 data-flow chain $\mathbf{D}_G$ to have at least one output.
 Therefore the $\wn$-rewrites changes the data-flow chain
 $\mathbf{D}_G$ to a new one while keeping its length.
 On the other hand, elimination of $\wn$-nodes may not happen in the
 graph $H$, or may decrease the length of the data-flow chain
 $\mathbf{D}_H$.
 As a result, after the maximal number of $\wn$-rewrites until the
 rewriting flag is changed to $\oc$, resulting graphs are still
 related by $\propto^*$ and the postion $\ell$ is not changed.
 Finally if $f = \oc$, until the rewriting flag is changed to
 $\square$, the same nodes ($\lb{D}$-nodes and $\lb{C}$-nodes) are
 eliminated in both graphs $G$ and $H$, and both the data-flow chains
 decrease their length by one, if possible.
 Once rewrites are done and the rewriting flag $\square$ is set, the
 position $\ell$ is still unchanged, and we use induction hypothesis.
 
 Second inductive case is when the position $\ell$ is the input of a
 $\$^1$-node, with rewriting flags $f_1 = \$^1(n_1)$ and
 $f_2 = \$^1(n_2)$ are raised.
 If $n_1 = 0$, by definition of the relation $\overline{\propto}$, it
 holds that $n_2 = 0$, and the proof follows stability.
 If not, the sequence (\ref{eq:5}) begins with non-trivial unfolding
 of the $\$^1$-node.
 The sequence (\ref{eq:6}) can be proved by induction on $n_2$, which
 is an arbitrary natural number, with $n_2 \leq n_1$ being the base
 case.

 Finally, the last inductive case is when the position $\ell$ is not
 in data-flow graphs replaced by $\overline{\propto}^*$, or the input
 of any $\$^1$-node.
 If $f = \square$ and the sequence (\ref{eq:5}) begins with a pass
 transition, the same transtion is possible in the graph
 $H \circ (\vec{q})^\ddag$, and we can use induction hypothesis.
 We use it more than once, when the pass transition is over a
 $\$$-node.
 Possibly the sequence (\ref{eq:5}) ends with a $\$^0$-rewrite, which
 may not be possible on the other side.
 However, this is when the position $\ell$ becomes an input of a
 data-flow chain in the resulting graph $G'$, and the difference of
 $\$^0$-rewrites is dealt with by the replacement $\propto$.
 If $f \neq \square$, discussion in the first inductive case is valid,
 except for any $\wn$-rewrites being possible, namely the decoupling
 rule.
 We use induction hypothesis once consecutive rewrites are done.
 The key fact is that, when the decoupling rule applies to graphs
 related by the replacement $\propto$, the resulting graphs are again
 related by $\propto$.
 The resulting graphs may differ in the size of extracted vectors and
 in the number of input links of the introduced $\lb{P}$-nodes.
 This is dealt with by the replacement $\propto$ of data-flow chains,
 in particular, a single constant node itself is a data-flow chain.
 Note that, if a $\lb{P}_0$-node is introduced on the side of graph
 $G$, it is also introduced on the other side, because any data-flow
 chain of null output is not replaced by $\propto$.
 The decoupling rule is essentially the only transition that is
 relevant to the condition of validation maps for the relation
 $\overline{\propto}$, and it does not violate the condition. 
 This concludes the whole proof.
\end{proof}
\begin{corollary}[Safety of dara-flow replacement]
 \label{cor:DataFlowReplacementSafety}
 Let $G \circ (\vec{p})^\ddag$ and $H \circ (\vec{q})^\ddag$ be
 composite graphs, meeting the name criteria and the graph criterion,
 such that $G \propto^* H$.
 If an execution on the graph $G \circ (\vec{p})^\ddag$ reaches a
 final state, an execution on the graph $H \circ (\vec{q})^\ddag$ also
 reaches a final state.
\end{corollary}

\section{Soundness}
\label{app:Soundness}

Our soundness proof uses logical predicates on definitive graphs.
These logical predicates are analogous to known ones on typed
lambda-terms.
\begin{definition}[Logical predicates]
 Given a term $T$, a logical predicate $P_T$ is on definitive graphs,
 that meet the name criteria and the graph criterion, of type $T$.
 It is inductively defined as below.
 \begin{itemize}
  \item When $T$ is a ground type, $G(1,n) \in P_T$ holds if: for any
	provisional context $\mathcal{P}[\square]_T^n$, there exist a
	composite graph	$H$ and an element $X$ of a computation stack
	such that $\Init(\mathcal{P}[G]) \to^* \Final(H,X)$.
  \item When $T = T_1 \to T_2$, $G(1,n) \in P_T$ holds if:
	\begin{enumerate}
	 \item for any provisional context
	       $\mathcal{P}[\square]_T^n$,
	       there exists a composite graph $H$ such that
	       $\Init(\mathcal{P}[G]) \to^* \Final(H,\lambda)$
	 \item for any $H(1,m) \in P_{T_1}$,
	       the following graph, denoted by
	       $G @ \oc H$, satisfies ${G @ \oc H} \in P_{T_2}$.
	\end{enumerate}
 \end{itemize}
 \begin{center}
  \includegraphics{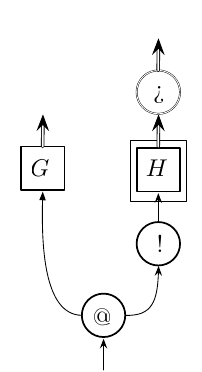}
 \end{center}
\end{definition}
\begin{proposition}[Deterministic termination]
 \label{prop:DeterministicTermination}
 If $G(1,n) \in P_T$, for any provisional context
 $\mathcal{P}[\square]_T^n$, there exist a unique definitive graph
 $G'(1,n)$ of type $T$ and a unique element $X$ of a computation stack
 such that $\Init(\mathcal{P}[G]) \to^* \Final(\mathcal{P}[G'],X)$.
\end{proposition}
\begin{proof}
 This is a direct consequence of Prop.~\ref{prop:ExecDeterminism} and
 Prop.~\ref{prop:ProvisionalCtxtPreservation}.
\end{proof}
\begin{proposition}[Congruence of termination]
 \label{prop:TerminationCongruence}
 Let $\mathcal{G}[G]$ be a composite graph where $G(1,n)$ is a
 definitive graph of type $T$, and $e$ be the root of the
 hole of the graph context $\mathcal{G}[\square]$.
 If $G(1,n) \in P_T$ holds, then for any stacks $S$ and $B$, there
 exist an element $X$ of a computation stack and a sequence
 \begin{align*}
  ((\mathcal{G}[G],\ell),(\up,\square,\star \cl S,B))
  \to^* ((\mathcal{G}[G'],\ell),(\dn,\square,X \cl S,B))
 \end{align*}
 of transitions.
 Moreover, if $T$ is a function type, there also exists the following
 sequence.
 \begin{align*}
  ((\mathcal{G}[G],\ell),(\up,\square,@ \cl S,B))
  \to^* ((\mathcal{G}[G'],\ell),(\up,\square,@ \cl S,B))
 \end{align*}
\end{proposition}
\begin{proof}
 This is a corollary of Prop.~\ref{prop:ExecutionCongruence}.
\end{proof}

The following properties relate logical predicates to transitions, in
both forward and backward ways.
\begin{proposition}[Forward/backward reasoning]
 \label{prop:FWBWReasoning}
 \noindent
 \begin{description}
  \item[Forward reasoning]
	     Let $G(1,n)$ be a definitive graph such that $G \in P_T$,
	     and $\mathcal{P}[\square]_T^n$ be a provisional context.
	     For any execution
	     $\Init(\mathcal{P}[G]) \to^* ((H',\ell),\delta)$ on the
	     graph $\mathcal{P}[G]$, there exists a definitive graph
	     $G'(1,n)$ such that $H' = \mathcal{P}[G']$ and
	     $G' \in P_T$.
  \item[Backward reasoning]
	     A definitive graph $G(1,n)$ satisfies $G \in P_T$, if:
	     (i) it meets the name criteria and the graph criterion,
	     and (ii) for
	     any provisional context $\mathcal{P}[\square]_T^n$, there
	     exist a definitive graph $G'(1,n) \in P_T$ and a state
	     $((\mathcal{P}[G'],\ell),\delta)$ such that
	     $\Init(\mathcal{P}[G])
	     \to^* ((\mathcal{P}[G'],\ell),\delta)$.
 \end{description}
\end{proposition}
\begin{proof}
 First of all, Prop.~\ref{prop:ProvisionalCtxtPreservation} ensures
 the decomposition $H' = \mathcal{P}[G']$, where $G'(1,n)$ is a
 definitive graph, in the forward reasoning.
 We prove the both reasoning simultaneously by induction on the type
 $T$.
 Base cases of both reasoning, where $T$ is a ground type, relies on
 determinism and stability, as we see below.
 
 We begin with the base case of the forward reasoning, where $T$ is a
 ground type.
 Given any execution
 $\Init(\mathcal{P}[G]) \to^* ((\mathcal{P}[G'],\ell),\delta)$
 where $G(1,n) \in P_T$ and $\mathcal{P}[\square]_T^n$ is a
 provisional context, by Prop.~\ref{prop:DeterministicTermination},
 there exists a unique final state $\Final(H'',X)$ such that
 $\Init(\mathcal{P}[G]) \to^* \Final(H'',X)$.
 The uniqueness implies the factorisation
 $\Init(\mathcal{P}[G]) \to^* ((\mathcal{P}[G'],\ell),\delta)
 \to^* \Final(H'',X)$.
 Since the state $((\mathcal{P}[G'],\ell),\delta)$ is stable by
 Prop.~\ref{prop:ExecStability}, we have a stable execution
 $\Init(\mathcal{P}[G']) \to ^* ((\mathcal{P}[G'],\ell),\delta)$.
 As a result we have an execution
 $\Init(\mathcal{P}[G']) \to ^* \Final(H'',X)$, which proves
 $G'' \in P_T$.

 In the base case of the backward reasoning, where $T$ is a ground
 type, $G' \in P_T$ implies an execution
 $\Init(\mathcal{P}[G']) \to^* \Final(H,X)$
 to a unique final state (Prop.~\ref{prop:DeterministicTermination}).
 Since the last state of the execution
 $\Init(\mathcal{P}[G]) \to^* ((\mathcal{P}[G'],\ell),\delta)$
 is stable by Prop.~\ref{prop:ExecStability}, there is a stable
 execution
 $\Init(\mathcal{P}[G']) \to^* ((\mathcal{P}[G'],\ell),\delta)$.
 The uniqueness of the final state $\Final(H,X)$ gives the sequence
 $((\mathcal{P}[G'],\ell),\delta) \to^* \Final(H,X)$
 of transitions, which yields an execution
 $\Init(\mathcal{P}[G]) \to^* ((\mathcal{P}[G'],\ell),\delta)
 \to^* \Final(H,X)$
 and proves $G \in P_T$.

 In inductive cases of both reasoning, where $T = T_1 \to t_2$, we
 need to check two conditions of the logical predicate $P_T$.
 The first condition, i.e.\ termination, is as the same as base cases.
 The other inductive condition can be proved using induction
 hypotheses of both properties, together with
 Prop.~\ref{prop:ExecutionCongruence} and
 Cor.~\ref{prop:TerminationCongruence}, as below.

 In the inductive case of the forward reasoning, our goal is to prove
 ${G' @ \oc H} \in P_{T_2}$ for any $H(1,m) \in P_{T_1}$, under the
 assumption of the execution
 $\Init(\mathcal{P}[G]) \to^* ((\mathcal{P}[G'],\ell),\delta)$ where
 $G(1,n) \in P_{T_1 \to T_2}$.
 Let $\mathcal{P'}[\square]_{T_2}^{n+m}$ be any provisional context.
 Since $H \in P_{T_1}$, Cor.~\ref{prop:TerminationCongruence} implies
 two executions, where the position $\ell'$ is the right output of the
 $@$-node,
 \begin{align}
  \Init(\mathcal{P'}[G @ \oc H])
  &\to^* ((\mathcal{P'}[G @ \oc H'],\ell'),
  (\dn,\square,X \cl \star \cl \square,\square)) \label{eq:3} \\
  \Init(\mathcal{P'}[G' @ \oc H])
  &\to^* ((\mathcal{P'}[G' @ \oc H'],\ell'),
  (\dn,\square,X \cl \star \cl \square,\square)) \label{eq:4}
 \end{align}
 such that
 $\Init(\mathcal{P''}[H]) \to^* \Final(\mathcal{P''}[H'],X)$
 for some provisional context $\mathcal{P''}[\square]_{T_1}^m$ and
 some element $X$ of a computational stack.
 By the assumption of $G \in P_{T_1 \to T_2}$ and
 Prop.~\ref{prop:ExecutionCongruence}, we can continue the
 execution (\ref{eq:3}) as:
 \begin{align*}
  \Init(\mathcal{P'}[G @ \oc H])
  &\to^* ((\mathcal{P'}[G @ \oc H'],\ell'),
  (\dn,\square,X \cl \star \cl \square,\square))
  \to^* ((\mathcal{P'}[G' @ \oc H'],\ell), \delta)
 \end{align*}
 for some token data $\delta$.
 Since $G \in P_{T_1 \to T_2}$ and $H \in P_{T_2}$ by the assumption,
 we can use induction hypothesis of the forward reasoning and obtain
 ${G' @ \oc H'} \in P_{T_2}$.
 Using induction hypothesis of the backward reasoning along the
 execution (\ref{eq:4}), we conclude ${G' @ \oc H} \in P_{T_2}$.
 
 In the inductive case of the backward reasoning, we aim to prove
 ${G @ \oc H} \in P_{T_2}$ for any $H(1,m) \in P_{T_1}$.
 Let $\mathcal{P'}[\square]_{T_2}^{n+m}$ be any provisional context.
 Since $H \in P_{T_1}$, Cor.~\ref{prop:TerminationCongruence} gives an
 execution, where the position $\ell'$ is the right output of the
 $@$-node,
 \begin{align*}
  \Init(\mathcal{P'}[G @ \oc H])
  &\to^* ((\mathcal{P'}[G @ \oc H'],\ell'),
  (\dn,\square,X \cl \star \cl \square,\square))
 \end{align*}
 such that there exists a provisional context
 $\mathcal{P''}[\square]_{T_1}^m$ and an execution
 $\Init(\mathcal{P''}[H]) \to^* \Final(\mathcal{P''}[H'],X)$.
 Using the assumption on the graph $G$, with
 Prop.~\ref{prop:ExecutionCongruence}, yields its expansion
 \begin{align*}
  \Init(\mathcal{P'}[G @ \oc H])
  &\to^* ((\mathcal{P'}[G @ \oc H'],\ell'),
  (\dn,\square,X \cl \star \cl \square,\square))
  \to^* ((\mathcal{P'}[G' @ \oc H'],\ell), \delta)
 \end{align*}
 such that $G' \in P_{T_1 \to T_2}$, arising in an execution
 $\Init(\mathcal{P'''}[G]) \to^* ((\mathcal{P'''}[G'],\ell),\delta)$,
 for some provisional context
 $\mathcal{P'''}[\square]_{T_1 \to T_2}^n$.
 Induction hypothesis of the forward reasoning implies 
 $H' \in P_{T_1}$, and therefore ${G' @ \oc H'} \in P_{T_2}$.
 We conclude ${G @ \oc H} \in P_{T_2}$ by induction hypothesis of the
 backward reasoning.
\end{proof}

The key ingredient of the soundness proof is ``safety'' of decoupling.
This is where we appreciate call-by-value evaluation.

\begin{proposition}[Safety of decoupling]
 \label{prop:AbductionSafety}
 If $G(1,n) \in P_T$ holds, the graph $\hat{G}$ given as below
 \begin{align*}
  \hat{G} =
  \includegraphics{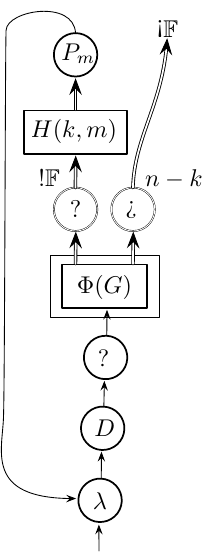}
 \end{align*}
 belongs to $P_{V_a \to T}$, where:
 \begin{itemize}
  \item the name $a \in \A$ is any name that does not appear in the
	graph $G'$
  \item the graph $H(k,m)$, where $k \leq n$ and $m$ is arbitrary,
	consists solely of $\lb{C}$-nodes connected to each other in
	an arbitrary way, and fulfills the graph criterion
  \item the graph $\phi(G)$ is obtained by:
	(i) choosing $k$ output links of the graph $G$ arbitrarily,
	and (ii) replacing any
	$\rot{\lb{C}}$-nodes, $\rot{\wn}$-nodes and
	$\rot{\lb{D}}$-nodes with $\lb{C}$-nodes, $\wn$-nodes and
	$\lb{D}$-nodes, respectively, if one of the chosen output
	links can be reachable from these nodes via links of only the
	provisional type $\rot{\oc}\F$.
 \end{itemize}
\end{proposition}
\begin{proof}
 It is easy to see the graph $\hat{G}$ meets the name criteria and the
 graph criterion, given $G \in P_T$.
 Since the graph $\hat{G}$ has a $\lambda$-node at the bottom, the
 termination condition of the logical predicate $P_{V_a \to T}$ is
 trivially satisfied.
 For any graph $E \in P_{V_a}$, we prove $\hat{G} @ \oc E \in P_T$.

 Let $\mathcal{P'}[\square]_T$ be any provisional context.
 By Cor.~\ref{prop:TerminationCongruence},
 an execution on the graph $\mathcal{P'}[\hat{G} @ \oc E]$ first
 yields the graph $\mathcal{P'}[\hat{G} @ \oc E']$, where the graph
 $E'$ comes from some execution
 $\Init(\mathcal{P}_0[E]) \to^* \Final(\mathcal{P}_0[E'],X)$ to a
 final state.
 Then the token eliminates the pair
 of the $\lambda$-node and the $@$-node at the bottom of the graph,
 and triggers the rewrite involving the graph $H$ ($\lb{C}$-nodes) and
 the $\lb{P}_m$-node of the graph $\phi(G)$.
 This rewrite duplicates the graph $E'$ in a $\oc$-box, introducing
 dot-product nodes and vector nodes.
 Finally the token eliminates the $\lb{D}$-node and the $\oc$-box
 around the graph $\phi(G)$.
 In the resulting graph, we shall write as $\mathcal{P'}[R]$, the
 graph $R$ consists of the graph $\phi(G)$, whose output links of type
 $\oc\F$ are connected to $\oc$-boxes, and further,
 $\rot{\lb{C}}$-nodes.
 The following illustrates the graph $R$, where $\rot{\lb{C}}$-nodes
 are omitted.
 \begin{center}
  \includegraphics{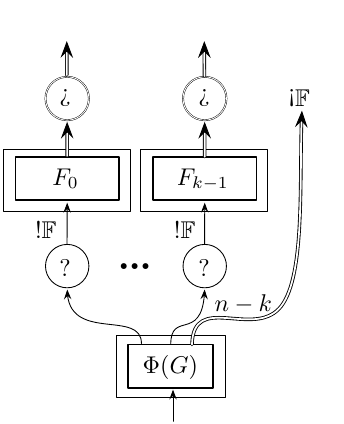}
  \includegraphics{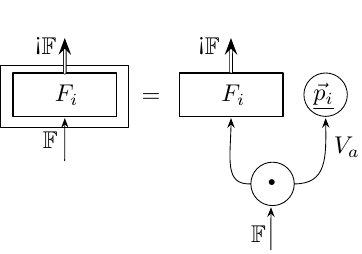}
 \end{center}
 Let $\vec{F}$ be these $\oc$-boxes.
 They all have type $\F$, and each of them contains the graph $E'$,
 with a dot-product node and a vector node.
 Since the provisional context $\mathcal{P'}$ is arbitrary, we can
 reduce the problem to $R \in P_T$, using the
 backward reasoning (Prop.~\ref{prop:FWBWReasoning}).

 If $n = 0$, the replacement $\phi$ actually changes nothing on the
 graph $G$, and hence $\phi(G) = G$.
 Therefore, in this case, $R \in P_T$ follows $G \in P_T$, by
 Prop.~\ref{prop:TerminationCongruence}.
 We deal with the case where $n > 0$ below.

 First, as a consequence of Prop.~\ref{cor:DataFlowReplacementSafety},
 the execution on the graph $\mathcal{P'}[R]$ reaches a final state,
 given the assumption $G \in P_T$.
 This is because, for any provisional context
 $\mathcal{P}[\square]_T^n$, we have
 $\mathcal{P}[G] \propto^* \mathcal{P'}[R]$.
 Since any name appears in the graph $\mathcal{P}[G]$ also appears in
 the graph $\mathcal{P'}[R]$, we can infer
 $\Init(\mathcal{P}[G]) \overline{\propto}^* \Init(\mathcal{P'}[R])$.
 This means $R \in P_T$ when $T$ is a ground type, and the termination
 condition of $R \in P_T$ when $T$ is a function type.

 To check the inductive condition of $R \in P_T$ where
 $T = T_1 \to T_2$, we need to show ${R @ \oc F} \in P_{T_2}$ for any
 $F \in P_{T_1}$.
 By the assumption $G \in P_{T_1 \to T_2}$, we have
 ${G @ \oc F} \in P_{T_2}$.
 Using induction hypothesis on this graph yields the graph
 $\widehat{G @ \oc F} \in P_{V_b \to T_2}$.
 Let $\tilde{E} \in P_{V_b}$ is a graph obtained by renaming $E$.
 We can take a provisional context $\mathcal{P''}[\square]_{T_2}$ such
 that an execution on the graph
 $\mathcal{P''}[(\widehat{G @ \oc F}) @ \oc \tilde{E}]$ yields a graph
 $\mathcal{P''}[(\widehat{G @ \oc F}) @ \oc \tilde{E'}]$ where the
 graph $\tilde{E'}$ is a renaming of the graph $E'$.
 By proceeding the execution, we obtain the graph $R'$, which consists
 of the graph $\phi(G) @ \oc F$ and $\oc$-boxes connected to some
 outputs of $\phi(G)$, each of which contains the graph $\tilde{E'}$,
 a dot-product node and a vector node.
 Since $\widehat{G @ \oc F} \in P_{V_b \to T_2}$, the forward reasoning
 (Prop.~\ref{prop:FWBWReasoning}) ensures $R' \in P_{T_2}$.
 Moreover, the graph $R'$ is in fact a renaming of the graph
 $R @ \oc F$, therefore we have $R @ \oc F \in P_{T_2}$.
\end{proof}

Finally the soundness theorem, stated below, is obtained as a
consequence of the so-called fundamental lemma of logical predicates.
\begin{theorem}[Soundness]
 For any closed program $t$ such that 
 $- \mid - \mid \vec{p} \vdash t : T$, there exist a graph $G$ and an
 element $X$ of a computation stack such that:
 \begin{align*}
  \Init((- \mid - \mid \vec{p} \vdash t : T)^\ddag) \to^* \Final(G,X).
 \end{align*}
\end{theorem}
\begin{proof}
 This is a corollary of Prop~\ref{prop:FundamentalLemma} below.
\end{proof}
\begin{proposition}[Fundamental lemma]
 \label{prop:FundamentalLemma}
 For any derivable judgement
 $A \mid \Gamma \mid \vec{p} \vdash t : T$, where
 $\Gamma = {x_0:T_0},\ldots,{x_{m-1}:T_{m-1}}$,
 and any graphs $\vec{H} = H_0,\ldots,H_{m-1}$ such that
 $H_i \in P_{T_i}$,
 if the following graph $G$ meets the name criteria and the graph
 criterion, it belongs to $P_T$.
 \begin{align*}
  G =
  \includegraphics{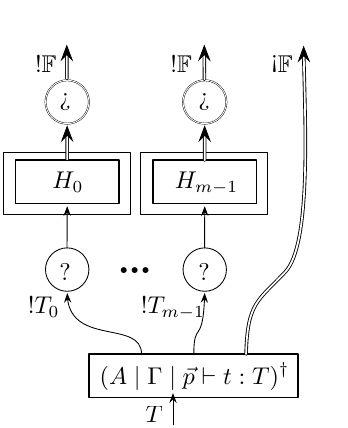}
 \end{align*}
\end{proposition}
\begin{proof}[Sketch of proof]
 The first observation is that the translation
 $(A \mid \Gamma \mid \vec{p} \vdash t : T)^\dag$ itself meets the
 name criterion and the graph criteria.
 Since $\vec{H} \in P_\Gamma$, the whole graph again meets the graph
 criteria.
 We can always make the whole graph meet the name criteria as well,
 by renaming the graphs $\vec{H}$.
 Note that some names in the translation
 $(A \mid \Gamma \mid \vec{p} \vdash t : T)^\dag$ are not bound or
 free, and turns free once we connect the graphs $\vec{H}$.

 The proof is by induction on a type derivation, that goes in a
 simliar way to a usual proof for the lambda-calculus.
 To prove $G \in P_T$, we look at an execution on the graph $G$
 using the backward reasoning (Prop.~\ref{prop:FWBWReasoning}) and
 the congruence property (Cor.~\ref{prop:TerminationCongruence}).
 The only unconventional cases are: an iterated vector operation
 $t \mathrel{\$^1} u$, whose proof is by induction on the number of
 bases introduced in unfolding the operation; and decoupling
 $\abd{a}{f}{x}{T'}{t}$, whose proof relies on
 Prop.~\ref{prop:AbductionSafety}.
\end{proof}

\end{document}